\documentclass[11pt]{article}

\usepackage{authblk}

\usepackage[pdftex]{graphicx}
\usepackage[ddmmyyyy]{datetime}

\usepackage{amssymb,amsfonts,amsmath}
\usepackage{amsthm,xypic,enumerate}
\usepackage{color,subfig,tikz,pgf}
\usepackage{fancyhdr}
\input xy
\xyoption{all}

\newcommand{\R}{\mathbb{R}}
\newcommand{\oF}{\overline{F}}
\newcommand{\oS}{\overline{S}}
\newcommand{\oE}{\overline{E}}
\newcommand{\oP}{\overline{P}}
\newcommand{\oA}{\overline{A}}
\DeclareMathOperator*{\sign}{sign}
\newcommand{\mmC}{\mathcal{C}}

\newtheorem{proposition}[]{Proposition}

\newtheorem{result}[]{Result}

\theoremstyle{definition}

\newtheorem{obs}[]{Remark}

\makeatletter 
\renewcommand\@biblabel[1]{#1.} 
\makeatother 

\begin{document}

\title{Enzyme sharing as a cause of multistationarity in signaling systems}

\author[1]{Elisenda Feliu}
\author[1]{Carsten Wiuf\thanks{wiuf@birc.au.dk}}

\affil[1]{Bioinformatics Research Centre, Aarhus University, C. F. M\o llers All\'{e} 8, DK-8000 Aarhus, Denmark}

\maketitle

\begin{abstract} Multistationarity in biological systems is a mechanism of cellular decision making. In particular,  signaling pathways regulated by protein phosphorylation display features that facilitate a variety of responses to different biological inputs. The features that lead to  multistationarity are of particular interest to determine as well as the stability properties of the steady states.

In this  paper we determine conditions for the emergence of  multistationarity in small  motifs  without feedback that repeatedly occur in signaling pathways. 
We derive an explicit mathematical relationship $\varphi$ between the concentration of a chemical species at steady state and a conserved quantity of the system such as the total amount of substrate available.  
We show that $\varphi$ determines the number of steady states and provides a necessary condition for  a steady state to be stable,  that is, to be biologically attainable. Further, we identify characteristics of the motifs that lead to multistationa\-ri\-ty, and extend the view that multistationarity in signaling pathways arises from multisite phosphorylation.

Our approach relies on mass-action kinetics and the conclusions are drawn in full generality without resorting to simulations or random generation of parameters. The approach is extensible to other systems.

\medskip
\textbf{Keywords}: steady state; kinase; stability; cross-talk; phosphorylation 
 \end{abstract}

\section{Introduction}
Multistationarity  (the existence of more than one  steady state under particular biological conditions)  in cellular systems can be seen  as a mechanism for cellular decision making.  How it arises is therefore fundamental to the understanding of cell signaling, that is, the communication of signals to regulate  cellular activities and responses. Generally, cell signaling  involves  post-translational modifications of proteins, such as phosphorylation, acetylation, or methylation. These modifications  change the state of a protein in a discrete manner, for example from an active to an inactive state.  

In eukaryotes,  reverse phosphorylation is the most frequent  form  of protein modification affecting $\sim\!30\%$ of all  proteins in humans \cite{cohen}. 
{\it Kinases} catalyze the transfer of phosphate groups to target proteins and {\it phosphatases} catalyze the reverse operation.
After the completion of the human genome project, genome analysis estimated the number of  kinases to $\sim\!500$ \cite{manning-kinome}, while the number of phosphatases is smaller by two thirds \cite{cohen}.
Two  protein phosphatases, PP-1 and PP-2A, account for the vast majority of all phosphatase activity
 \cite{xiao-li} with more than 50 PP-1 targets  being characterized  \cite{Cohen:2002p271}. 

As a consequence, there is a substantial complexity in the interplay between enzymes (kinases and phosphatases) and substrates, exemplified by systems where protein substrates use the same catalyzing enzymes (enzyme sharing),  and systems where different enzymes catalyze  the same reaction (enzyme competition). Competition and sharing are general examples of cross-talk between motifs.

The aim of this work is to determine  characteristics that lead to multistationarity. Following different modeling strategies it has already been shown that feedback in signaling networks as well as multisite phosphorylation can both account for multistationarity \cite{Kapuy:2009p16,Markevich-mapk,TG-Nature}.

We present a mathematical approach for analyzing the steady states of small systems.  Our method leads to explicit conditions for when multistationarity occurs in terms of rate constants and conserved total amounts of substrates and enzymes. Further  the approach provides means to study the stability of steady states. 

First we present the motifs that we analyze and then we develop the method  to determine  multistationarity and to study stability. The paper concludes with some perspectives and discussion.

\begin{table}[b!]
\begin{tabular}{lp{13cm}}
\hline
Motifs & Biological phenomena \\ \hline
(b) & A kinase acting also as phosphatase on the same substrate, e.g. HPrK/P kinase-phosphatase in Gram positive bacteria \cite{chaptal-Hpr}.  \\
(c),(d) & Several kinases and/or phosphatases acting on the same substrate, e.g.  (i) Several kinases  phosphorylate the alpha subunit eukaryotic initiation factor (eIF2$\alpha$) at Ser51 \cite{haro-santoyo}; (ii) The  phosphatases MPK-1 and  PTP-SL both modify ERK1 \cite{Shaul:2007p309}.
\\
 (e) & Multi-site phosphorylation  by different kinases and phosphatases at each site, e.g.  (i) Primed kinases, e.g.~GSK-3 \cite{Cohen:2000p313}; (ii) Akt1 is (de)activated through three-site sequential (de)phosphorylation by three  different kinases (phosphatases) \cite{xiao-li}.  \\
 (f),(g) & Multi-site phosphorylation with the same kinase and/or phosphatase responsible for all modifications, e.g. (i) Two-site phosphorylation of ERK catalyzed by MEK; (ii) Dephosphorylation of ERK2 catalyzed by DUSP6 \cite{Ferrell:1997p318}. \\
 (h),(i) & The same enzyme catalyzing the modification of two different substrates, e.g.  the kinases ERK1, ERK2 and the kinase products of the p38 pathway catalyze phosphorylation of two substrates (the mitogen- and stress-activated protein kinase (MSK) 1/2 and the MAP kinase signal-integrating kinase (MNK) 1/2) \cite{Keshet:2010p272}.  \\
 (j),(k),(l) & Cascades with several modification steps and substrates, e.g. 
 (i) MAPK cascades; (ii) Protein kinase A (PKA) phosphorylates  phosphorylase kinase, which in turn phosphorylates glycogen phosphorylase (with dephosphorylation carried out by the same phosphatase, PP-1,  in the two different layers)  \cite[Fig. 7.17]{Fell-book},\cite{PCohen}.\\ \hline
\end{tabular}
\caption{Cellular systems represented by  Motifs (a)-(l) }\label{biotable}
\end{table}

\section{Motifs}
\subsection{Description}\label{motif-description}
We analyze the motifs shown in Fig.~\ref{motifs}. The motifs are referred to as Motif (a)-(l) and  provide simple abstract representations of  known cellular systems.  Some examples motivating our choice of motifs are given in Table \ref{biotable}.
A rich source of examples is  found in the well-studied MAPK cascades. 

To understand how multistationarity relates to enzyme usage we  base our investigation on a motif that does not show multistationarity itself. Therefore we build  the motifs from a one-site phosphorylation cycle which is mono\-stable  \cite{Gold-Kosh-81,Gold-Kosh-84,Bluthgen-sequestration,salazarN1} and shown in  Motif (a). A specific kinase (phosphatase) catalyzes phosphorylation (dephosphorylation) and all mo\-di\-fications can be reversed.    In general, protein phosphoforms are denoted by $S$ and $P$, Fig.~\ref{motifs}. If one phosphoform is converted into another, an arrow is drawn and  the enzyme ($E$ or $F$) catalyzing the reaction is indicated.
\begin{figure}[b!]
\centering
\includegraphics{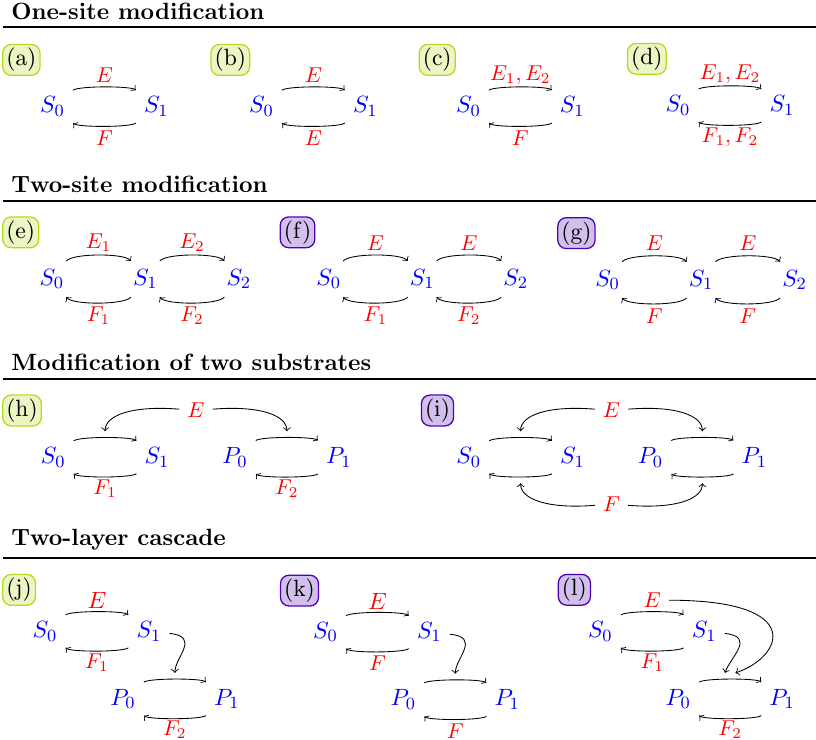}
\caption{\small Motifs composed of one or two one-site cycles. Motifs with purple label, and only these, admit multiple biologically meaningful steady states. $S_i$ and $P_i$ are substrates with $i=0,1,2$  phosphorylated sites. $E, E_1, E_2$ denote kinases, and $F, F_1, F_2$  phosphatases. In Motif (b) the  kinase and the phosphatase are the same enzyme.
}\label{motifs}
\end{figure}

Motifs (a)-(d) cover different possibilities for a one-site modification process. In Motif (b), the same enzyme catalyzes  phosphorylation and dephosphorylation. Motifs (c) and (d) account for  competition between  kinases and/or phosphatases to catalyze the same modification(s).

In eukaryotes, phosphorylation of most proteins takes place in more than one site \cite{olsen-phosphoproteome}, potentially with different bio\-logical effects \cite{wu-selective}. Combination of two one-site cycles into a two-site sequential  cycle   yields three motifs: (e) all enzymes are different, (f)  only one kinase but two phosphatases, and (g)  one kinase and one phosphatase.  By symmetry, Motif (f)  represents as well a motif with one phosphatase but two kinases. We assume for simplicity that both phosphorylation and dephosphorylation proceed in a sequential and distributive manner \cite{salazar_review}; that is, one site is (de)phosphorylated at a time in a specific order.

Motif (h) represents one-site modification of two substrates that share the same kinase but  use different phosphatases. This motif represents  by symmetry also a system with shared phosphatase.  If  both the kinase and the phosphatase are shared, we obtain Motif (i).

Finally, two one-site modification cycles can be combined in a cascade motif, where the activated substrate of the first cycle acts as the kinase of the next. The interplay between enzymes is represented by three cascades: (j)  dephosphorylation at each layer uses different phosphatases, (k) the phosphatase is  not layer specific, and (l) the  kinase of the first layer ca\-ta\-lyzes the modification in the second layer as well.

\subsection{Mathematical modeling}

We assume that any modification $S \rightarrow S^*$  follows the classical \emph{Michaelis-Menten} mechanism  in which an intermediate complex $Z$ is formed reversibly but dissociates into product and enzyme  $G$  irreversibly:

\centerline{\xymatrix{
S + G \ar@<0.5ex>[r]^(0.6){a} & Z  \ar@<0.5ex>[l]^(0.4){b} \ar[r]^(0.4){c} & S^* + G  
}}

The phosphate donor, generally ATP, is assumed to be in large constant concentration   and hence embedded into the rate constants. Imposing mass-action kinetics, the species concentrations over time  can be modeled by a system of polynomial differential equations. For example, in Motif (a) the  equations are (here $E$  also refers to the concentration of the kinase $E$, and similarly for the other species):
\begin{align*}
\dot{E} &= (b^E + c^E)X- a^E E S_{0} &\dot{X} &=  -(b^E + c^E)X + a^E E S_{0}  \\
\dot{F} &=(b^F + c^F)Y -a^F F S_{1} & \dot{Y} &= -(b^F + c^F)Y + a^F F S_1   \\
\dot{S_{0}} &=  b^{E} X +  c^F Y  - a^{E} E S_{0} & \dot{S_{1}} &=  c^E X +b^{F} Y - a^{F} F S_{1},
\end{align*}

\noindent where  $X$ $(Y)$ is the intermediate complex formed by the enzyme $E$ $(F)$ and the substrate $S_0$ $(S_1)$, and $\dot{x}$ denotes differentiation of $x=x(t)$ with respect to time.   For all motifs, there are conservation laws which define time-conserved quantities (total amounts), e.g. $\dot{E}+\dot{X}=0$ and so  $\oE=E+X$ is conserved. The total amounts are fixed  by the initial  concentrations  and determine the state space of the dynamical system. Motif (a) has three conserved  total amounts, namely $\oF=F+Y$ and $\oS=S_0+S_1+X+Y$ in addition to that of $\oE$.

The steady states of the system are  solutions (potentially negative) to the polynomial equations obtained by setting all derivatives  to zero with the constraints imposed by the conservation laws, once total amounts have been fixed. These laws imply that some steady state equations are redundant, 
e.g. either $\dot{E}=0$ or $\dot{X}=0$ can be disregarded. 
We focus on  the \emph{biologically meaningful steady states} (BMSSs), that is, the steady states for which all concentrations are non-negative (positive or zero). If at least two BMSSs exist for fixed total amounts, then  the system is said to be multistationary.

The specific form of the chemical reactions for  Motifs (a)-(l) together with the corresponding systems of differential equations are described in the Supplementary Material (SM), attached at the end of the main text.

\section{The Steady State Function $\varphi$}
In this section we outline the procedure used to analyze the  motifs. Details of the mathematical analysis are in the SM.

The system of equations describing the steady states can be reduced  substantially by elimination of variables \cite{TG-Nature,Feliu:2010p94}. For the  motifs considered here, elimination of variables implies that the steady states are characterized by a relation $\oS = \varphi(Y)$ between the concentration of one of the species, typically an intermediate complex $Y$, and the total amount of a substrate $\oS$. The concentrations   of the other species are given in terms of $Y$, usually as ratios of polynomials in $Y$.
By imposing all concentrations to be non-negative, $Y$ is restricted to a set $\Gamma$ of possible values. 
Further, for any $\oS\geq 0$, there is {\it at least one} BMSS, that is, $\oS=\varphi(Y)$ for some $Y$ in $\Gamma$.
The function  $\varphi$ is continuous  and differentiable in $\Gamma$ and depends on the rate constants and the total amounts, except for $\oS$.

The number of BMSSs can be found from the  ana\-ly\-sis of $\varphi$. If $\varphi$ is strictly increasing or decreasing in $\Gamma$,  $\varphi$ is one-to-one and hence, for a given total amount $\oS$, there is a corresponding unique $Y$ at steady state. Consequently, multistationarity cannot occur (Fig. \ref{summary}a).

Figs. \ref{summary}b and \ref{summary}c show situations where multistationarity occurs. If $\varphi$ has increasing and decreasing parts, or if $\Gamma$ is not connected, then   $Y_1\not=Y_2$ with $\varphi(Y_1)=\varphi(Y_2)=\oS$ might exist. Hence, there are at least two  BMSSs with the same  $\oS$.

These two figures represent substantially different switch responses. In Fig. \ref{summary}c, there is only one BMSS for low $\oS$. An increase of $\oS$ to $\oS_{max}$ causes the system to switch to  a `high' steady state (high $Y$)   under the assumption that the green steady states in the figure are stable. If $\oS$ is decreased again to $\oS_{min}$, then  the system switches back  to a `low' steady state. In Fig. \ref{summary}b,  there is  one BMSS for low $\oS$. An increase of $\oS$  keeps the system in the first branch of $\varphi$ and thus it will behave as a monostationary system.  

\begin{figure}[b!]
\includegraphics{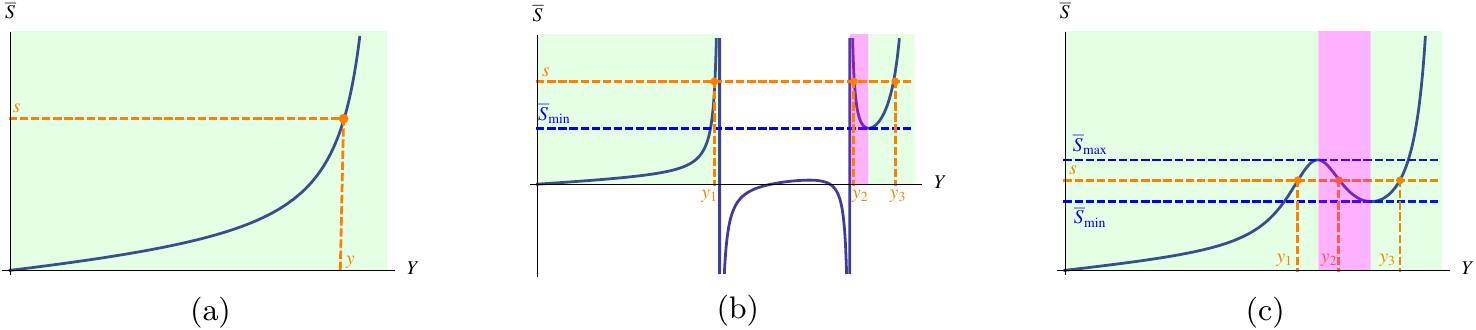}
\caption{\small  Possible shapes of $\varphi$ in $\Gamma$ (colored regions: magenta=unstable BMSSs; green=(possible) stable BMSSs).
(a) $\varphi$ is increasing and for any $s$, there is one BMSS ($y$) such that $\varphi(y)=s$. (b) $\Gamma$ consists of two disconnected regions. For $s<\oS_{\rm min}$ there is one BMSS; for $s=\oS_{\rm min}$ there are precisely two; and for $s>\oS_{\rm min}$ there are three;  $\varphi$ is also defined  in the white region but some concentrations become negative. (c) $\varphi$ is in part decreasing, in part increasing. For $\oS_{\rm min}<s<\oS_{\rm max}$ there are three BMSSs; for $s=\oS_{\rm min}$ or $s=\oS_{\rm max}$, there are two; and for $s<\oS_{\rm min}$ or $s>\oS_{\rm max}$,  there is  one. }\label{summary}
\end{figure}

Interestingly,  the derivative $\varphi'(Y)$ of $\varphi(Y)$  provides means to determine whether some steady states are unstable. Unstable steady states are unattainable under biological conditions. Specifically, we find that either the regions in which $\varphi$ is increasing or those in which it is decreasing must correspond to unstable steady states, see Section \ref{stability}.

 In summary, the function $\varphi$ determines whether multiple BMSSs exist and encodes information about the stability of steady states.  
In the next section we analyze  $\varphi$ for Motifs (e) and (f).  We show how enzyme sharing in a two-site cycle (f)  leads to multistationa\-ri\-ty, as opposed to a two-site cycle with different enzymes (e). A detailed analysis of all motifs is given in the SM.

\section{Mono- versus Multistationarity}

\subsection{Monostationarity}

Motifs (a)-(e), (h), and (j) have exactly one BMSS for any choice of rate constants and total amounts. In all cases, the function $\varphi$ is increasing in $\Gamma$. The procedure is very similar in all cases and  is thus only illustrated for Motif (e).  We take some effort in explaining the details as the procedure might have general applicability.

Motif (e) consists of three phosphoforms of the substrate, $S_0,S_1,S_2$, with subscript indicating the number of phosphorylated sites.  The chemical reactions of the system are:

\noindent
\centerline{\small
\xymatrix@R=11pt@C=16pt{
S_{0} + E_1 \ar@<0.4ex>[r]^(0.6){a_{1,E}} & X_1  \ar@<0.4ex>[l]^(0.4){b_{1,E}} \ar[r]^(0.4){c_{1,E}} & S_{1} + E_1 & 
S_{1} + E_2 \ar@<0.4ex>[r]^(0.6){a_{2,E}} & X_2  \ar@<0.4ex>[l]^(0.4){b_{2,E}} \ar[r]^(0.4){c_{2,E}} & S_{2} + E_2 \\
S_{1} + F_1 \ar@<0.4ex>[r]^(0.6){a_{1,F}} & Y_1  \ar@<0.4ex>[l]^(0.4){b_{1,F}} \ar[r]^(0.4){c_{1,F}} & S_{0} + F_1 &
S_{2} + F_2 \ar@<0.4ex>[r]^(0.6){a_{2,F}} & Y_2  \ar@<0.4ex>[l]^(0.4){b_{2,F}} \ar[r]^(0.4){c_{2,F}} & S_{1} + F_2
}}
We denote the inverse of the Michaelis-Menten constants of $E_i$  by $\kappa_{i,E}=a_{i,E}/(b_{i,E}+ c_{i,E})$ and of $F_i$ by $\kappa_{i,F}=a_{i,F}/(b_{i,F}+ c_{i,F})$.  The ratio of the catalytic constants of phosphatase and kinase is denoted by $\mu_i=c_{i,F}/c_{i,E}$.

The system has five conserved total amounts, which are assumed to be positive: Four for enzymes, $\overline{E}_i = E_{i} + X_{i}$ and $\overline{F}_i = F_{i} + Y_i$ ($i=1,2$),  and one for substrate, $ \overline{S}  =   S_{0}+S_1+S_2 + X_{1}+X_2+Y_1+Y_2$.  The steady state equations can be rewritten as
\begin{align}\label{twocycle_dif} 
X_1 &= \kappa_{1,E} E_1 S_{0} &\quad Y_1&=\kappa_{1,F} F_1 S_1 &\quad X_1 &= \mu_1 Y_1 \\
X_2 &= \kappa_{2,E} E_2 S_{1} &\quad Y_2&=\kappa_{2,F} F_2 S_2 &\quad X_2 &= \mu_2 Y_2. \nonumber
\end{align}
The last column gives $X_i$ in terms of $Y_i$. The total amounts $\oE_i,\oF_i$ give $E_i,F_i$ in terms of $Y_i$ as well: $E_i=\oE_i-\mu_iY_i$, $F_i=\oF_i-Y_i$. 
Further, if $\oE_i,\oF_i>0$ then $E_i=0$ or $F_i=0$ cannot be solutions to Eq.~\eqref{twocycle_dif}. 
It follows that  the concentrations $E_i,F_i$ are positive if and only if $Y_i$ is in $\Gamma_i=[0,\xi_i)$ with
$\xi_i=\min(\oF_i,\oE_i/\mu_i)$.

We further isolate $S_0,S_1$ from the first row in Eq.~\eqref{twocycle_dif} and $S_2$ from the second and obtain
\begin{equation}\label{sy}
S_0 = \frac{\mu_1 Y_1}{\kappa_{1,E}(\oE_1-\mu_1Y_1) }, \quad S_i= \frac{Y_i}{\kappa_{i,F}(\oF_i-Y_i) }
\end{equation}
for $i=1,2$. Then, $S_0,S_1$ (resp. $S_2$) are non-negative  increasing  continuous functions of $Y_1$ in $\Gamma_1$ (resp. $Y_2$ in $\Gamma_2$).
The remaining equation, $X_2=\kappa_{2,E} E_2 S_{1}$, gives $Y_2$ in terms of $Y_1$:
\begin{equation}\label{f1}
Y_2 =f(Y_1)=\frac{\kappa_{2,E} \oE_2 Y_1}{\mu_2(\kappa_{1,F}(\oF_1 - Y_1) + \kappa_{2,E} Y_1)}. \end{equation}
The function $f$ is non-negative increasing and continuous in $\Gamma_1$. Further, for $Y_2$ to be in $\Gamma_2$, 
it is required that $Y_1$ is  in $\Gamma=[0,\xi)\subseteq\Gamma_1$ with $\xi=\min(\xi_1, f^{-1}(\xi_2))$.

Finally, using Eq.~\eqref{f1}, we find that  $X_2$ and $S_2$ are increasing functions of $Y_1$  in $\Gamma$. Therefore, using the formulas above,  all  concentrations at steady state  are non-negative if and only if  $Y_1$ is in $\Gamma$.  We conclude that  the BMSSs of the system satisfy  
$$\oS =S_{0}+S_1+S_2 + X_{1}+X_2+Y_1+Y_2 = \varphi(Y_1)$$
for $Y_1$ in $\Gamma$. Since $\varphi$ is a sum of  increasing  continuous functions in $Y_1$, then so is $\varphi$. Additionally, $\varphi(0)=0$ and $\varphi(Y_1)$ tends to infinity as $Y_1$ tends to $\xi$. Thus, $\varphi$ has  the form in Fig. \ref{summary}a with a unique $Y_1$ for any given $\oS$, that is, there is one BMSS.

\subsection{Multistationarity}
We consider a two-site phosphorylation system with one kinase but different phosphatases for each phosphoform, as shown in Motif (f). Multistationarity has been observed numerically in this system  \cite{Markevich-mapk}. The system reduces to Motif (e)  by setting $E_1=E_2$ and  we use the notation introduced previously. The conservation laws are the same with the exception that  there is only one kinase law, $\oE=E+X_1+X_2$. Define $\xi_i=\min(\oF_i,\oE/\mu_i)$ and $\Gamma_i=[0,\xi_i)$. 

The system of equations to be solved is similar to  Eq.~\eqref{twocycle_dif} with $E=E_i$. Thus, we start by writing $X_i,E,F_i$ as  functions of $Y_1,Y_2$. Since $E,F_i$ must be positive at any BMSS, we require $0\leq Y_i<\oF_i$ and $\mu_1Y_1+\mu_2Y_2<\oE$.
For these values we obtain
$$ S_0 = \frac{\mu_1 Y_1}{\kappa_{1,E}(\oE-\mu_1Y_1-\mu_2Y_2) }, \quad S_i= \frac{Y_i}{\kappa_{i,F}(\oF_i-Y_i) }$$
for $i=1,2$, which are non-negative increasing continuous functions of $Y_i$. Using $X_2 = \kappa_{2,E} E S_{1}$, we obtain  $Y_2$ as a non-negative continuous function of $Y_1$ in $\Gamma_1$:
$$Y_2=f(Y_1)=\frac{\kappa_{2,E}(\oE -\mu_1 Y_1) Y_1}{\mu_2(\kappa_{1,F}( \oF_1 -Y_1) + \kappa_{2,E} Y_1)}. $$ 
This function resembles that in Eq.~\eqref{f1} except for the quadratic term in the numerator, which is a consequence of the conservation law for $\oE$ involving both $Y_1$ and $Y_2$.  Further, $f$  might not be increasing for all $Y_1$.

Let  $\Gamma=\{Y_1\in \Gamma_1, \textrm{ such that } f(Y_1)\in \Gamma_2 \}.$ Using the formulas derived above,  all concentrations at steady state  are non-negative if and only if  $Y_1$ is  in $\Gamma$. Hence, for any BMSS,
$$\oS =S_{0}+S_1+S_2 + X_{1}+X_2+Y_1+Y_2 = \varphi(Y_1)$$
with $Y_1$ in $\Gamma$. The function $\varphi$ is continuous with $\varphi(0)=0$ but $\Gamma$ might not be a connected interval.

Define $\Lambda =(1+\kappa_{2,E}/\kappa_{1,F})\mu_1 \oF_1 - \oE $. If $\Lambda \leq 0$, 
then $f$ is an increasing function in $\Gamma$ and we conclude that there is exactly one BMSS. 
If  $\Lambda> 0$,  then  $f$ has a unique local maximum for some  $\alpha$ in $\Gamma_1$ and all cases  in Fig. \ref{summary} can occur.   By varying the value of $\oF_2$ while keeping the other constants fixed, we obtain (Fig. \ref{motiff}):
\begin{itemize}
\item[(i)]  $\oF_2\leq (\oE-\mu_1\oF_1)/\mu_2$ (orange): $\Gamma=[0,\alpha_1)$ with $f(\alpha_1)=\oF_2$. The function $f$, and thus $\varphi$, are increasing  
and  there is  one BMSS  (Fig.~\ref{summary}a). 
\item[(ii)] $ (\oE-\mu_1\oF_1)/\mu_2 < \oF_2\leq f(\alpha)$ (green): $\Gamma= [0,\alpha_1)\cup (\alpha_2,\xi_1)$ with $\alpha_1\leq \alpha\leq\alpha_2$ and $f(\alpha_1)=f(\alpha_2)=\oF_2$. Hence,  $f$ is increasing in $[0,\alpha_1)$, decreasing in $(\alpha_2,\xi_1)$ and multistationarity occurs (Fig. \ref{summary}b).
\end{itemize}
When $f(\alpha)< \oF_2$ there is an $M$ such that: 
\begin{itemize}
\item[(iii)] $f(\alpha)<\oF_2< M$ (purple): $\Gamma=[0,\xi_1)$. The function $\varphi$ has a decreasing part and multistationarity occurs (Fig.~\ref{summary}c).
\item[(iv)] $M\leq \oF_2$ (blue):  $\Gamma=[0,\xi_1)$. The function $\varphi$ is  increasing  and there is  one BMSS  (Fig.~\ref{summary}a). 
\end{itemize}

\begin{figure}[t!]
\centering
\includegraphics[scale=1]{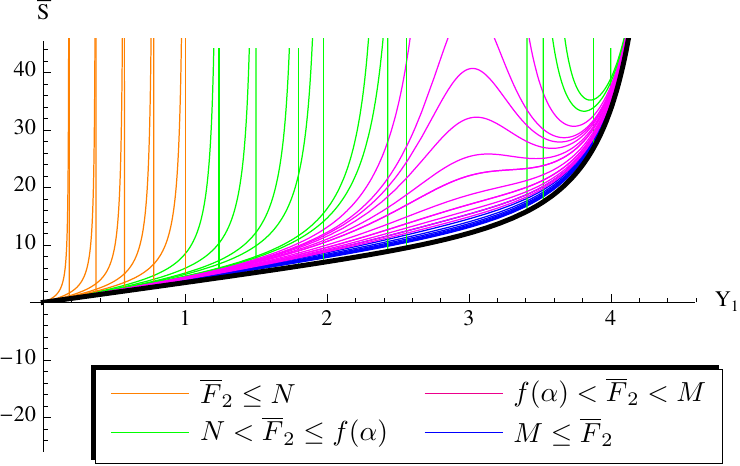}
\caption{\small The function $\varphi$ for Motif (f)  for different values of $\oF_2$ and fixed $\oE>\mu_1\oF_1$ and $\Lambda>0$. Let $N=(\oE-\mu_1\oF_1)/\mu_2$. The values $M$ and $\alpha$ depend on $\oE$, $\oF_1$ and the rate constants (see text). The set $\Gamma$ is disconnected only when $N<\oF_2\leq f(\alpha)$. For large $\oF_2$, $\varphi$ approaches the black line. The vertical bars mark the boundary of $\Gamma$.}\label{motiff}
\end{figure}

\begin{figure}[b!]
\begin{center}
\includegraphics[scale=1]{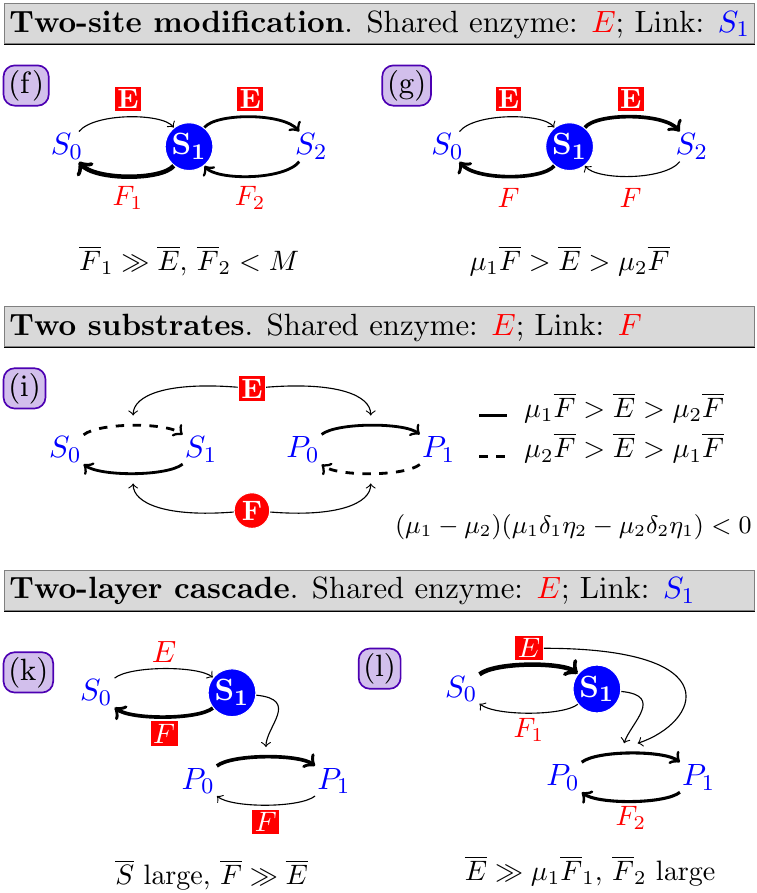}
\end{center}
\caption{\small Conditions for multistationarity are given. The shared enzyme is marked with a colored square; the link is marked with a colored circle; predominant modifications are marked in bold.  The symbol $\gg$ is short for `much larger'. }\label{multitable}
\end{figure}

\subsection{Understanding Multistationarity}
Motifs (f), (g), (i), (k) and (l) exhibit multistationarity for some choices of total amounts and  rate constants  (Fig.~\ref{multitable}).  
The regions for which multistationarity occurs are detailed in the SM. 
In Motifs (i), (k) and (l) it appears only as in Fig.~\ref{summary}c, while in Motifs (f) and (g) both the forms in Fig.~\ref{summary}b and Fig.~\ref{summary}c occur. 

 It is remarkable  that  in Motifs (f), (k) and (l) multistationarity occurs for any set of rate constants and depends only on the initial conditions (that is,  the total amounts). Thus, multistationarity can occur in these systems independently of the specific kinetics. In contrast,  multistationarity in Motif (i) depends on the rate constants and hence not all  kinetics exhibit multistationarity. The same  appears to be  the case for Motif (g)  \cite{conradi-mapk,WangSontag}.
 
The common characteristic of these motifs is that a single enzyme is responsible for catalyzing two different substrate modifications, which  at the same time are \emph{linked} (Fig.~\ref{multitable}). Indeed, in Motifs (f) and (g) the substrates are linked through $S_1$, which is a modified as well as an unmodified substrate for the shared enzyme $E$. For the Motifs (k) and (l), the link is given by $S_1$ which is a modified substrate and a kinase, and the common enzymes are $F$ and $E$, respectively. In Motif (i) the kinase $E$ is common and the phosphatase $F$ provides the link (or {\it vice versa}). In contrast, in Motif (h)  an enzyme is responsible for two different modifications, but there is no link between the two substrates. Consequently,  multistationarity cannot be observed.

Multistationarity  can arise from two opposing dynamics acting on the same substrate (Fig.~\ref{multitable}). For example,  in Motif (f), if $\oF_1$ is much bigger than $\oE$ and $\oF_2<M$, then there are multiple BMSSs. Thus, since the amount of phosphatase in the first cycle is much larger than the amount of kinase, $S_1$ is pushed towards the unmodified form $S_0$, while in the second cycle  $S_1$ is driven towards the fully modified form $S_2$ (because $\oF_2<M$). 

In Motif (i), provided the conditions on the parameters are fulfilled  (Fig.~\ref{multitable}), multistationarity occurs if either $\mu_1\oF>\oE  >\mu_2 \oF$ or  $\mu_2\oF<\oE<\mu_1 \oF$. It implies that in one cycle the phosphatase `wins', while in the other the kinase does.

\section{Stability Analysis}\label{stability}

BMSSs are defined as steady states for which all concentrations are non-negative. However, a steady state is biologically attainable only if it is (asymptotically) stable, that is, nearby trajectories are attracted to it. We show here for our motifs that if  $\varphi'(Y)<0$ for some steady state $\varphi(Y)=\oS$, then it is unstable.

\subsection{The Jacobian and variable elimination}
For a system of ordinary differential equations in $\R^m$, a steady state $z$ is asymptotically stable if all eigenvalues of the Jacobian evaluated at $z$ have negative real parts \cite[Thm.~1.1.1]{Wiggins}. 
Since the Jacobian is a real matrix,   the complex eigenvalues come in pairs of conjugates  and their product is a positive number. If $m$ is odd and all eigenvalues have negative real parts, their product, and hence the determinant of the Jacobian, must be negative. If $m$ is even and $z$ stable, then the product of the eigenvalues must be positive.  Thus,  the sign of the determinant of the Jacobian  provides a necessary condition for a steady state to be stable and a sufficient condition for it to be unstable.

For $x=(x_1,\dots,x_n)$ let $x^{(j)}=(x_1,\dots,\widehat{x_{j}},\dots,x_n)$ ($x$ with $x_j$ removed).
We make the following observation (see SM for a proof): Let $f=(f_1,\dots,f_n)\colon\Omega\subseteq \R^n\rightarrow \R^n$ be a differentiable function defined on an open set $\Omega$ and $z$ such that $f(z)=0$. Assume that $x_j$ can be eliminated from the equation $f_i=0$ in a neighborhood around $z$; that is, there exists a differentiable function $\psi\colon\Omega^{(j)}\subseteq \R^{n-1}\rightarrow \R$, $z^{(j)}$ in  $\Omega^{(j)}=\{x^{(j)}|x\in\Omega\}$, such that $x_j=\psi(x^{(j)})$ if $f_i(x)=0$.
Define $\bar{f} \colon \Omega^{(j)} \rightarrow \R^{n-1}$ by
$\bar{f}_k(x)=f_k(x_1,\dots,x_{j-1},\psi(x),x_j,\dots,x_{n-1})$ for all $k\neq i$ and let $\bar{J}$ denote the associated Jacobian. Then, the determinant of the Jacobian of $f$ at $z$ satisfies
\begin{equation}\label{jacobian}
 (-1)^{i+j}\frac{\partial f_i}{\partial x_j}(z)\ \det(\bar{J}(z^{(j)})) = \det(J(z)). 
 \end{equation}

\subsection{Unstable steady states} 
The relation between the sign of the determinant of the Jacobian and stability, together with Eq. \eqref{jacobian},  lead to a criterion to detect unstable steady states. For each motif, let $x=(x_1,\dots,x_n)$ be the species concentrations, $\dot{x}_i=h_i(x)$ the differential equations and   $\oA_1=g_1(x),\dots,\oA_c=g_c(x)$  the equations for the total amounts. We choose the order of the species such that $x_i$, $i=1,\dots,c$,  can be isolated from $\oA_i=g_i(x)$ and the steady state equation $\dot{x}_i=0$ becomes redundant.  For fixed total amounts, $\oA_1,\dots,\oA_c$, the steady states are the solutions to the system  $f(x)=0$ of $n$ equations in $n$ variables with $f_i(x)=g_i(x)-\oA_i$ for $i=1,\dots,c$ and $f_i(x)=h_i(x)$ for $i=c+1,\dots,n$.

Let $J(z)$ denote the Jacobian of $f$ at $z$. In the SM we prove: If $z$ is a steady state, that is, $f(z)=0$, and either (i) $n-c$ is even and $\det(J(z))<0$ or (ii) $n-c$ is odd and $\det(J(z)) > 0$, then $z$ is unstable. The proof relies on the observation made about the eigenvalues and Eq.~\eqref{jacobian}.

The function $\varphi$ of our motifs is derived through successive elimination of variables precisely from the system of equations $f(x)=0$. Using Eq.~\eqref{jacobian},  the sign of $\det(J(z))$ at a steady state $z$ can be traced back from the sign of the derivative of $\varphi$ (the Jacobian of a system with one equation) by considering  the equation number ($i$), the equation variable ($j$), and the sign of $\partial f_i/\partial x_j$ after each elimination.

To exemplify the  procedure, consider Motif (f), where $n=10$ and $c=4$. The  system is (see SM for details):
\begin{align*}
f_1(x)  &= E + X_1+X_2 - \oE  & f_6(x) & =  X_2 - \kappa_{2,E} E S_{1}  \\ 
f_2(x)  &= F_1 + Y_1 -\oF_1  & f_7(x) &= X_1-\mu_1 Y_1 \\  
f_3(x) &= F_2 + Y_2 -\oF_2  &   f_8(x) &=X_2- \mu_2 Y_2 \\
f_4(x) &=  S_0 + S_1 + S_2+X_1+X_2 + Y_1+Y_2 - \oS  & f_9(x) &=  Y_2- \kappa_{2,F} F_2 S_2   \\  
f_5(x) &=  X_1 - \kappa_{1,E} E S_{0}  & f_{10}(x) &=Y_1- \kappa_{1,F} F_1 S_1   
\end{align*}
with species $x=(E,F_1,F_2,S_0,X_1,X_2,S_1,S_2,Y_2,Y_1).$ 
The function $\varphi$ is  $f_4$ in terms of $Y_2$ after successive eliminations. Let $\epsilon_k=\pm 1$ depending on whether the sign of the determinant of the Jacobian changes ($-$) or not ($+$) after the $k$-th elimination. Then, $\sign(\varphi'(Y_2)) = \left(\prod_k \epsilon_k\right) \sign(\det(J(z)))$, where $z$ is the steady state with $z_9= Y_2$.

The order and sign of the eliminations are shown in Table~\ref{elimination_table}. We find that  $\prod_k \epsilon_k=1$, implying   that the sign of $\varphi'(Y_2)$ agrees with the sign of  the determinant of the Jacobian of $f$ evaluated at the corresponding steady state.
Since $n-c=6$ is even, we conclude that the values of $Y_2$ for which $\varphi$ is decreasing, that is, $\varphi'(Y_2)<0$, correspond to unstable steady states. Further, it follows that unstable points come between other steady states that  presumably are stable. 

\begin{table}[t!]
\renewcommand{\arraystretch}{1.3}
\setlength{\tabcolsep}{3pt}
\centering
\begin{tabular}{cccc||cccc}
\hline
$k$ & Elimination &  Behavior$^{\rm a}$ & $\epsilon_k$$^{\rm b}$ & $k$ & Elimination &  Behavior & $\epsilon_k$\\ \hline
1 & $(f_1,E)$  &  $(1,1,+)$ & +  &6 & $(f_5,S_0)$ & $(2,1,-)$ & +  \\
2 & $(f_2,F_1)$  & $(1,1,+)$ & +  & 7 & $(f_6,S_1)$ &  $(2,1,-)$ & +  \\
3 & $(f_3,F_2)$ &  $(1,1,+)$ & +  & 8 & $(f_9,S_2)$ & $(2,1,-)$ & +  \\ 
4 & $(f_7,X_1)$  &  $(4,2,+)$ & +  & 9 & $(f_{10},Y_1)$ & $(2,1, -)$ & + \\
5 & $(f_8,X_2)$  & $(4,2,+)$ & + \\ \hline
\end{tabular}
\begin{flushleft}
{\small $^{\rm a}$ $(i,j,\sigma)$  indicates that $i,j$ are the indices of the equation and variable iteratively being eliminated and $\sigma$ corresponds to whether $\bar{f}_i$ ($f_i$ after substitution of the previous eliminations) is increasing ($\sigma=+$) or decreasing ($\sigma=-$) as a function of $x_j$.
  \\
$^{\rm b}$ obtained as obtained as $\sigma(-1)^{i+j}$.
}
\end{flushleft}
\caption{Elimination of variables for Motif (f). After each elimination  the system $\bar{f}$ is rewritten  to correctly determine the sign  of $\partial\bar{f}_i/\partial x_j$ before the next elimination}

\label{elimination_table}
\end{table}

\subsection{Stability in monostationarity motifs}
The Routh-Hurwicz criterion \cite{Hurwitz} gives sufficient and necessary conditions for the Jacobian to have all eigenvalues with negative real parts. Thus, the (asymptotic) stability of a steady state can be determined by this criterion.
For the Motifs (a)-(e) and (h) the criterion is fulfilled and the unique BMSS is asymptotically stable. We have not been able to determine this for Motif (j).

\section{Discussion}
We have investigated small motifs without feedback that account for cross-talk, enzyme competition, sharing and specificity in post-translational modification systems and  determined some features that lead to multistationarity in signaling pathways.

Bistability, and generally multistability, in biological systems is seen as a mechanism of cellular decision making. Compared to systems with a single steady state, the presence of multiple stable steady states provides a possible switch between different responses and increased robustness with respect to environmental noise.  Our study has been driven by the observation that biological systems deviate from a one-to-one correspondence between enzymes and the modifications they catalyze. This phenomenon, known as cross-talk and enzyme sharing, can  cause  multistationarity and hence be essential for regulating  signaling systems. 

Our work extends the view of multistationarity as arising from multisite phosphorylation \cite{TG-Nature}  to the view that  multistationarity is driven by a single enzyme that  catalyzes linked substrates.  Two opposing dynamics acting on the same substrate is a recurrent characteristic of multistationarity. These observations await a precise mathematical formulation and an investigation of its generality.

Our approach is conceptually simple and reduces to the study of analytical properties of a function $\varphi$ that relates a conserved total amount and the concentration of a species at steady state. The graph $(\varphi(Y),Y)$ can be seen as a bifurcation diagram with one parameter, $\oS$.
When monostationarity occurs,  analysis of $\varphi$ is quite straightforward, while a more in-depth analysis is required when  multistationarity occurs. An advantage of this approach is that unstable steady states are readily detected from the  form of  $\varphi$.

The existence of $\varphi$ is not guaranteed  in general. For larger systems, a detailed study independently of the rate constants  cannot be pursued. However, we have shown that for cascades of arbitrary length (extentions of Motif (j)), the function $\varphi$ exists and  properties of the full cascade can be derived from properties of the building block, the one-site cycle  \cite{Feliu:2010p94}. We are currently working on extending this approach to other systems.

There are some mathematical characterizations of mono\-stationarity in chemical reaction networks that relate to our work. 
The theory of monotone systems \cite{smith-monotone} characterizes systems in which there is only one BMSS and at the same time gives conditions for global stability of the BMSS.  However, a condition for the theory to be applicable here is that  no species take part in more than two reactions \cite{angelisontag2}. This condition is only fulfilled for Motif (a) and hence the theory cannot be applied here.

The only general theory of applicability is that of injective systems \cite{craciun-feinbergI,craciun-feinbergII,craciun-feinberg-pnas}. The motifs that do not allow multiple steady states are in fact  injective in the sense of \cite{craciun-feinbergI} when modeled as a {\it continuous flow stirred tank reactor}. This fact  implies that {\it at most one} steady state  exists \cite{craciun-feinberg}. The advantage of this theory is that monostationarity is derived for more general kinetics than mass-action \cite{Banaji-donnell}.
However, when restricted to mass-action, our approach is as simple as checking for injectivity and provides in addition simple  rational functions that enable further comprehensive studies of variation in species concentrations, such as stimulus-response curves and signal amplification \cite{Feliu:2010p94,Feliu-proc}.  

\bigskip
\emph{Acknowledgements.}
EF is supported by a postdoctoral grant from the ``Ministerio de Educaci\'on'' of Spain and the project  MTM2009-14163-C02-01 from the ``Ministerio de Ciencia e Innovaci\'on''.  CW is supported by the Lundbeck Foundation, Denmark. Neil Bristow, Freddy Bugge Christiansen and Michael Knudsen are thanked for commenting on the manuscript.


\newpage 
\appendix

\begin{center}
{\LARGE \bf Supplementary Material}
\end{center}

\section{Introduction}

In this Supplementary Material (SM) we provide  details of the analysis of multistationarity given in the main text, as well as proofs of the results mentioned there. We go through all motifs and the reader will easily see that many arguments are repeated as the motifs share common structures. We have tried to keep the analysis of each motif independent of the analyses of the other motifs.  However, the motifs progress from simple one-site modification cycles to more complex motifs and the line of thought transpires most easily from the simple motifs. It is therefore advisable to study the simple motifs before the more complicated motifs. To keep the SM self-contained, Figures 1 and 2 of the main text are reproduced here.

Motif (g) (a futile cycle with two parts) has been studied  extensively in the literature. Motif (j) (a linear cascade with two layers) was studied for arbitrary length in \cite{Feliu:2010p94}, where we showed that the system admits only one steady state. It is briefly covered here for completeness.

\begin{figure}[!b]
\centering
\framebox{
\includegraphics[scale=1.0]{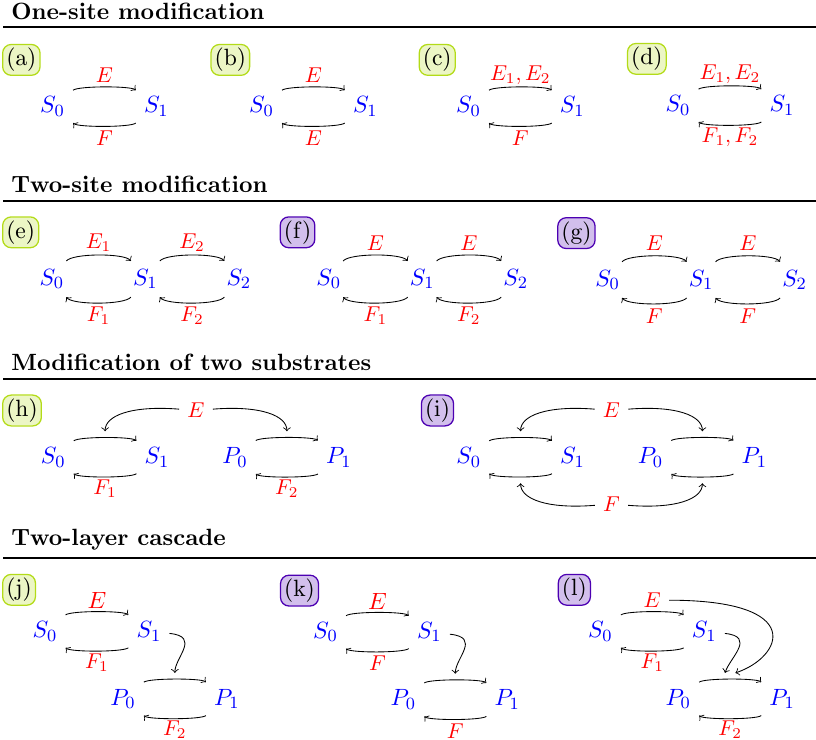}}
\caption{\small Motifs covered in this work. Figure 1 of the main text.
}\label{motifs}
\end{figure}

\begin{figure}[!t]
\includegraphics{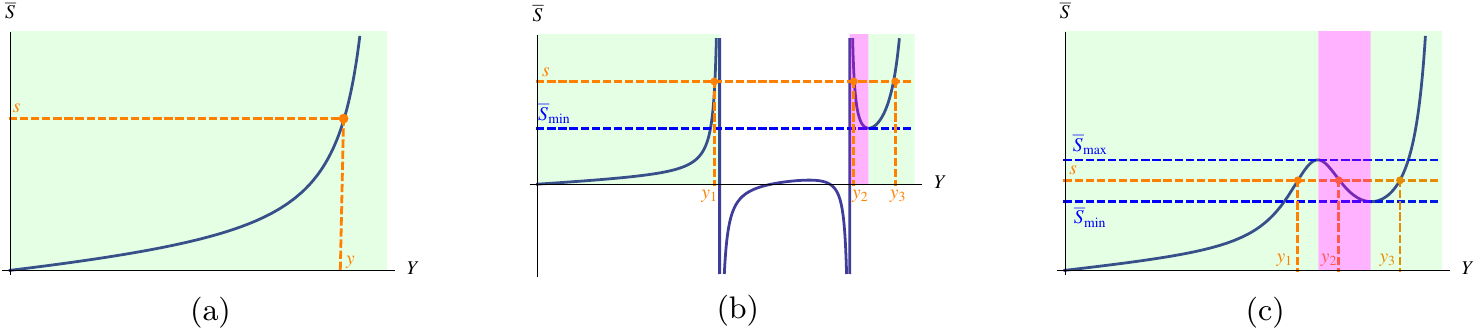}
\caption{  \small  Possible shapes of $\varphi$ in $\Gamma$ (colored regions: magenta=unstable BMSSs; green=(possible) stable BMSSs).
(a) $\varphi$ is increasing and for any $s$, there is one BMSS ($y$) such that $\varphi(y)=s$. (b) $\Gamma$ consists of two disconnected regions. For $s<\oS_{\rm min}$ there is one BMSS; for $s=\oS_{\rm min}$ there are precisely two; and for $s>\oS_{\rm min}$ there are three;  $\varphi$ is also defined  in the white region but some concentrations become negative. (c) $\varphi$ is in part decreasing, in part increasing. For $\oS_{\rm min}<s<\oS_{\rm max}$ there are three BMSSs; for $s=\oS_{\rm min}$ or $s=\oS_{\rm max}$, there are two; and for $s<\oS_{\rm min}$ or $s>\oS_{\rm max}$,  there is  one.  Figure 2 in the main text.  }\label{summary}
\end{figure}

\subsection{Notation and preliminaries}
A continuous and differentiable function with  continuous derivative is said to be $\mathcal{C}^1$. We denote by $\R_{+}$ the set of positive real numbers and by $\overline{\R}_{+}$ the set of non-negative real numbers. A rational function is a quotient of two polynomials.  An increasing (decreasing) function $f$ fulfills $f(x)<f(y)$ ($f(x)>f(y)$) for $x<y$, i.e.~we take increasing (decreasing) to mean strictly increasing (decreasing). The notation $x\in A$ means that $x$ belongs to the set $A$.

We make frequent use of the Implicit Function Theorem in two dimensions to relate two variables to each other and to find derivatives of implicit functions. Details about the Implicit Function Theorem can be found in text books on functional analysis.

For $x=x(t)$ a real function of $t$, we denote by $\dot{x}$ the derivative of $x$ with respect to $t$, $dx/dt$.

\subsection{Constants and variables} 
We consider the rate constants to be fixed and positive, i.e. in $\R_{+}$. The constants are $a_*^*, b_*^*, c_*^*$ and those derived from these:  $\kappa_{*,*}$,  the inverse of the Michaelis-Menten constants  and $\mu_*^*$, the ratios of the catalytic constants of phosphatase and kinase. Here, to ease the presentation, $\eta_*$ denotes $\kappa_{*,E}$ for a kinase $E$   and $\delta_*$ denotes $\kappa_{*,F}$ for a phosphatase $F$.

The total amounts are likewise considered fixed and positive. Species concentrations are considered variables of the system. For example, in $X-\eta ES_0$, $\eta$ is a fixed unknown constant and $X,E$ and $S_0$ are variables that depend on each other, e.g.~$X$ can be considered a function of $E$ and $S_0$, $X=X(E,S_0)$. That is,  the differential equations describing the system is a set of polynomials with coefficients in the ring $\mathbb{R}[a_*^*,b_*^*,c_*^*,\oS_*,\oP_*,\oE_*,\oF_*]$ or $\mathbb{R}[\eta_*,\delta_*,\mu_*,\oS_*,\oP_*,\oE_*,\oF_*]$ (where some constants might not be present in all systems, e.g.~$\oP$ is not in Motif (a)) and variables being a finite list of species concentrations, $E_*,F_*,S_*,P_*,X_*,Y_*$  (where some variables might not be present in all systems, e.g.~$P_0$ is not in Motif (a)). 

\subsection{BMSSs and total amounts} 
\label{BMSStot}
 We only consider the system at steady state, that is all differential equations are put to zero and solved for the species concentrations under the constraints imposed by the conservation laws. Only solutions to the system of equations for which all species concentrations are non-negative, i.e.~positive or zero, are of interest. These solutions are called Biologically Meaningful Steady States, or BMSSs. Also, we assume that the {\it  total amounts are positive} (non-zero), i.e.~enzymes and substrates are always present in the system in some form (e.g.~in some phosphoform or intermediate complex). 

One can prove for each specific motif that  if all total amounts are positive, then all species concentrations at a BMSS are positive as well (i.e. if  a BMSS exists then all species concentrations must be positive). It follows that if 1) the species concentrations are non-negative (i.e.~positive or zero), 2) the total amounts are positive, and 3) the conservation laws and the steady state equations are fulfilled, then the species concentrations constitute in fact a BMSS and hence are positive. This observation is  frequently used in the following.  

\subsection{Method} 
For all motifs we follow the same procedure. We take the set of differential equations describing the systems together with the equations for the total amounts (the conservation laws) and solve for one variable, say an intermediate complex $Y$. Specifically, we choose one equation for a total amount, e.g.~$\oS=S_0+S_1+X+Y$, and use the other equations (differential equations and equations for total amounts) to  provide expressions for the variables at steady state in terms of $Y$, i.e.~we find expressions such that $S_0=S_0(Y)$, $S_1=S_1(Y)$ and $X=X(Y)$ are functions of $Y$. These functions  only depend on $Y$, the rate constants and the total amounts, \emph{excluding} $\oS$.  The functions are substituted into the expression for $\oS$ to obtain a function $\varphi(Y)$ that relates $Y$ to $\oS$, i.e.~$\oS=\varphi(Y)$. The analytic form of $\varphi$ determines how many BMSSs the system has for a given set of rate constants and total amounts. For example, if $\varphi$ is increasing there can only be one $Y$ corresponding to a given $\oS$. See Figure~\ref{summary} for illustration.

We allow $Y$ to be zero in which case $\oS=\varphi(Y)$ also is zero (it follows from the construction of $\varphi$ for each motif). This is done for simplicity as we otherwise would have to consider the limit of $\varphi(Y)$ for $Y\rightarrow 0$ in each case (see Section~\ref{BMSStot} about positive and non-negative concentrations). Some variables (depending on the motif) are not allowed to be zero as this could lead to division by zero.

\subsection{Technicalities}  
Some  manipulations of the equations are left to the reader if they only involve standard  elimination, rewriting and differentiation techniques. In general, it is a good idea  to do the calculations yourself (either by hand or in Mathematica$^{TM}$, Maple or similar) as many details are left out due to space and readability constraints. We have done all derivations  in Mathematica$^{TM}$\, and checked them manually. Proofs of propositions are collected in an appendix to keep the presentation simpler. Equation numbers are to the immediate right of the equations and only equations that are used later  are numbered.

\section{Steady states of the motifs}

\subsection{One-site phosphorylation cycles}

\paragraph{Motif (a).}
We recently used this motif as the building block of linear cascades and showed that it admits only one BMSS \cite{Feliu:2010p94}. We follow the approach taken in  \cite{Feliu:2010p94} and summarize it below. Essentially this is the approach taken for all motifs.

The chemical reactions of the motif are 

\centerline{\xymatrix{
S_{0} + E \ar@<0.5ex>[r]^(0.6){a^E} & X  \ar@<0.5ex>[l]^(0.4){b^E} \ar[r]^(0.4){c^E} & S_{1} + E  &
S_{1} + F \ar@<0.5ex>[r]^(0.6){a^F} & Y  \ar@<0.5ex>[l]^(0.4){b^F} \ar[r]^(0.4){c^F} & S_{0} + F
}}

\noindent The corresponding system of differential equations is:

\vspace{-0.2cm}
\begin{minipage}[h]{0.45\textwidth}
\begin{align}
\dot{S_{0}} &=  b^{E} X +c^F Y - a^{E} E S_{0}    \nonumber \\
\dot{E} &= (b^E + c^E)X- a^E E S_{0}  \nonumber \\
\dot{F} &=(b^F + c^F)Y -a^F F S_{1} \nonumber 
\end{align}
\end{minipage}
\hspace{0.5cm}
\begin{minipage}[h]{0.45\textwidth}
\begin{align}
\dot{S_{1}} &=  c^E X +b^{F} Y - a^{F} F S_{1}   \label{eq1a2}  \\
\dot{X} &=  -(b^E + c^E)X + a^E E S_{0} \label{eq1a4} \\
\dot{Y} &= -(b^F + c^F)Y + a^F F S_1. \label{eq1a5}
\end{align}
\end{minipage}

\medskip
It follows that $\dot{E}+\dot{X}=0$, $\dot{F}+\dot{Y}=0$ and $\dot{S_0}+\dot{S_1}+\dot{X}+\dot{Y}=0$. Hence, the conservation laws are given by  
$$\overline{E} = E + X,\quad\oF=F+Y,\quad   \overline{S}   =   S_{0}+S_1+X+Y.$$
Due to the conservation laws some equations are redundant, for example $\dot{F}=0$ and $\dot{Y}=0$ are the same equation. If fixed total amounts are given, we have to solve a system of 6 equations in 6 variables consisting of the equations for the total amounts and, for instance, equations \eqref{eq1a2}-\eqref{eq1a5}. From the latter equations, we obtain the equivalent system
\begin{equation}\label{reduced}
X - \eta E S_{0}=0, \quad Y-\delta F S_1=0,\quad X-\mu Y=0\end{equation}
with constants $\eta=a^E/(b^E + c^E)$, $\delta=a^F/(b^F + c^F)$ and 
$\mu= c^F/ c^E.$

It follows that at steady state  $E=\oE - \mu Y$ and $F=\oF - Y$. Note that if both  $\oE,\oF$ are positive (i.e. non-zero), then we cannot have $E=0$ or $F=0$ at  steady state (for example if $E=0$, then $X=0$ from \eqref{eq1a4}, hence $\oE=0$). Let $\xi = \min(\oF,\oE/\mu)$ and  $\Gamma=[0,\xi)$. The variables $E,F$ are both positive and $Y$ is non-negative only if $Y$ is in $\Gamma$. Further  it follows from \eqref{reduced} that
$$S_0=\frac{\mu Y}{\eta(\oE -\mu Y)}, \quad\textrm{and}\quad  S_1=\frac{Y}{\delta(\oF - Y)}.$$
(Note that division by zero does not occur as $E=\oE -\mu Y$ and $F=\oF - Y$ are both greater than zero.) These two functions are non-negative, increasing and $\mmC^1$ for $Y\in \Gamma$.  When $Y$ tends to $\xi$, one of them tends to $+\infty$, while the other remains bounded.

Let $\Delta=\mu\oF-\oE$, so that  $\Gamma=[0,\oF)$  if $\Delta\leq 0$, and $\Gamma=[0,\oE/\mu)$ if $\Delta>0$.  Using the conservation law for $\oS$ we obtain:

\begin{result}[Motif (a)]\label{motif1a} 
Let a one-site modification cycle  be given with positive   total amounts $\oS,\oE,\oF$. Then  the system has a  unique BMSS.  Specifically, the BMSS satisfies $\oS=\varphi(Y)$ for $Y$ in $\Gamma=[0,\xi)$, $\xi=\min(\oF,\oE/\mu)$, where
$$\oS=\varphi(Y)=\frac{\mu Y}{\eta(\oE -\mu Y)}+\frac{Y}{\delta(\oF - Y)}+(1+\mu)Y$$
is an increasing rational $\mmC^1$ function which tends to infinity as $Y$ tends to $\xi$ and fulfills $\varphi(0)=0$.
  \end{result}

\begin{obs}\label{poly}
Since $\varphi$ is a rational function, the equation $\oS=\varphi(Y)$ can be written in polynomial form by   elimination of denominators. In the present case
$$p(Y)=\mu\delta Y(\oF - Y)+\eta Y(\oE -\mu Y)+\eta\delta(\,(1+\mu)Y-\oS\,)(\oF - Y)(\oE -\mu Y),$$
which is a third degree polynomial in $Y$. Note that $p(0)<0$, $p(\xi)>0$, $p(\zeta)<0$ with $0<\xi<\zeta=\max(\oF,\oE/\mu)$, and $p(Y)$ tends to $+\infty$ as $Y$ tends to $+\infty$; hence $p(Y)$ has three positive roots. However, only the {\it first} root is in $\Gamma$, and it corresponds to the only BMSS of the system. In some systems, several positive roots are BMSSs. This remark is applicable in all cases below whenever $\varphi$ is a rational function.
\end{obs}

\begin{obs}\label{division_zero}
We argued above that $E$ and $F$ are non-zero if $\oE$ and $\oF$ are positive. This ensures that we do not divide by zero when constructing $\varphi$. For the remaining variables we only need to ensure non-negativity, cf.~Section \ref{BMSStot}.  For the other motifs, we  make use of similar reasoning.
\end{obs}

\paragraph{Motif (b).} Consider the system  in Motif (b) where the two catalyzing enzymes are the same. 
The chemical reactions of the system are given by

\centerline{\xymatrix{
S_{0} + E \ar@<0.5ex>[r]^(0.6){a_1} & X  \ar@<0.5ex>[l]^(0.4){b_1} \ar[r]^(0.4){c_1} & S_{1} + E  &
S_{1} + E \ar@<0.5ex>[r]^(0.6){a_2} & Y  \ar@<0.5ex>[l]^(0.4){b_2} \ar[r]^(0.4){c_2} & S_{0} + E
}}
\noindent
The corresponding system of differential equations is:

\vspace{-0.2cm}
\begin{minipage}[h]{\textwidth}
\begin{minipage}[t]{0.45\textwidth}
\begin{align}
\dot{S_{0}} &=  b_{1} X +  c_2 Y - a_{1} E S_{0}   \nonumber \\
\dot{E} &= (b_1 + c_1)X+(b_2 + c_2)Y - a_1 E S_{0} -a_2 E S_{1}  \nonumber
\end{align}
\end{minipage}
\begin{minipage}[t]{0.4\textwidth}
\begin{align}
\dot{S_{1}} &=  c_1 X +b_{2} Y - a_{2} E S_{1}   \label{eq1b2} \\
\dot{X} &=  -(b_1 + c_1)X + a_1 E S_{0} \label{eq1b4} \\
\dot{Y} &= -(b_2 + c_2)Y + a_2 E S_1. \label{eq1b5}\end{align}
\end{minipage}
\end{minipage}

\medskip
 We find that $\dot{E}+\dot{X}+\dot{Y}=0$ and $\dot{S_0}+\dot{S_1}+\dot{X}+\dot{Y}=0$. Hence, the conservation laws are given by  
$$\overline{E} = E + X+Y,\qquad  \overline{S}   =   S_{0}+S_1+X+Y.$$
If the total amounts are given, we need to solve a system of 5 equations in 5 variables consisting of those for the total amounts and, for instance, equations \eqref{eq1b2}-\eqref{eq1b5}. From the latter equations we obtain an equivalent system given by
$$X - \eta E S_{0}=0, \quad Y-\delta E S_1=0,\quad X-\mu Y=0$$
with constants $\eta=a_1/(b_1 + c_1)$, $\delta=a_2/(b_2 + c_2)$ and  $\mu= c_2/ c_1.$ 
Note that if $\oE> 0$, then $E=0$ cannot be a solution of the steady state equations and thus we require $E>0$ at any BMSS.

It follows that $E=\oE-(\mu+1)Y$. For $E>0$ and $Y\geq 0$ we have that  $Y$ must be in $\Gamma= [0,\oE/(1+\mu))$. Further, we  find 
 $$S_0 = \frac{\mu Y}{\eta (\oE-(\mu+1)Y)},\quad S_1=\frac{Y}{\delta (\oE-(\mu+1)Y)}.$$
 These functions are continuous and increasing for  $Y$ in $\Gamma$. In addition, all concentrations $S_0,S_1,X,E,F$ are non-negative as functions of $Y$ if and only if $Y\in \Gamma$. Using the conservation law for $\oS$ we obtain:

\begin{result}[Motif (b)]\label{motif1b} 
Let a one-site modification cycle be given with one enzyme acting as  kinase as well as phosphatase. Further, assume that the total amounts $\oS,\oE$ are positive. Then, the system has a  unique BMSS. 

Specifically, the BMSS satisfies $\oS=\varphi(Y)$ for $Y$ in $\Gamma=[0,\xi)$, $\xi=\oE/(1+\mu)$, where
$$\oS=\varphi(Y)= \frac{\mu Y}{\eta (\oE-(\mu+1)Y)}+\frac{ Y}{\delta (\oE-(\mu+1)Y)}+(1+\mu)Y$$
 is an increasing, rational $\mmC^1$ function which tends to infinity as $Y$ tends to $\xi$ and fulfills $\varphi(0)=0$.
 \end{result}

\paragraph{Motif (c).}
Consider now a one-site modification cycle with two competing kinases.  The chemical reactions of the system are given by  ($k=1,2$):

\centerline{\xymatrix{
S_{0} + E_k \ar@<0.5ex>[r]^(0.6){a^E_k} & X_k  \ar@<0.5ex>[l]^(0.4){b^E_k} \ar[r]^(0.4){c^E_k} & S_{1} + E_k  &
S_{1} + F \ar@<0.5ex>[r]^(0.6){a^F} & Y  \ar@<0.5ex>[l]^(0.4){b^F} \ar[r]^(0.4){c^F} & S_{0} + F
}}
\noindent
The corresponding system of differential equations is:

\vspace{-0.5cm}
\noindent
\begin{minipage}[t]{0.5\textwidth}
{\small\begin{align}
\dot{E_{k}} &= (b^E_k + c^E_k)X_k - a^E_k E_k S_{0} \nonumber \\
\dot{F} &= (b^F + c^F)Y - a^F F S_1 \nonumber  \\
\dot{S_{0}} &=  b^E_{1} X_{1}  + b^E_{2} X_{2}  +  c^F Y  - a^E_{1} E_{1}S_{0}   - a^E_{2} E_{2}S_{0}   \nonumber 
\end{align}}
\end{minipage}
\begin{minipage}[t]{0.48\textwidth}
{\small\begin{align}
\dot{X_k} &=  -(b^E_k + c^E_k)X_k + a^E_k E_k S_{0} \label{eq1c5} \\
\dot{Y} &= -(b^F + c^F)Y + a^F F S_1 \label{eq1c6} \\
\dot{S_{1}} &=  c^E_1 X_1 + c^E_2 X_2     +b^{F} Y - a^{F} F S_{1}\label{eq1c2}  
\end{align}}
\end{minipage}

\vspace{-0.3cm}\noindent
with $k=1,2$. 
It follows that  $\dot{X_k}+\dot{E_k}=0$, $\dot{Y}+\dot{F}=0$ and $\dot{S_0}+\dot{S_1}+\dot{X_1}+\dot{X_2}+\dot{Y}=0$, leading   to the conserved total amounts ($k=1,2$):
$$\overline{E}_k = E_{k} + X_{k},\quad \overline{F} = F + Y,\quad  \overline{S}  =   S_{0}+S_1 + X_{1}+X_2+Y. $$
Therefore, if total amounts are given, the steady states of the system are solutions to a system of 8 equations in 8 variables consisting of the equations for the total amounts (4 equations) and, for instance, equations \eqref{eq1c5} for $k=1,2$, \eqref{eq1c6}  and  \eqref{eq1c2}. From the latter equations we obtain an equivalent system given by
\begin{equation}\label{twokin}
X_k - \eta_k E_k S_{0}=0, \quad Y-\delta F S_1=0,\quad c_1^E X_1+c_2^EX_2 - c^F Y=0
\end{equation}
for $k=1,2$,  with constants $\eta_k=a^E_k/(b^E_k + c^E_k)$ and $\delta=a^F/(b^F + c^F)$. Note that if $\oE_k,\oF> 0$, then neither $E_k=0$ nor $F=0$ are solutions of the steady state equations. Thus, $E_k,F>0$ at any BMSS.

From the total amounts we have that $E_k=\oE_k-X_k$ and $F=\oF-Y$. Hence, for $E,F>0$ and $Y\geq 0$ we must have that $0\leq X_k<\oE_k$, $0\leq Y<\oF$ at BMSS.  We find 
 \begin{equation}\label{S1motifc}
S_1=\frac{Y}{\delta(\oF-Y)},
\end{equation}
 which is increasing in $Y$. Likewise, we obtain the following relation (with non-zero deno\-mi\-nators)
$$\frac{X_1}{\eta_1(\oE_1-X_1)}=S_0=\frac{X_2}{\eta_2(\oE_2-X_2)},$$
and it follows that
$$X_1=\phi_X(X_2)=\frac{\eta_1\oE_1 X_2}{\eta_2(\oE_2-X_2)+\eta_1X_2}.$$
This function is increasing and $\mmC^1$ for $X_2$ in $\Gamma_2= [0,\oE_2)$. Additionally,    $0\leq X_1= \phi_X(X_2)< \oE_1$ for $X_2$ in $\Gamma_2$. Note  that $S_0$ is also expressed as a non-negative,   increasing $\mmC^1$ function of $X_2\in \Gamma_2$.

From the last equation in \eqref{twokin} we obtain
$$Y=\phi_Y(X_2)=(c_1^E/c^F) \phi_X(X_2)+(c_2^E/c^F)X_2$$
which  is  a non-negative, increasing, rational  $\mmC^1$ function of $X_2$ for $X_2\in\Gamma_2$.  

If $\oF$ is in  $\phi_Y(\Gamma_2)$, then the condition $Y=\phi_Y(X_2)<\oF$ sets a more restrictive upper bound to $X_2$ than $\oE_2$ does. The functions $\phi_X,\phi_Y$ are defined for $X_2=\oE_2$, and we have $\phi_Y(\oE_2)  = (c_1^E/c^F) \oE_1+(c_2^E/c^F)\oE_2$. Let 
$\sigma=(c_1^E/c^F) \oE_1+(c_2^E/c^F)\oE_2$, $\xi'=\min(\oF,\sigma)$, $\xi=\phi_Y^{-1}(\xi')$ (which is well-defined) and $\Gamma=[0,\xi)\subseteq \Gamma_2$. Then  for $X_2\in \Gamma$, we have $Y<\oF$.

Since $S_1$ is a non-negative,  increasing  $\mmC^1$  function of $Y$ (for $Y<\oF$, \eqref{S1motifc}), it is also  a non-negative, increasing function of $X_2\in \Gamma$.   In conclusion, all steady state concentrations  are non-negative if and only if $X_2\in \Gamma$. Additionally, either $S_0$ or $S_1$ tend to infinity as $X_2$ approaches $\xi$. Using the conservation law for $\oS$ we obtain:

\begin{result}[Motif (c)]\label{motif1c} 
Let a one-site modification cycle be given with two different competing kinases and one phosphatase. Further, assume that positive total amounts $\oS,\oE_1,\oE_2,\oF$ are given. Then the system has a  unique BMSS. 

Specifically, the BMSS satisfies $\oS=\varphi(X_2)$ for $X_2$ in $\Gamma=[0,\xi)$, where
$$\oS=\varphi(X_2)=\frac{X_2}{\eta_2(\oE_2-X_2)} + \frac{\phi_Y(X_2)}{\delta(\oF-\phi_Y(X_2))}+\phi_X(X_2)+X_2+\phi_Y(X_2)$$
 is an increasing, rational $\mmC^1$ function which tends to infinity as $X_2$ tends to $\xi$ and fulfills $\varphi(0)=0$.
 \end{result}

\paragraph{Motif  (d).}
Consider a one-site modification cycle with two competing kinases and two competing phosphatases. The chemical reactions of the system are given by

\centerline{\xymatrix{
S_{0} + E_1 \ar@<0.5ex>[r]^(0.6){a^E_1} & X_1  \ar@<0.5ex>[l]^(0.4){b^E_1} \ar[r]^(0.4){c^E_1} & S_{1} + E_1  &
S_{1} + F_1 \ar@<0.5ex>[r]^(0.6){a^F_1} & Y_1  \ar@<0.5ex>[l]^(0.4){b^F_1} \ar[r]^(0.4){c^F_1} & S_{0} + F_1 \\
S_{0} + E_2 \ar@<0.5ex>[r]^(0.6){a^E_2} & X_2  \ar@<0.5ex>[l]^(0.4){b^E_2} \ar[r]^(0.4){c^E_2} & S_{1} + E_2  &
S_{1} + F_2 \ar@<0.5ex>[r]^(0.6){a^F_2} & Y_2  \ar@<0.5ex>[l]^(0.4){b^F_2} \ar[r]^(0.4){c^F_2} & S_{0} + F_2
}}
\noindent
The corresponding system of differential equations is:

\vspace{-0.5cm}
\noindent
\begin{minipage}[t]{0.65\textwidth}
{\small\begin{align}
\dot{X_k} &=  -(b^E_k + c^E_k)X_k + a^E_k E_k S_{0} \label{eq1d5} \\
\dot{Y_k} &= -(b^F_k + c^F_k)Y_k + a^F_k F_k S_1 \label{eq1d6} \\
\dot{S_{0}} &=  b^E_{1} X_{1}   +  c^F_{1} Y_{1} - a^E_{1} E_{1}S_{0} + b^E_{2} X_{2} +  c^F_{2} Y_{2} - a^E_{2} E_{2}S_{0}   \nonumber \\
\dot{S_{1}} &=  c^E_1 X_1 +b^{F}_1 Y_1 - a^{F}_1 F_1 S_{1}+ c^E_2 X_2 +b^{F}_2 Y_2 - a^{F}_2 F_2 S_{1}   \label{eq1d2}  
\end{align}}
\end{minipage}
\hspace{0.2cm}
\begin{minipage}[t]{0.3\textwidth}
{\small \begin{align}
\dot{E_{k}} &= (b^E_k + c^E_k)X_k - a^E_k E_k S_{0} \nonumber \\
\dot{F_{k}} &= (b^F_k + c^F_k)Y_k - a^F_k F_k S_1 \nonumber  
\end{align}}
\end{minipage}

\vspace{-0.5cm}\noindent
with $k=1,2$. 
It follows that  $\dot{X_k}+\dot{E_k}=0$, $\dot{Y_k}+\dot{F_k}=0$ and $\dot{S_0}+\dot{S_1}+\dot{X_1}+\dot{X_2}+\dot{Y_1}+\dot{Y_2}=0$, leading to the following conserved total amounts ($k=1,2$):
$$\overline{E}_k = E_{k} + X_{k},\quad \overline{F}_k = F_{k} + Y_k,\quad  \overline{S}  =   S_{0}+S_1 + X_{1}+X_2+Y_1+Y_2. $$
Therefore, if total amounts are given, we have to solve a system of 10 equations in 10 variables consisting of the equations for the total amounts (5 equations) and, for instance, equations \eqref{eq1d5}, \eqref{eq1d6} for $k=1,2$ and  \eqref{eq1d2}.  From the latter equations we obtain an equivalent system given by
\begin{equation}\label{twokinphos}
X_k - \eta_k E_k S_{0}=0, \quad Y_k-\delta_k F_k S_1=0,\quad c_1^E X_1+c_2^EX_2 - c_1^F Y_1 - c_2^FY_F=0
\end{equation}
for $k=1,2$ with constants $\eta_k=a^E_k/(b^E_k + c^E_k)$ and $\delta_k=a^F_k/(b^F_k + c^F_k)$.

If $\oE_k,\oF_k> 0$, then $E_k,F_k\neq 0$ at steady state. As for the previous motif, we have $E_k=\oE_k-X_k$ and $F_k=\oF_k-Y_k$; hence  $0\leq X_k<\oE_k$ and $0\leq Y_k<\oF_k$ is required for any BMSS. The concentration $S_0$ can be  expressed in two different ways as increasing $\mmC^1$ functions: As a function of $X_1$ and  as a function of  $X_2$. When these two expressions are equated we obtain the relation (similar to the relation obtained for the previous motif)
$$X_1=\phi_X(X_2)=\frac{\eta_1\oE_1 X_2}{\eta_2(\oE_2-X_2)+\eta_1X_2}.$$
It is a non-negative,  increasing $\mmC^1$ function for $X_2\in [0,\oE_2)$ such that $E_1=\oE_1-\phi_X(X_2)$ also is  a positive function of $X_2\in [0,\oE_2)$.  Note that $\phi_X(\oE_2)=\oE_1$ is well-defined. 

Similarly, $S_1$ can be expressed as  increasing $\mmC^1$ functions of $Y_1$ and of $Y_2$, respectively, which provide the relation
$$Y_1=\phi_Y(Y_2)=\frac{\delta_1\oF_1 Y_2}{\delta_2(\oF_2-Y_2)+\delta_1Y_2}.$$
It is a non-negative, increasing $\mmC^1$ function for $Y_2\in [0,\oF_2)$ and  $F_1=\oF_1-\phi_Y(Y_2)$ is a positive function of $Y_2\in [0,\oF_2)$.

Finally, the last relation in \eqref{twokinphos} gives
$$  c_1^E \phi_X(X_2)+c_2^EX_2 =   c_1^F \phi_Y(Y_2)+ c_2^FY_2. $$
The left  side is an increasing $\mmC^1$ function in $X_2$, the right side an increasing $\mmC^1$ function in $Y_2$ which tends to infinity as $Y_2$ tends to infinity.  Hence,  there exists an increasing $\mmC^1$ function $g(X_2)=Y_2$ defined on $X_2\in [0,\oE_2)$ relating $X_2$ to $Y_2$.

In summary, the concentrations $E_1,E_2,X_1,S_0$ are non-negative  functions of $X_2$ if and only if $X_2\in [0,\oE_2)$.  The concentrations $F_1,F_2,Y_1,S_1$ are non-negative if and only if $Y_2\in [0,\oF_2)$. Hence, to ensure that all concentrations are non-negative, we require $Y_2=g(X_2)<\oF_2$. Since $g$ is increasing, it is bounded above by $g(\oE_2)$ (well-defined).  Therefore, let $\Gamma=[0,\oE_2)$ if $g(\oE_2)\leq \oF_2$ and $\Gamma=[0,g^{-1}(\oF_2))$ otherwise. Using the conservation law for $\oS$ we obtain:

\begin{result}[Motif (d)]\label{motif1d} 
Let a one-site modification cycle with two competing kinases and two competing phosphatases be given. Further, assume that  the total amounts $\oS,\oE_1,\oE_2,\oF_1,\oF_2$ are positive. Then, the system has a  unique BMSS. 

Specifically, the BMSS satisfies $\oS=\varphi(X_2)$ for $X_2$ in $\Gamma=[0,\xi)$, where
$$\oS=\varphi(X_2)=\frac{X_2}{\eta_2(\oE_2-X_2)} + \frac{g(X_2)}{\delta_2(\oF_2 g(X_2))}+\phi_X(X_2)+X_2+\phi_Y^1(g(X_2))+g(X_2)$$
 is an increasing $\mmC^1$ function which tends to infinity as $X_2$ tends to $\xi$ and fulfills $\varphi(0)=0$.
 \end{result}

The function $g$ is not rational, hence neither $\varphi$ is rational.

\subsection{Two-site phosphorylation cycles}
\label{motiftwosite}

\paragraph{Motif  (e).}
First we  consider a two-site phosphorylation system in which modifications are carried out by different kinases and phosphatases for each phosphoform. For simplicity, we assume that both phosphorylation and dephosphorylation occur sequentially.
The chemical reactions of the system are:

\centerline{
\xymatrix{
S_{0} + E_1 \ar@<0.5ex>[r]^(0.6){a^E_1} & X_1  \ar@<0.5ex>[l]^(0.4){b^E_1} \ar[r]^(0.4){c^E_1} & S_{1} + E_1 & 
S_{1} + E_2 \ar@<0.5ex>[r]^(0.6){a^E_2} & X_2  \ar@<0.5ex>[l]^(0.4){b^E_2} \ar[r]^(0.4){c^E_2} & S_{2} + E_2 \\
S_{1} + F_1 \ar@<0.5ex>[r]^(0.6){a^F_1} & Y_1  \ar@<0.5ex>[l]^(0.4){b^F_1} \ar[r]^(0.4){c^F_1} & S_{0} + F_1 &
S_{2} + F_2 \ar@<0.5ex>[r]^(0.6){a^F_2} & Y_2  \ar@<0.5ex>[l]^(0.4){b^F_2} \ar[r]^(0.4){c^F_2} & S_{1} + F_2
}}
\noindent
The differential equations describing the  system are:

\vspace{-0.5cm}
\noindent
\begin{minipage}[t]{0.3\textwidth}
{\small \begin{align}
   \dot{X_k} &=   a^E_k E_k S_{k-1}-(b^E_k + c^E_k)X_k \nonumber  \\
   \dot{Y_k} &=  a^F_k F_k S_k-(b^F_k + c^F_k)Y_k , \nonumber  \\
\dot{S_{0}} &=  b^E_{1} X_{1} +  c^F_{1} Y_{1}  - a^E_{1} E_{1}S_{0}   \nonumber 
\end{align}}
\end{minipage}
\hspace{0.2cm}
\begin{minipage}[t]{0.65\textwidth}
{\small \begin{align}
\dot{E_{k}} &= (b^E_k + c^E_k)X_k - a^E_k E_k S_{k-1}  \label{eq1e3}  \\
\dot{F_{k}} &= (b^F_k + c^F_k)Y_k - a^F_k F_k S_k  \label{eq1e4}  \\
\dot{S_{2}} &=  c^E_2 X_2  +b^F_2 Y_2  - a^F_2 F_2 S_2  \label{eq1e2}  \\
\dot{S_{1}} &=  c^E_1 X_1 +b^E_{2} X_{2}   +b^F_1 Y_1 +  c^F_{2} Y_{2}  - (a^F_1 F_1+ a^E_{2} E_{2})S_{1}  \label{eq1e5} 
\end{align}}
\end{minipage}

\vspace{-0.5cm}
\noindent with $k=1,2$. We have $\dot{X_k}+\dot{E_k}=0$, $\dot{Y_k}+\dot{F_k}=0$ and $\dot{S_0}+\dot{S_1}+\dot{S_2}+\dot{X_1}+\dot{X_2}+\dot{Y_1}+\dot{Y_2}=0$, which  lead to  the following conserved total amounts ($k=1,2$):
\begin{equation}\label{totalamounts3}
\overline{E}_k = E_{k} + X_{k},\quad \overline{F}_k = F_{k} + Y_k,\quad  \overline{S}  =   S_{0}+S_1+S_2 + X_{1}+X_2+Y_1+Y_2. \end{equation}
Therefore, if total amounts are given, we have to solve a system of 11 equations in 11 variables consisting of those in \eqref{totalamounts3} and, for instance, equations \eqref{eq1e3}-\eqref{eq1e5}.
From the latter equations, we obtain 
\begin{equation}\label{twocycle_dif}
X_k = \eta_k E_k S_{k-1}, \quad Y_k=\delta_k F_k S_k,\qquad k=1,2
\end{equation}
with constants $\eta_k=a^E_k/(b^E_k + c^E_k)$ and $\delta_k=a^F_k/(b^F_k + c^F_k)$. Equation \eqref{eq1e4}  and \eqref{eq1e2} for $k=2$ give $c^E_2 X_2 - c^F_{2} Y_{2}=0.$ This relation together with equation \eqref{eq1e5}, \eqref{eq1e3} for $k=2$, and \eqref{eq1e4} for $k=1$ give
$c^E_1 X_1 - c^F_{1} Y_{1}=0.$ Therefore, we have that
$$X_k = \mu_k Y_k,\qquad \mu_k= c^F_{k}/ c^E_{k}. $$

Let $\Delta_k=\mu_k\oF_k-\oE_k$,  $\xi_k = \min(\oF_k,\oE_k/\mu_k)$  and $\Gamma_k=[0,\xi_k)$.
Note that if $\oE_k,\oF_k > 0$, then at steady state $E_k,F_k \neq 0$. Since $E_k=\oE_k-\mu_k Y_k$ and $F_k=\oF_k-Y_k$ at steady state, any BMSS must satisfy $Y_k\in \Gamma_k$ to ensure non-negativity of $X_k, Y_k$ and positivity of $E_k,F_k$ as functions of $Y_k$ for each $k$.

For fixed values of $S_2,X_2,Y_2$, the steady state values of the remaining variables satisfy the steady state equations of  Motif (a) (a one-site phosphorylation cycle) with species $S_0,S_1,X_1,Y_1,E_1,F_1$ and total amounts $\oE_1$, $\oF_1$ and $\oS - S_2-X_2-Y_2$. Therefore, using Result \ref{motif1a}, the BMSSs  of the system satisfy 
$$\oS = \varphi_1(Y_1)+(S_2+X_2+Y_2)$$
with $Y_1\in \Gamma_1$. Here $\varphi_1$ denotes the function $\varphi$ in Result \ref{motif1a}.

Using the second equality in \eqref{twocycle_dif} for $k=2$ together with the conservation law for $\oF_2$, we obtain
\begin{equation}\label{s2motife}
\varphi_2(Y_2)=S_2+X_2+Y_2=\frac{Y_2}{\delta_2(\oF_2-Y_2)}+ \mu_2 Y_2  + Y_2,
\end{equation}
which is a non-negative, increasing $\mmC^1$ function  of $Y_2\in \Gamma_2$. 
Consequently, the BMSSs of the system satisfy the relation
\begin{equation}\label{stot1} 
\oS = \varphi_1(Y_1) + \varphi_2(Y_2),
\end{equation}
for $Y_1\in \Gamma_1$ and $Y_2\in \Gamma_2$. The right hand side is an increasing function in both variables. 

From  the  equation $X_2=\eta_2 E_2 S_1$ together with  $S_1$ expressed as a function of $Y_1$ (similar to that of $S_2$ in equation \eqref{s2motife}) we obtain
$$X_2= \eta_2 (\oE_2-X_2) S_1= \frac{\eta_2 (\oE_2-X_2)  Y_1}{\delta_1(\oF_1 - Y_1) }.$$ 
Rewriting this equation yields 
\begin{equation}\label{f1}
X_2 =\frac{\eta_2 \oE_2 Y_1}{\delta_1 (\oF_1 - Y_1) + \eta_2 Y_1}, \quad \textrm{and hence} \quad Y_2 =f(Y_1)=\frac{\eta_2 \oE_2 Y_1}{\mu_2\delta_1 (\oF_1 - Y_1) + \mu_2\eta_2 Y_1}. 
\end{equation}
The function $f$ is an increasing function, which is non-negative and $\mmC^1$ for $Y_1\in \Gamma_1$. Additionally, $f$ is defined for $Y_1=\oF_1$ with $f(\oF_1)=\oE_2/\mu_2$.

By substitution of $f(Y_1)$ into \eqref{stot1}, the BMSSs of the system satisfy  
$$\oS = \varphi(Y_1)= \varphi_1(Y_1) + \varphi_2(f(Y_1))$$
for   $Y_1\in \Gamma_1$ and $f(Y_1)\in \Gamma_2$. Since $f$ is increasing, continuous and $f(0)=0$, this condition is equivalent to  $Y_1\in \Gamma=[0,\xi)$ with $\xi=\min(\xi_1, f^{-1}(\xi_2))$ (note that $f^{-1}(\xi_2)$ is well-defined). The function $\varphi$ is a rational function, which is increasing and $\mmC^1$ in $\Gamma$ and either $\varphi_1$ or $\varphi_2\circ f$ tends to infinity as $Y_1$ tends to $\xi$. 
Additionally,  the BMSS concentrations of all other species derived from the formulas above are non-negative functions of $Y_1$ if and only if  $Y_1\in \Gamma$. 

Using the conservation law for $\oS$ we obtain:
 
 \begin{result}[Motif (e)]\label{result_twocycle_dif} 
Let a two-site phosphorylation cycle be given with different kinases and phosphatases. Further assume that the total amounts $\oE_k,\oF_k,\oS$, $k=1,2$ are positive.   Then, the system has a  unique BMSS. 

Specifically, the BMSS satisfies $\oS=\varphi(Y_1)$ for $Y_1$ in $\Gamma=[0,\xi)$, where
$$\oS = \varphi(Y_1)= \varphi_1(Y_1) + \varphi_2(f(Y_1))$$
 is an increasing rational $\mmC^1$ function, which tends to infinity as $Y_1$ tends to $\xi$ and fulfills $\varphi(0)=0$.
 \end{result}

\paragraph{Motif  (f).}
Next, we consider    a two-site phosphorylation system where phosphorylation is catalyzed by the same kinase at both sites  but dephosphorylation is catalyzed by different phosphatases. Again, we assume sequential (de)phosphorylation.  The chemical reactions of the system are:

\centerline{
\xymatrix{
S_{0} + E \ar@<0.5ex>[r]^(0.6){a^E_1} & X_1  \ar@<0.5ex>[l]^(0.4){b^E_1} \ar[r]^(0.4){c^E_1} & S_{1} + E & 
S_{1} + E \ar@<0.5ex>[r]^(0.6){a^E_2} & X_2  \ar@<0.5ex>[l]^(0.4){b^E_2} \ar[r]^(0.4){c^E_2} & S_{2} + E \\
S_{1} + F_1 \ar@<0.5ex>[r]^(0.6){a^F_1} & Y_1  \ar@<0.5ex>[l]^(0.4){b^F_1} \ar[r]^(0.4){c^F_1} & S_{0} + F_1 &
S_{2} + F_2 \ar@<0.5ex>[r]^(0.6){a^F_2} & Y_2  \ar@<0.5ex>[l]^(0.4){b^F_2} \ar[r]^(0.4){c^F_2} & S_{1} + F_2
}}
\noindent
The differential equations describing the system  are the following:

\vspace{-0.5cm}
\noindent
\hspace{-0.4cm}
\begin{minipage}[t]{0.63\textwidth}
{\small\begin{align}
\dot{S_{0}} &=  b^E_{1} X_{1} +  c^F_{1} Y_{1}  - a^E_{1} E S_{0}   \nonumber   \\
\dot{Y_k} &=  -(b^F_k + c^F_k)Y_k +a^F_k F_k S_k  \nonumber  \\
\dot{E} &= (b^E_1 + c^E_1)X_1+(b^E_2 + c^E_2)X_2 - (a^E_1 S_{0} +a^E_2 S_{1})E\nonumber \\
\dot{S_{1}} &=  c^E_1 X_1 +b^E_{2} X_{2}+b^F_1 Y_1 +  c^F_{2} Y_{2} - (a^F_1 F_1+ a^E_{2} E) S_{1}    \label{eq1h4}
\end{align}}
\end{minipage}
\hspace{-0.7cm}
\begin{minipage}[t]{0.4\textwidth}
{\small \begin{align}
\dot{S_{2}} &=  c^E_2 X_2  +b^F_2 Y_2  - a^F_2 F_2 S_2  \label{eq1h3}  \\
\dot{X_k} &=  -(b^E_k + c^E_k)X_k  +a^E_k E S_{k-1} \label{eq1h1}  \\
\dot{F_{k}} &= (b^F_k + c^F_k)Y_k - a^F_k F_k S_k  \label{eq1h2} 
\end{align}}
\end{minipage}

\vspace{-0.3cm}
\noindent 
with $k=1,2$. The conservation laws are given by  
\begin{align}\label{totalamounts2}
\overline{E} &= E + X_{1}+X_2,\quad \overline{F}_k = F_{k} + Y_k,\quad  \overline{S}   =   S_{0}+S_1+S_2 + X_{1}+X_2+Y_1+Y_2, \nonumber \end{align}
$k=1,2$. Therefore, if total amounts are given, the system to be solved consists of 10 equations in 10 variables which are the equations for the total amounts and, for instance, equations \eqref{eq1h4}-\eqref{eq1h2} for $k=1,2$.
From \eqref{eq1h1} and \eqref{eq1h2} we obtain 
$$X_k = \eta_k E S_{k-1}, \quad Y_k=\delta_k F_k S_k,\qquad k=1,2$$
with constants $\eta_k=a^E_k/(b^E_k + c^E_k)$ and $\delta_k=a^F_k/(b^F_k + c^F_k)$.
Proceeding  as in the previous system, we find that 
$$X_k = \mu_k Y_k,\qquad \mu_k= c^F_{k}/ c^E_{k}. $$

Let  $\Delta_k=\mu_k\oF_k-\oE$,  $\xi_k = \min(\oF_k,\oE/\mu_k)$  and $\Gamma_k=[0,\xi_k)$. If $\oE,\oF_k > 0$, then at steady state $E,F_k \neq  0$. From $E=\oE - X_1-X_2$ and $F_k= \oF_k - Y_k$ we see that positivity of $E,F_k$ and non-negativity of $Y_k$ requires that (at least)  $Y_k\in \Gamma_k$.  If $Y_k\in\Gamma_k$, then $X_k$ is also a non-negative, increasing function of $Y_k$.

The situation resembles the situation of the previous system where catalysis is mediated by two different kinases. 
In both systems we  have $X_2=\mu_2Y_2$ and $S_2=\frac{Y_2}{\delta_2(\oF_2-Y_2)}$ and so  the concentration of $S_2$ is a  function of $Y_2$. It follows that 
$$\varphi_2(Y_2)=S_2+X_2+Y_2=\frac{Y_2}{\delta_2(\oF_2-Y_2)}+\mu_2Y_2+Y_2$$ 
is a non-negative, increasing continuous function of $Y_2\in \Gamma_2$.

For a fixed value of $Y_2$, the steady state values of the remaining variables (except $X_2$ and $S_2$) satisfy the steady state equations of a  one-site phosphorylation cycle  with species $S_0,S_1,X_1,Y_1,E,F_1$ and total amounts
$\oE-X_2=\oE-\mu_2Y_2$, $\oF_1$ and $\oS - \varphi_2(Y_2)$. Using Result \ref{motif1a} with  $\varphi_1(\cdot,Y_2)$ denoting $\varphi$ for a fixed $Y_2$, the BMSSs satisfy the relation
$$\oS = \varphi_1(Y_1,Y_2) + \varphi_2(Y_2)$$
 for any $Y_2\in \Gamma_2$ and $0\leq Y_1 < \min(\oF_1,(\oE-\mu_2Y_2)/\mu_1)$. Under these conditions, all concentrations of the other chemical species are non-negative functions of $Y_1, Y_2$. Note that the total amount of enzyme $\oE-\mu_2Y_2$ is part of the function $\varphi_1$ and hence $\varphi_1$ depends on $Y_2$. Indeed, we have
$$S_0=\frac{\mu_1 Y_1}{\eta_1(\oE -\mu_1 Y_1-\mu_2Y_2)}.$$
The function $\varphi_1$ is increasing in $Y_1$ and in $Y_2$.

The equation $X_2=\eta_2 E S_1$ combined with  $S_1$ expressed as a function of $Y_1$  provide (after isolation of $X_2$) the following relation at steady state
 $$Y_2=f(Y_1)=X_2/\mu_2=\frac{\eta_2(\oE -\mu_1 Y_1) Y_1}{\mu_2\delta_1( \oF_1 -Y_1) + \mu_2\eta_2 Y_1}. $$
This function resembles that in \eqref{f1} except from the quadratic term in the numerator which gives $f$ a very different analytic form from that in \eqref{f1}.

The function $f$ is $\mmC^1$ and  takes non-negative values for $Y_1<\xi_1$. Therefore, a BMSS satisfies
\begin{equation}\label{stot2}\oS = \varphi(Y_1)= \varphi_1(Y_1,f(Y_1)) + \varphi_2(f(Y_1)),\end{equation}
for $Y_1\in \Gamma_1$,  such that $f(Y_1)\in \Gamma_2$.  Since both $\varphi_1,\varphi_2$ are increasing functions, the behavior of $\varphi$ needs to be understood from the behavior of $f$. 

If we let 
$$\Gamma=\{Y_1\in \Gamma_1| f(Y_1)\in \Gamma_2 \},$$
then the concentrations of all the chemical species  are non-negative when expressed in term of  $Y_1$, if and only if $Y_1\in \Gamma$. If $Y_1=0$, then $f(0)=0$ and it follows that $\varphi(0)=0$; in particular, it follows that $0\in \Gamma$. 
Note that $\varphi$ is $\mmC^1$ in $\Gamma$; however $\Gamma$ might not be a connected interval as we will see below.

The function $f$ is $\mmC^1$ and non-negative for $Y_1\in \Gamma_1$. Depending on whether it is monotone or not, different forms of $\varphi$ are expected. These behaviors can be found by computing the derivative of $f$ (see Figure~\ref{ffig}). Let
\begin{align*}
\Lambda & =(1+\eta_2/\delta_1)\mu_1 \oF_1 - \oE.
\end{align*}
Observe that if $\Lambda\leq 0$ then $\Delta_1\leq 0$.

\begin{figure}[t]
\includegraphics{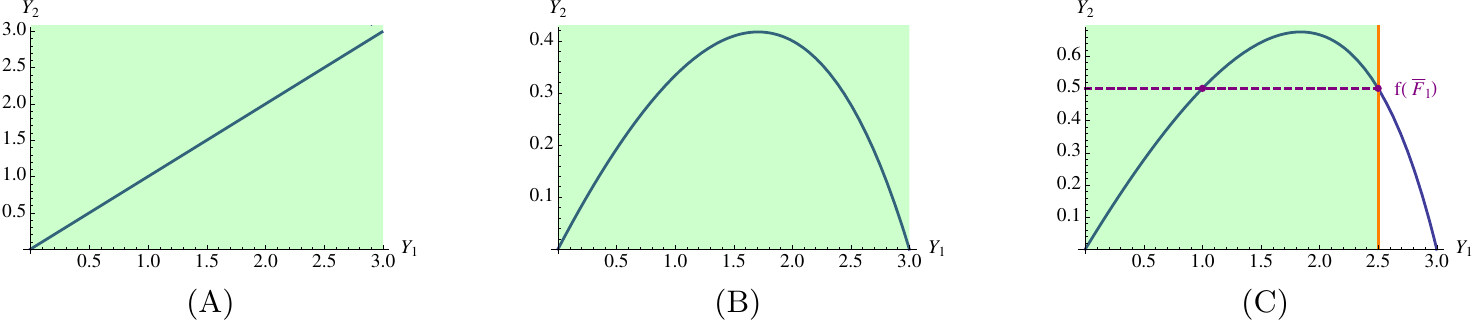}
\caption{Different behaviors of the function $f(Y_1)=Y_2$.  The region $\Gamma_1$ is marked in green. (A) corresponds to Proposition~\ref{motif1f1} (i) with $Y_1\in \Gamma_1=[0,\oF_1)$. (B) and (C) differ in whether (B) $Y_1\in \Gamma_1=[0,\oE/\mu_1)$  or (C) $Y_1\in \Gamma_1=[0,\oF_1)$; in both cases $f(\oE/\mu_1)=0$. If $\oF_2>f(\alpha)$ (the top point of the curve) then Proposition~\ref{motif1f1} (ii)(c) applies. If $\oF_2\leq f(\alpha)$ then Proposition~\ref{motif1f1} (ii)(b) applies for (B). In (C), if $\oF_2\leq f(\oF_1)$ (the dashed line) then Proposition~\ref{motif1f1} (ii)(a) applies, whereas Proposition~\ref{motif1f1} (ii)(b) applies if $f(\oF_1)<\oF_2\leq f(\alpha)$. }\label{ffig}
\end{figure}
 
\begin{proposition}\label{motif1f1} 
The following statements hold:
\begin{enumerate}[(i)]
\item If $\Lambda \leq 0$, then $f$ is an increasing $\mmC^1$ function of $Y_1\in \Gamma_1= [0,\oF_1)$. Further, we have  $\Gamma=[0,\min(\oF_1,f^{-1}(\xi_2))$ (with $f^{-1}(\xi_2)$ set to $+\infty$ if not defined).
\item If $\Lambda>0$, then there is $\alpha\in \Gamma_1=[0,\xi_1)$ such that $f'(\alpha)=0$, $f$ is increasing in $[0,\alpha)$ and decreasing in $(\alpha,\xi_1)$.
In this case, we obtain:
\begin{enumerate}[(a)]
\item  If $\oE-\mu_1\oF_1-\mu_2\oF_2\geq 0$,  then there is $\alpha_1\leq\alpha$ such that $f$ is increasing and $\mmC^1$  in  $\Gamma=[0,\alpha_1)$ and $f(\alpha_1)=\oF_2$. 
\item If $\oE-\mu_1\oF_1-\mu_2\oF_2<0$ and $\oF_2\leq f(\alpha)$, then
$\Gamma= [0,\alpha_1)\cup (\alpha_2,\xi_1)$ with $\alpha_1\leq\alpha\leq\alpha_2$ and $f(\alpha_1)=f(\alpha_2)=\oF_2$. Hence,  $f$ is increasing in $[0,\alpha_1)$ and decreasing in $(\alpha_2,\xi_1)$.
\item If $\oF_2>f(\alpha)$, then $\Gamma=[0,\xi_1)$.
\end{enumerate}
\end{enumerate}
\end{proposition}
All possibilities are covered in the proposition. Indeed, the condition in (ii)(c) implies $\oE-\mu_1\oF_1-\mu_2\oF_2<0$: If not, according to (ii)(a) there would be $\alpha_1\leq\alpha$ such that $\oF_2=f(\alpha_1)\leq f(\alpha)<\oF_2$, which is a contradiction.

In the cases (i) and (ii)(a), $f$ is  an increasing $\mmC^1$ function in a connected interval $\Gamma=[0,\xi)$. Hence, by composition of functions, $\varphi$ in \eqref{stot2} is also an increasing $\mmC^1$ function of $Y_1\in \Gamma$ with $\varphi(0)=0$. Additionally, in both cases $\varphi$ tends to infinity as  $Y_1$ tends to $\xi$ since  either $Y_2=f(Y_1)$ tends to $\oF_2$ or $Y_1$ to $\oF_1$. We conclude that there is exactly one BMSS in each of these cases.

The cases (ii)(b) and (ii)(c) require further analysis. In case (ii)(b), let
$$\Gamma'=[0,\alpha_1), \qquad \Gamma''=(\alpha_2,\xi_1).$$ The values $\alpha_1,\alpha_2$ correspond to the values of $Y_1$ for which $f(Y_1)=\oF_2$ and hence they are the zeros of the denominator of $S_2$:
$$S_2=\frac{Y_2}{\delta_2(\oF_2-Y_2)}=\frac{Y_2}{\delta_2(\oF_2-f(Y_1))}.$$
 In $\Gamma'$, $f$ is  increasing  and therefore $\varphi$ is also increasing. It  tends to infinity as $Y_1$ tends to $\alpha_1$. Hence, for any value of $\oS$, there is a BMSS corresponding to a $Y_1\in \Gamma'$.

In $\Gamma''$, $f$ is a decreasing function and hence it is uncertain when/whether $\varphi$ is increasing. However, we find that (A) when $Y_1$ tends to $\alpha_2$ from the right,   the function $\varphi_2\circ f$ tends to $+\infty$, while $\overline{\varphi}_1=\varphi_1(\cdot,f(\cdot))$ is bounded; (B) when $Y_1$ tends to $\xi_1$ from the left, the function $\varphi_2\circ f$ is bounded, while $\overline{\varphi}_1$ tends to $+\infty$ (in fact, either $S_1$ or $S_0$ expressed in terms of $Y_1$ does). It follows that in the interval $\Gamma''=(\alpha_2,\xi_1)$, the function $\varphi$ starts  decreasing from infinity and ends  increasing towards infinity. By continuity, there exists a minimum $\oS_{\min}$ of $\varphi$ in this interval. We see that for $\oS> \oS_{\min}$, at least two values of $Y_1$  satisfy $\varphi(Y_1)=\oS$; hence  at least two  BMSSs exist in this interval. All together, we conclude that at least three BMSSs occur in this case,  one in  $\Gamma'$ and two in  $\Gamma''$. Note that when $\oS= \oS_{\min}$ then there are at least two (not at least three) as the two in $\Gamma''$ coincide.

\begin{figure}[!t]
\includegraphics{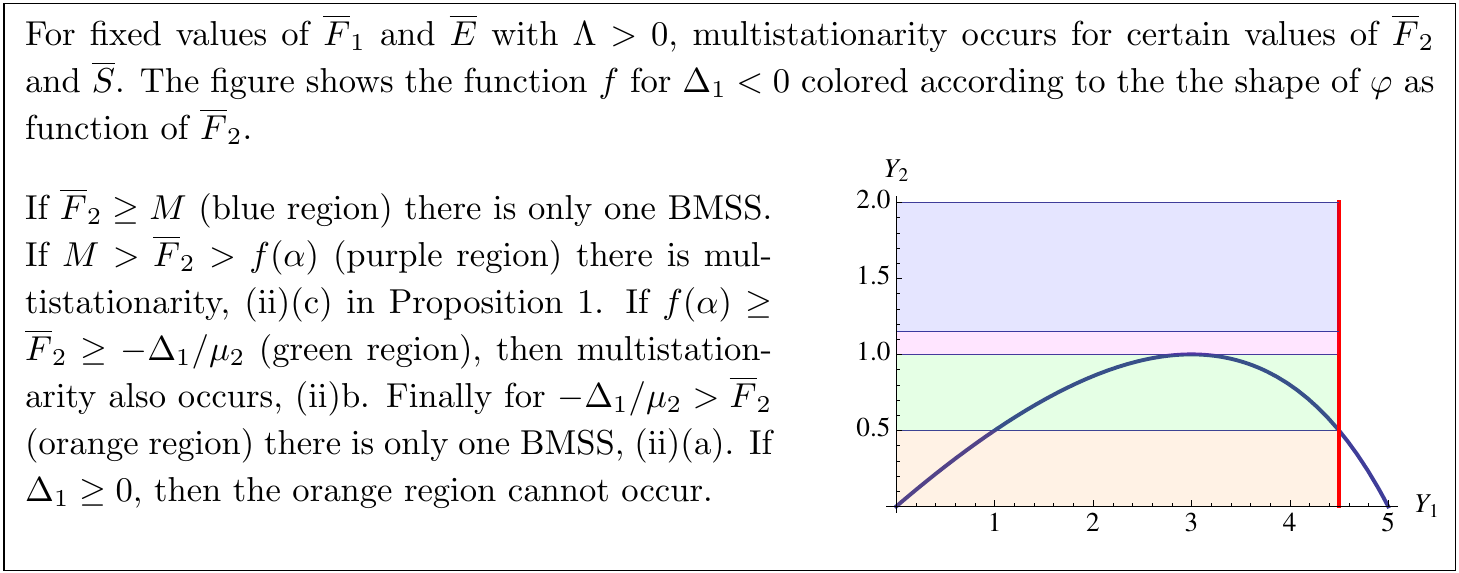}
\caption{Different shapes of $\varphi$ depending on $\oE,\oF_1,\oF_2$. }\label{regions}
\end{figure}

Finally, let us consider case (ii)(c), that is the case $\Lambda>0$, $\oF_2>f(\alpha)$ and consequently   $\Gamma=[0,\xi_1)$. The function $\varphi$ is increasing at $Y_1=0$ and tends to infinity as $Y_1$ approaches $\xi_1$. 
Multistationarity can only occur in the form of Figure~\ref{summary}(c) and there will be different values of $Y_1$ corresponding to the same value of $\oS$, only if the function $\varphi$ is not always increasing; that is if the derivative of $\varphi$ has more than two zeros.  
Equivalently, if there exists $Y_1$ for which $\varphi'(Y_1)<0$, then we are guaranteed multistationarity.  This result is stated in the following proposition:

\begin{proposition}\label{varphider}
Assume that $\Lambda>0$. Then there exists a value $M>f(\alpha)$ such that, for  $\oF_2\geq M$, the derivative of $\varphi$ is positive (except potentially for a finite number of points where it is zero). Hence,  $\varphi$ is increasing.  Further for  $M>\oF_2>f(\alpha)$,  there exist values of $Y_1$ for which  $\varphi'(Y_1)< 0$. 
\end{proposition}

We conclude that for all $\oF_2\geq M$, there is exactly one BMSS for any value of $\oS$, while for all $M> \oF_2>f(\alpha)$ there is multistationarity for certain values of $\oS$. In particular, multistationarity occurs for $\oS$ satisfying $\oS_{\min}\leq \oS \leq \oS_{\max}$, where $\oS_{min}$ is the smallest  local minimum (excluding $0$) and $\oS_{max}$ the largest local  maximum of $\varphi$ (excluding infinity), cf. Figure~\ref{summary}(c).

Based on numerical examples we have made the following observation: When  multistationarity occurs, the function $\varphi$ has one local minima and one local maxima, $\beta_1\leq \beta_2$, resulting in at most three BMSSs. When $\oF_2$  increases, $\beta_1$ increases, while $\beta_2$ decreases.
For some $\oF_2$, $\beta_1=\beta_2$ and  the decreasing part of $\varphi$ is lost and  $\varphi$ becomes increasing. In fact, the value of $\oF_2$ for which $\beta_1=\beta_2$ is $M$. 

All together, this implies that for fixed values of $\oE,\oF_1$, the value of $\oF_2$ determines whether  values of $\oS$  for which multistationarity occurs exist. Figure~\ref{regions} shows how the number of steady states changes with $\oF_2$.

We conclude that:
 \begin{result}[Motif (f)]
Let a two-site phosphorylation cycle be given with one kinase ca\-ta\-lyzing phosphorylation at both sites and two different phosphatases catalyzing dephosphorylation. Further, assume  that positive total amounts $\oE,\oF_1,\oF_2,\oS$ are given. 

Then, the BMSSs satisfy $\oS=\varphi(Y_1)$ for $Y_1\in \Gamma=[0,\xi)$, where $\varphi$ is a $\mmC^1$ function with $\varphi(0)=0$. 

Let  $\Lambda=(1+\eta_2/\delta_1)\mu_1 \oF_1- \oE$. Then we have:
\begin{enumerate}[(i)]
\item If $\Lambda\leq 0$ or  $\oE-\mu_1\oF_1-\mu_2\oF_2\geq 0$ then there is a unique BMSS.
\item If $\Lambda> 0$ then there exists values of $\oF_2$ and $\oS$ such that the system has more than one BMSS. Further,  there is an upper bound to $\oF_2$ for which  multistationarity can occur. This bound is independent of $\oS$.
\end{enumerate}
 \end{result}

\paragraph{Motif (g).}
Next, we consider   a  two-site phosphorylation cycle, where phosphorylation is catalyzed  by the same kinase at both sites and dephosphorylation is catalyzed by the same phosphatase at both sites. Again, we assume sequential phosphorylation.  The chemical reactions of the system are:

\centerline{
\xymatrix{
S_{0} + E \ar@<0.5ex>[r]^(0.6){a^E_1} & X_1  \ar@<0.5ex>[l]^(0.4){b^E_1} \ar[r]^(0.4){c^E_1} & S_{1} + E & 
S_{1} + E \ar@<0.5ex>[r]^(0.6){a^E_2} & X_2  \ar@<0.5ex>[l]^(0.4){b^E_2} \ar[r]^(0.4){c^E_2} & S_{2} + E \\
S_{1} + F \ar@<0.5ex>[r]^(0.6){a^F_1} & Y_1  \ar@<0.5ex>[l]^(0.4){b^F_1} \ar[r]^(0.4){c^F_1} & S_{0} + F &
S_{2} + F \ar@<0.5ex>[r]^(0.6){a^F_2} & Y_2  \ar@<0.5ex>[l]^(0.4){b^F_2} \ar[r]^(0.4){c^F_2} & S_{1} + F
}}
\noindent
The differential equations describing the system  are the following:

\vspace{-0.5cm}
\noindent
\hspace{-0.4cm}
\begin{minipage}[t]{0.62\textwidth}
{\small\begin{align}
\dot{S_{0}} &=  b^E_{1} X_{1} +  c^F_{1} Y_{1} - a^E_{1} E S_{0}   \nonumber   \\
\dot{E} &= (b^E_1 + c^E_1)X_1+(b^E_2 + c^E_2)X_2 - (a^E_1 S_{0} +a^E_2 S_{1})E\nonumber \\
\dot{F} &= (b^F_1 + c^F_1)Y_1 +(b^F_2 + c^F_2)Y_2 - (a^F_1 S_1 +a^F_2 S_2)F   \nonumber \\
\dot{S_{1}} &=  c^E_1 X_1 +b^E_{2} X_{2}+b^F_1 Y_1 +  c^F_{2} Y_{2} - (a^F_1 F+ a^E_{2} E) S_{1}    \label{eq1g4}
\end{align}}
\end{minipage}
\hspace{-0.4cm}
\begin{minipage}[t]{0.4\textwidth}
{\small \begin{align}
\dot{S_{2}} &=  c^E_2 X_2  +b^F_2 Y_2  - a^F_2 F S_2  \label{eq1g3}  \\
\dot{X_k} &=  -(b^E_k + c^E_k)X_k + a^E_k E S_{k-1}  \label{eq1g1}  \\
\dot{Y_k} &=  -(b^F_k + c^F_k)Y_k +a^F_k F S_k  \label{eq1g2} 
\end{align}}
\end{minipage}

\vspace{-0.3cm}
\noindent 
with $k=1,2$. As in the previous system, the conservation laws are given by  
\begin{align}\label{totalamounts2}
\overline{E} &= E + X_{1}+X_2,\quad \overline{F} = F + Y_1+Y_2,\quad  \overline{S}   =   S_{0}+S_1+S_2 + X_{1}+X_2+Y_1+Y_2. \nonumber \end{align}
Therefore, if total amounts are given, the system to be solved consists of 9 equations in 9 variables which are the equations for the total amounts and, for instance, equations \eqref{eq1g4}-\eqref{eq1g2}.
From \eqref{eq1g1} and \eqref{eq1g2} we obtain 
\begin{equation}\label{relations2}
X_k = \eta_k E S_{k-1}, \quad Y_k=\delta_k F S_k,\qquad k=1,2\end{equation}
with constants $\eta_k=a^E_k/(b^E_k + c^E_k)$ and $\delta_k=a^F_k/(b^F_k + c^F_k)$.
From \eqref{eq1g3} and \eqref{eq1g2} with $k=2$, and from \eqref{eq1g4}, \eqref{eq1g3}, \eqref{eq1g2} with $k=1,2$ and \eqref{eq1g1} with $k=2$, we obtain
$$X_k = \mu_k Y_k,\qquad \mu_k= c^F_{k}/ c^E_{k}. $$

If $\oE,\oF > 0$ then  $E,F \neq 0$ at steady state. We isolate $E,F$ from the corresponding conservation laws and write them as functions of $Y_1,Y_2$. Since the denominators are non-zero, we find
$$ S_0 = \frac{\mu_1 Y_1}{\eta_1(\oE - \mu_1Y_1-\mu_2Y_2)},\quad S_2= \frac{Y_2}{\delta_2(\oF-Y_1-Y_2)}.$$ 
From $Y_1=\delta_1 F S_1=\delta_1(\oF-Y_1-Y_2)S_1$ we obtain
$Y_1= \frac{\delta_1 S_1(\oF - Y_2)}{1+\delta_1 S_1}$.
Now from the equality $X_2= \eta_2 E S_1=\eta_2(\oE-\mu_1Y_1-\mu_2Y_2)S_1$, we find that
$$Y_2 = \frac{\eta_2 S_1 (\oE + \delta_1(\oE - \mu_1 \oF)S_1)}{\mu_2 + \mu_2(\delta_1+\eta_2)S_1 + \delta_1\eta_2(\mu_2-\mu_1)S_1^2}. $$ 
We can therefore write all concentrations as functions of $S_1$. Let 
\begin{align*}
p_1 (S_1) &= \oE + \delta_1(\oE - \mu_1 \oF)S_1 \\
p_2(S_1) & = \mu_2\oF + \eta_2(\mu_2 \oF-\oE)S_1 \\
p_3 (S_1) &= \mu_2 + \mu_2(\delta_1+\eta_2)S_1 + \delta_1\eta_2(\mu_2-\mu_1)S_1^2.
\end{align*}
Then we have
$$Y_1 =\frac{\delta_1 S_1 p_2(S_1)}{p_3(S_1)}, \quad Y_2 =\frac{\eta_2 S_1 p_1(S_1)}{p_3(S_1)}, \quad S_0 =\frac{\mu_1 \delta_1 S_1 p_2(S_1)}{\mu_2 \eta_1 p_1(S_1)}, \quad S_2 =\frac{\eta_2 S_1 p_1(S_1)}{\delta_2 p_2(S_1)}. $$ 

This leads to the following relation for the BMSSs
$$ \oS =  \varphi(S_1)= S_{0}+S_1+S_2 + X_{1}+X_2+Y_1+Y_2$$ 
with $\varphi(0)=0$ and where each term is considered a function of $S_1$, as defined above. For this motif, a $\varphi$-function in terms of an intermediate complex cannot be obtained: The relationship between $S_1$ and $Y_1$ or $Y_2$ is in general not invertible.

 We next seek to determine the values of $S_1$ that lead to non-negative concentrations of the species; that is to determine the region $\Gamma$, where $S_1$ is defined. From \eqref{relations2} we see that if $Y_k,S_0,S_1,S_2$ are non-negative, then so are $E, F, X_k$ at steady state. The roots of $p_1,p_2$ are
$$\xi_1= \frac{\oE}{\delta_1 ( \mu_1 \oF - \oE)}, \qquad \xi_2 = \frac{\mu_2 \oF }{ \eta_2 ( \oE - \mu_2 \oF)},$$
 respectively. Since $p_1, p_2$ are in the denominators of $S_0, S_2$, these polynomials are not allowed to vanish.  Let $\Delta_1 = \mu_1 \oF - \oE$ and $\Delta_2 = \mu_2 \oF - \oE$. Then
\begin{enumerate}[(i)]
\item  If $\Delta_2 <0$, then $p_2(S_1)> 0$ if and only if $S_1< \xi_2$.  If $\Delta_2\geq 0$, then $p_2(S_1)> 0$ for all  $S_1\geq 0$.
\item  If $\Delta_1 \leq 0$, then $p_1(S_1)> 0$ for all $S_1>0$. Otherwise,  if $\Delta_1 > 0$, then $p_1(S_1)> 0$ 
if and only if $S_1< \xi_1$.
\end{enumerate}

The polynomial $p_3(S_1)$ has degree 2 and is positive  whenever  $S_1$ is positive and  $\mu_2\geq\mu_1$. If $\mu_2<\mu_1$, then there is exactly one positive root $\xi$. Therefore, we have the following situations:
\begin{enumerate}[A.]
\item If $\mu_2\geq \mu_1$ ($\Delta_2\geq\Delta_1$), we require $p_1(S_1),p_2(S_1)>0$. Three different scenarios occur:
\begin{enumerate}[1)]
\item $\Delta_1>0$ and $\Delta_2>0$, i.e. $\oE < \mu_1 \oF \leq \mu_2 \oF$, $\Gamma=[0,\xi_1)$.
\item $\Delta_1\leq 0$ and $\Delta_2 \geq 0$, i.e. $\mu_1 \oF \leq \oE \leq \mu_2 \oF$, $\Gamma=[0,+\infty)$.
\item $\Delta_1<0$ and $\Delta_2<0$, i.e.  $\mu_1 \oF \leq \mu_2 \oF < \oE $, $\Gamma=[0,\xi_2)$.
\end{enumerate}
\item If $\mu_1> \mu_2$ ($\Delta_1>\Delta_2$), we have
\begin{enumerate}[1)]
\item $\Delta_1>0$ and $\Delta_2\geq 0$, i.e. $\oE \leq  \mu_2 \oF < \mu_1 \oF$, $\Gamma=[0,\widetilde{\xi}_1)$ for $\widetilde{\xi}_1= \min(\xi,\xi_1)$.
\item $\Delta_1\geq 0$ and $\Delta_2<0$, i.e. $\mu_2 \oF < \oE \leq  \mu_1 \oF$, $\Gamma=[0,\widetilde{\xi})\cup (\overline{\xi},+\infty)$ for $\widetilde{\xi}= \min(\xi_1,\xi_2,\xi)$ and $\overline{\xi}=\max(\xi_1,\xi_2,\xi)$.
\item $\Delta_1<0$ and $\Delta_2<0$, i.e. $\mu_2 \oF < \mu_1 \oF < \oE $, $\Gamma=[0,\widetilde{\xi}_2)$ for $\widetilde{\xi}_2= \min(\xi,\xi_2)$.
\end{enumerate}
\end{enumerate}
Observe that in all cases, the function $\varphi$ is $\mmC^1$ and non-negative in $\Gamma$. In the cases A.1, A.3, B.1 and B.3, when $S_1$ approaches  the upper limit of $\Gamma$, then $\varphi$ tends to infinity, since at least one of $Y_1,Y_2,S_0,S_2$ does. In case A.2 the function tends to infinity as $S_1$ tends to infinity. In all these cases, multistationarity would appear in the form of Figure~\ref{summary}(c), that is, $\varphi$ should decrease for some values of $S_1$.

In case B.2, $\Gamma$ is not connected so we let $\Gamma=\Gamma'\cup \Gamma''$. When $S_1$ tends to $\widetilde{\xi}$, then $\varphi$ tends to infinity. It follows that for every $\oS$, there is one BMSS located in $\Gamma'$. Additionally, when $S_1$ approaches $\overline{\xi}$ from the right, $\varphi$ also tends to infinity, implying that $\varphi$ comes down  from infinity in $\Gamma''$. When $S_1$ tends to infinity, $\varphi$ tends to infinity. This implies that in case B.2 the function resembles that in Figure~\ref{summary}(b), potentially with more increasing and decreasing parts. In any case, when  $\oS$ is large, multistationarity occurs and there are at least three steady states. 

Let us consider the derivatives with respect to $S_1$ of the following summands in $\varphi$:
\begin{align*}
\frac{\partial(Y_1+Y_2)}{\partial S_1} &= \frac{\mu_2(\eta_2\oE + \delta_1\mu_2\oF) + 2 \delta_1\eta_2\mu_1(\mu_2-\mu_1)\oF S_1 + \eta_2\delta_1 (\mu_2-\mu_1)\Delta_2 S_1^2}{p_3(S_1)^2} \\
\frac{\partial(\mu_1Y_1+\mu_2Y_2)}{\partial S_1} &= \frac{\mu_2^2(\eta_2\oE + \delta_1\mu_1\oF) + 2 \delta_1\eta_2\mu_2(\mu_2-\mu_1)\oE S_1 - \eta_2\delta_1^2\mu_2 (\mu_2-\mu_1)\Delta_1 S_1^2}{p_3(S_1)^2}
\\
\frac{\partial S_0}{\partial S_1} &= \frac{\delta_1\mu_1\mu_2\oE\ \oF + 2 \eta_2\oE\Delta_2S_1 -\mu_1 \eta_2\delta_1^2\Delta_1\Delta_2 S_1^2}{(\eta_1\mu_2\ p_1(S_1))^2}
\\
\frac{\partial S_2}{\partial S_1} &= \frac{\eta_2\mu_2\oE\ \oF - 2 \eta_2 \delta_1\mu_2 \oF\Delta_1S_1 - \eta_2^2\delta_1\Delta_1\Delta_2 S_1^2}{(\delta_2\ p_2(S_1))^2}
\end{align*}

In case A.2 these derivatives are all positive since $\Delta_1\leq 0$, $\Delta_2 \geq 0$ and $\mu_2-\mu_1\geq 0$.
Hence, there is exactly one steady state. For the remaining cases, one can always find combinations of parameters for which the function $\varphi$ is decreasing for some value of  $S_1$.
Thus, further analysis of the derivatives or the function $\varphi$ itself is required. For example, in case A.3, the following choices of numerical values provide multiple steady states:  $\oF=3$, $\oE=10$, $\delta_1=10$, $\delta_2=100$, $\mu_1=1$, $\mu_2=3$, $\eta_1=0.002$, $\eta_2=100$.

We conclude that:
 \begin{result}[Motif (g)] 
Let a two-site phosphorylation cycle be given with phosphorylation at both sites being catalyzed  by the same kinase and dephosphorylation at both sites being catalyzed by the same phosphatase. Further, assume  that the total amounts $\oE,\oF,\oS$ are positive. Then, any BMSSs satisfy $\oS=\varphi(S_1)$ for $S_1\in \Gamma$ (with $\Gamma$ defined  above)  where $\varphi$ is a $\mmC^1$ function with $\varphi(0)=0$. 

Additionally,
\begin{enumerate}[(i)]
\item If $\mu_1 \oF \leq \oE \leq \mu_2 \oF$, then for any total amount $\oS$ there is a unique BMSS.
\item If $\mu_2 \oF < \oE < \mu_1 \oF$, then there exists a value $\oS_{min}$ such that for any $\oS > \oS_{min}$ the system has at least three BMSSs.
\end{enumerate}
 \end{result}

\subsection{Modification of two different substrates}
\label{motiftwodiff}

\paragraph{Motif  (h).}
In this system, two cycles are connected through a joint catalyzing kinase. The chemical reactions of the system are:

\centerline{
\xymatrix{
S_{0} + E \ar@<0.5ex>[r]^(0.6){a^E_1} & X_1  \ar@<0.5ex>[l]^(0.4){b^E_1} \ar[r]^(0.4){c^E_1} & S_{1} + E & 
P_{0} + E \ar@<0.5ex>[r]^(0.6){a^E_2} & X_2  \ar@<0.5ex>[l]^(0.4){b^E_2} \ar[r]^(0.4){c^E_2} & P_{1} + E \\
S_{1} + F_1 \ar@<0.5ex>[r]^(0.6){a^F_1} & Y_1  \ar@<0.5ex>[l]^(0.4){b^F_1} \ar[r]^(0.4){c^F_1} & S_{0} + F_1 &
P_{1} + F_2 \ar@<0.5ex>[r]^(0.6){a^F_2} & Y_2  \ar@<0.5ex>[l]^(0.4){b^F_2} \ar[r]^(0.4){c^F_2} & P_{0} + F_2
}}

\noindent
The differential equations describing the system  are the following:

\vspace{-0.5cm}
\noindent
\begin{minipage}[t]{0.55\textwidth}
{\small \begin{align}
\dot{S_{0}} &=  b^E_{1} X_{1}+  c^F_{1} Y_{1}   - a^E_{1} E S_{0}   \nonumber   \\
\dot{S_{1}} &=  b^F_{1} Y_{1} +  c^E_{1} X_{1}  - a^F_{1} F_1 S_{1}   \label{cherry5}  \\
\dot{X_1} &= -(b^E_1 + c^E_1)X_1 + a^E_1 E S_0  \label{cherry1}  \\
\dot{Y_1} &= -(b^F_1 + c^F_1)Y_1 + a^F_1 F_1 S_1  \label{cherry3}  \\
\dot{F_{1}} &= (b^F_1 + c^F_1)Y_1 - a^F_1 F_1 S_1   \nonumber \\
\dot{E} &= (b^E_1 + c^E_1)X_1+(b^E_2 + c^E_2)X_2 - (a^E_1 S_{0} +a^E_2 P_{0})E.\nonumber
\end{align}}
\end{minipage}
\hspace{-0.1cm}
\begin{minipage}[t]{0.4\textwidth}
{\small \begin{align}
\dot{P_{0}} &=  b^E_{2} X_{2} +  c^F_{2} Y_{2}   - a^E_{2} E P_{0}  \nonumber   \\
\dot{P_{1}} &=  b^F_{2} Y_{2} +  c^E_{2} X_{2}  - a^F_{2} F_2 P_{1}  \label{cherry6}   \\
\dot{X_2} &= -(b^E_2 + c^E_2)X_2 + a^E_2 E P_0   \label{cherry2}  \\
\dot{Y_2} &= -(b^F_2 + c^F_2)Y_2 + a^F_2 F_2 P_1   \label{cherry4} \\
\dot{F_{2}} &= (b^F_2 + c^F_2)Y_2 - a^F_2 F_2 P_1  \nonumber 
\end{align}}
\end{minipage}

\vspace{-0.3cm}
The conservation laws are given by  
\begin{align}\label{totalamounts_cherry}
\overline{E} &= E + X_{1}+X_2,\quad \overline{F}_k = F_{k} + Y_k,\quad  \overline{S}   =   S_0+S_1+ X_1+Y_1,
\quad  \overline{P}   =   P_0+P_1 +X_2+Y_2, \nonumber 
\end{align}
with $k=1,2$. Therefore, if total amounts are given, the system to be solved consists of 11 equations in 11 variables which are the equations for the total amounts (5 equations) and, for instance, equations \eqref{cherry5}-\eqref{cherry4}. From  \eqref{cherry1}, \eqref{cherry3}, \eqref{cherry2} and \eqref{cherry4} we obtain 
\begin{equation*} 
X_1 = \eta_1 E S_{0}, \quad X_2 = \eta_2 E P_{0}, \quad Y_1=\delta_1 F_1 S_1, \quad Y_2=\delta_2 F_2 P_1,
\end{equation*}
with constants $\eta_k=a^E_k/(b^E_k + c^E_k)$ and $\delta_k=a^F_k/(b^F_k + c^F_k)$.
From these equations, \eqref{cherry5} and \eqref{cherry6} we obtain
$$X_k = \mu_k Y_k,\qquad \mu_k= c^F_{k}/ c^E_{k}. $$
If $\oE,\oF_k> 0$, then at steady state $E,F_k\neq 0$. Since $F_k=\oF -Y_k$ we require $0\leq Y_k<\oF_k$ for any BMSS.
Using the conservation laws for $\oP$ and $\oF_2$ we have
$$\oP= \frac{\mu_2 Y_2}{\eta_2 E}  + \frac{Y_2}{\delta_2(\oF_2-Y_2)}+\mu_2 Y_2+Y_2,$$
and it follows that $E$ is an increasing $\mmC^1$ function of $Y_2$
$$E=g(Y_2)=\frac{\mu_2 \delta_2 Y_2(\oF_2-Y_2)}{\eta_2(\delta_2\oP(\oF_2-Y_2)-Y_2-\delta_2(\mu_2+1)Y_2(\oF_2-Y_2))}$$
The numerator is non-negative for $0\leq  Y_2<\oF_2$. The denominator is a degree two polynomial in $Y_2$ with positive independent and leading coefficients. For $Y_2=\oF_2$, the denominator is negative and thus has two positive real roots. It follows that for $E$ to be non-negative we require  $Y_2\in [0,\xi_2)$ where $\xi_2<\oF_2$  is the first positive root of the denominator. We have that $g(0)=0$ and  that $g(Y_2)$ goes to infinity as $Y_2$ tends to $\xi_2$. To sum up, for $0\leq Y_2<\xi_2$, the steady state values of $F_2,P_0,P_1,X_2,E$ are non-negative as well.

Using the conservation law for $\oE$, it follows that
\begin{equation}\label{totEmotifh}
\oE=E+X_1+X_2=g(Y_2)+\mu_1 Y_1+\mu_2 Y_2,
\end{equation}
and $Y_2$ is a decreasing $\mmC^1$ function of $Y_1$, $h(Y_1)=Y_2$, defined on $[0,\oE/\mu_1]$ such that $h(\oE/\mu_1)=0$ and $h(0)<\xi_2$. Consequently,  $g(h(Y_1))=E$ is a decreasing $\mmC^1$  function defined on $[0,\oE/\mu_1]$ and since $E>0$, we require $Y_1\in [0,\oE/\mu_1)$. 

Let $ \Gamma=[0,\xi_1)$  with $\xi_1=\min(\oF_1,\oE/\mu_1)$ and consider the conservation law for $\oS$. We obtain
$$\oS= \varphi(Y_1) = \frac{\mu_1 Y_1}{\eta_1 g(h(Y_1))}  + \frac{Y_1}{\delta_1(\oF_1-Y_1)}+\mu_1 Y_1+Y_1.$$
Then, $\varphi$ is an increasing $\mmC^1$  function defined on $ \Gamma$ with image $\overline{\R}_+$ and all steady state concentrations are non-negative if and only if  $Y_1\in \Gamma$. Therefore, we have shown:

\begin{result}[Motif (h)]\label{motif1h} 
Let two one-site modification cycles with joint kinase and distinct phosphatases be given. Further, assume that the total amounts $\oS,\oP,\oE,\oF_k$, $k=1,2$, are positive. Then, the system has a  unique BMSS. 

Specifically, the BMSS satisfies $\oS=\varphi(Y_1)$ for $Y_1\in \Gamma=[0,\xi_1)$, where
$$\oS=\varphi(Y_1)=  \frac{\mu_1 Y_1}{\eta_1 g(h(Y_1))}  + \frac{Y_1}{\delta_1(\oF_1-Y_1)}+\mu_1 Y_1+Y_1$$
 is an increasing $\mmC^1$ function which tends to infinity as $Y_1$ tends to $\xi_1$ and fulfills $\varphi(0)=0$.
 \end{result}

\begin{obs}
In equation \eqref{totEmotifh}, if we isolate $Y_1$ instead of $Y_2$, we would obtain a relationship between $Y_1$ and $Y_2$, $Y_1=\overline{h}(Y_2)$ and the steady state relation in terms of $Y_2$, $\oS = \psi(Y_2)$. In this case, $Y_2\in (\xi,\xi')$ where $\xi=\overline{h}(\xi_1)$ and $\xi'$ is the pre-image of  $\oE$  of the increasing function $\mu_2Y_2+g(Y_2)$. If $\xi_1=\oE/\mu_1$, then $\xi=0$. 
\end{obs}

\paragraph{Motif (i).}
In this system, two cycles are connected through a joint catalyzing kinase and a joint catalyzing phosphatase. The chemical reactions of the system are:

\centerline{
\xymatrix{
S_{0} + E \ar@<0.5ex>[r]^(0.6){a^E_1} & X_1  \ar@<0.5ex>[l]^(0.4){b^E_1} \ar[r]^(0.4){c^E_1} & S_{1} + E & 
P_{0} + E \ar@<0.5ex>[r]^(0.6){a^E_2} & X_2  \ar@<0.5ex>[l]^(0.4){b^E_2} \ar[r]^(0.4){c^E_2} & P_{1} + E \\
S_{1} + F \ar@<0.5ex>[r]^(0.6){a^F_1} & Y_1  \ar@<0.5ex>[l]^(0.4){b^F_1} \ar[r]^(0.4){c^F_1} & S_{0} + F &
P_{1} + F \ar@<0.5ex>[r]^(0.6){a^F_2} & Y_2  \ar@<0.5ex>[l]^(0.4){b^F_2} \ar[r]^(0.4){c^F_2} & P_{0} + F
}}
\noindent
The differential equations describing the system  are the following:

\vspace{-0.5cm}
\noindent
\hspace{-0.4cm}
\begin{minipage}[t]{0.6\textwidth}
{\small \begin{align}
\dot{S_{0}} &=  b^E_{1} X_{1}+  c^F_{1} Y_{1} - a^E_{1} E S_{0}    \nonumber   \\
\dot{S_{1}} &=  b^F_{1} Y_{1} +  c^E_{1} X_{1} - a^F_{1} F S_{1}    \label{GMcherry5}  \\
\dot{X_1} &= -(b^E_1 + c^E_1)X_1 + a^E_1 E S_0  \label{GMcherry1}  \\
\dot{Y_1} &= -(b^F_1 + c^F_1)Y_1 + a^F_1 F S_1  \label{GMcherry3}  \\
\dot{F} &= (b^F_1 + c^F_1)Y_1 + (b^F_2 + c^F_2)Y_2  - (a^F_1 S_1+a^F_2 P_1)F   \nonumber \\
\dot{E} &= (b^E_1 + c^E_1)X_1+(b^E_2 + c^E_2)X_2 - (a^E_1 S_{0} +a^E_2 P_{0})E.\nonumber
\end{align}}
\end{minipage}
\hspace{-0.1cm}
\begin{minipage}[t]{0.4\textwidth}
{\small \begin{align}
\dot{P_{0}} &=  b^E_{2} X_{2} +  c^F_{2} Y_{2} - a^E_{2} E P_{0}    \nonumber   \\
\dot{P_{1}} &=  b^F_{2} Y_{2} +  c^E_{2} X_{2} - a^F_{2} F P_{1}   \label{GMcherry6}   \\
\dot{X_2} &= -(b^E_2 + c^E_2)X_2 + a^E_2 E P_0   \label{GMcherry2}  \\
\dot{Y_2} &= -(b^F_2 + c^F_2)Y_2 + a^F_2 F P_1   \label{GMcherry4} 
\end{align}}
\end{minipage}

\vspace{-0.3cm}
The conservation laws are given by  
\begin{align}\label{totalamounts_cherry}
\overline{E} &= E + X_{1}+X_2,\; \overline{F} = F+Y_1+Y_2,\;  \overline{S}   =   S_0+S_1+ X_1+Y_1,
\;  \overline{P}   =   P_0+P_1 +X_2+Y_2. \nonumber 
\end{align}
If total amounts are given, the system to be solved consists of 10 equations in 10 variables which are chosen to be the equations for the total amounts (4 equations) and  equations \eqref{GMcherry5}-\eqref{GMcherry4}. As usual we derive
\begin{equation*} 
X_1 = \eta_1 E S_{0}, \quad X_2 = \eta_2 E P_{0}, \quad Y_1=\delta_1 F S_1, \quad Y_2=\delta_2 F P_1,\quad X_k = \mu_k Y_k,
\end{equation*}
$k=1,2$, with constants $\eta_k=a^E_k/(b^E_k + c^E_k)$, $\delta_k=a^F_k/(b^F_k + c^F_k)$ and $ \mu_k= c^F_{k}/ c^E_{k}$.

Let $\xi_k=\min(\oF,\oE/\mu_k)$, $k=1,2$.
For fixed $Y_2\in [0,\xi_2)$, the equations above provide the steady state equations corresponding to  a one-site cycle with species $E,F,S_0,S_1,X_1,Y_1$ and total amounts $\oS,\oE-\mu_2Y_2, \oF-Y_2$. Analogously, for fixed $Y_1\in [0,\xi_1)$, we obtain the equations for a one-site cycle with species $E,F,P_0,P_1,X_2,Y_2$ and total amounts $\oP,\oE-\mu_1Y_1, \oF-Y_1$. Therefore, using Result \ref{motif1a} we obtain that the steady states are solutions to the system
\begin{align*}
\oS &= \varphi_1(Y_1,Y_2)=   \frac{\mu_1 Y_1}{\eta_1 (\oE-\mu_1Y_1-\mu_2Y_2)}  + \frac{Y_1}{\delta_1(\oF-Y_1-Y_2)}+\mu_1 Y_1+Y_1 \\
\oP & = \varphi_2(Y_1,Y_2)=  \frac{\mu_2 Y_2}{\eta_2 (\oE-\mu_1Y_1-\mu_2Y_2)}  + \frac{Y_2}{\delta_2(\oF-Y_1-Y_2)}+\mu_2 Y_2+Y_2
\end{align*}
for any $Y_1,Y_2$ and for the solution to be a BMSS we require that $Y_1+Y_2<\oF$ and $\mu_1Y_1+\mu_2Y_2<\oE$. All species concentrations derived with $Y_1,Y_2$ fulfilling these conditions are non-negative.

Note that $\varphi_1,\varphi_2$ are increasing functions of both $Y_1$ and $Y_2$. Therefore, using $\varphi_1$ for a fixed $\oS$, there is a decreasing $\mmC^1$ function $f(Y_2)=Y_1$, defined for $Y_2\in [0,\xi_2)$. Further, since $f(Y_2)$ is the steady state value of $Y_1$ in the first cycle with total amounts $\oS,\oE-\mu_2Y_2, \oF-Y_2$, the derived concentrations $E,F,S_0,S_1,X_1$ are also non-negative. Thus, so are $P_0,P_1,X_2$ for $Y_2\in [0,\xi_2)$.

Let $\Gamma=[0,\xi_2)$ so that all concentrations are non-negative if and only if $Y_2\in \Gamma$. We conclude that the steady states of the system are described by a $\mmC^1$ function 
$$\oP = \varphi(Y_2) = \varphi_2(f(Y_2),Y_2)$$
defined for $Y_2\in \Gamma$.

The behavior of the function  $\varphi$ determines the presence/absence of multiple steady states. Note that when $Y_2$ tends to $\xi_2$, then $Y_1=f(Y_2)$ tends to zero and $\varphi$ tends to infinity, implying that $\varphi$ increases as we approach the upper limit of $\Gamma$. Additionally, at $Y_2=0$, $f(Y_2)$ is finite, $\varphi(0)=0$ and hence the function is increasing at $0$ as well (it is positive for $Y_2>0$). We conclude that the existence of at least one steady state is guaranteed and multistationarity occurs if  $\varphi'(Y_2)<0$ for some $Y_2\in \Gamma$. 

To proceed, let $\psi=(\varphi_1,\varphi_2):\R^2\rightarrow \R^2$ and let $J_{\psi}$ be the Jacobian matrix of $\psi$, that is, 
$$J_{\psi} = \left( \begin{array}{cc} \frac{\partial \varphi_1}{\partial Y_1} &  \frac{\partial \varphi_1}{\partial Y_2} \\ 
 \frac{\partial \varphi_2}{\partial Y_1} &  \frac{\partial \varphi_2}{\partial Y_2} \end{array}\right).$$ 
Then, a simple observation  (proved in Proposition \ref{jacobian_det}) shows that $\varphi'(Y_2)<0$ if and only if $\det(J_{\psi}(f(Y_2),Y_2))<0$. Therefore, if $\det(J_{\psi}(Y_1,Y_2))<0$ for some values  $\oE,\oF,Y_1,Y_2$, then for total amounts $\oS=\varphi_1(Y_1,Y_2)$ and $\oP=\varphi_2(Y_1,Y_2)$, the system exhibits multistationarity.

\begin{proposition}\label{motif1i_prop}
Let $\sigma=(\mu_1-\mu_2)( \mu_1\delta_1\eta_2 - \mu_2\delta_2\eta_1)$ and $\overline{\Gamma}=\{(Y_1,Y_2)\in \R^2|\ Y_1+Y_2<\oF,\  \mu_1Y_1+\mu_2Y_2<\oE\}$. 
\begin{enumerate}[(i)]
\item If $\sigma\geq 0$, then $\det(J_{\psi}(Y_1,Y_2))>0$ for all $(Y_1,Y_2)\in \overline{\Gamma}$.
 \item Assume that $\sigma< 0$. If (a) $\mu_1-\mu_2>0$, and either $\mu_2\oF \geq \oE $ or $\oE \geq \mu_1 \oF$, or (b) $\mu_2-\mu_1>0$, and either $\mu_1\oF \geq \oE$ or $\oE \geq \mu_2 \oF$, then $\det(J_{\psi}(Y_1,Y_2))>0$ for all $(Y_1,Y_2)\in \overline{\Gamma}$.
 \item If $\sigma< 0$ and either $\mu_1\oF > \oE  >\mu_2 \oF$ or  $\mu_2\oF > \oE  >\mu_1 \oF$, then there exist values $(Y_1,Y_2)\in \overline{\Gamma}$, such that $\det(J_{\psi}(Y_1,Y_2))<0$.
\end{enumerate}
\end{proposition}
The proof  of this proposition is found in Appendix \ref{B}. Using it, we derive the following result:

\begin{result}[Motif (i)]\label{motif1i} 
Let two one-site modification cycles with joint kinase and joint phosphatase be given. Further, assume that the total amounts $\oS,\oP,\oE,\oF$ are positive. Then, the BMSSs satisfy $\oP=\varphi(Y_2)$ for $Y_2\in \Gamma$, where $\varphi$ is a $\mmC^1$ function which tends to infinity as $Y_2$ tends to $\xi_2$ and fulfills $\varphi(0)=0$.

Let  $\sigma=(\mu_1-\mu_2)( \mu_1\delta_1\eta_2 - \mu_2\delta_2\eta_1)$. Then:
\begin{itemize}
\item The function $\varphi$ is always increasing if either
$(i)$ $\sigma\geq 0$ or $(ii)$  $\sigma< 0$ and either (a) $\mu_1-\mu_2>0$, together with $\mu_2\oF \geq \oE $ or $\oE \geq \mu_1 \oF$, or  (b) $\mu_2-\mu_1>0$, together with $\mu_1\oF \geq \oE$ or $\oE \geq \mu_2 \oF$.
\item If $\sigma< 0$ and either $\mu_1\oF > \oE  >\mu_2 \oF$ or  $\mu_2\oF > \oE  >\mu_1 \oF$, then, there exist values $Y_2\in \Gamma$ for which $\varphi'(Y_2)<0$. Hence multistationarity occurs.  In this case  the total amounts  $\oS,\oP$ are required to be large.

\end{itemize}
 \end{result}

\subsection{Cascade motifs}
\label{motifcascade}

\paragraph{Motif (j).} 
We consider here the combination of two one-site modification cycles  in a cascade motif with a specific phosphatase acting in each layer. The chemical reactions of the system are:
 
 \centerline{
\xymatrix{
S_{0} + E \ar@<0.5ex>[r]^(0.6){a^E_1} & X_1  \ar@<0.5ex>[l]^(0.4){b^E_1} \ar[r]^(0.4){c^E_1} & S_{1} + E & 
P_{0} + S_1 \ar@<0.5ex>[r]^(0.6){a^E_2} & X_2  \ar@<0.5ex>[l]^(0.4){b^E_2} \ar[r]^(0.4){c^E_2} & P_{1} + S_1 \\
S_{1} + F_1 \ar@<0.5ex>[r]^(0.6){a^F_1} & Y_1 \ar@<0.5ex>[l]^(0.4){b^F_1} \ar[r]^(0.4){c^F_1} & S_{0} + F_1 &
P_{1} + F_2 \ar@<0.5ex>[r]^(0.6){a^F_2} & Y_2  \ar@<0.5ex>[l]^(0.4){b^F_2} \ar[r]^(0.4){c^F_2} & P_{0} + F_2
}}
\noindent
The differential equations describing the  system  are the following:

\vspace{-0.3cm}
\noindent
\begin{minipage}[t]{0.33\textwidth}
{\small \begin{align}
\dot{S_{0}} &=  b^E_{1} X_{1} +  c^F_{1} Y_{1} - a^E_{1} E S_{0}    \nonumber   \\
\dot{P_{0}} &=  b^E_{2} X_{2}  +  c^F_{2} Y_{2} - a^E_{2} S_1 P_{0}   \nonumber   \\
\dot{E} &= (b^E_1 + c^E_1)X_1- a^E_1 E S_{0} \nonumber \\
\dot{F_1} &= (b^F_1 + c^F_1)Y_1 - a^F_1 F_1 S_1    \nonumber \\
\dot{F_2} &=  (b^F_2 + c^F_2)Y_2 - a^F_2 F_2 P_1    \nonumber 
\end{align}}
\end{minipage}
\begin{minipage}[t]{0.65\textwidth}
{\small \begin{align}
\dot{X_1} &= a^E_1 E S_{0} -(b^E_1 + c^E_1)X_1  \label{eq1jj1} \\
\dot{X_2} &= a^E_2 S_1 P_{0} -(b^E_2 + c^E_2)X_2  \label{eq1jj4} \\
\dot{Y_1} &=a^F_1 F_1 S_1 -(b^F_1 + c^F_1)Y_1   \label{eq1jj2}  \\
\dot{Y_2} &=a^F_2 F_2 P_1 -(b^F_2 + c^F_2)Y_2  \label{eq1jj6} \\
\dot{P_{1}} &=  b^F_{2} Y_{2} +  c^E_{2} X_{2} - a^F_{2} F_2 P_{1}  \label{eq1jj3}  \\
\dot{S_{1}} &=  c^E_1 X_1  +b^F_1 Y_1  + (b^E_2 + c^E_2)X_2- (a^F_1 F_1 +a^E_2 P_{0})S_1. \label{eq1jj5}   
\end{align}}\end{minipage}

\vspace{0.3cm}
The conservation laws are given by  
$$
\overline{E} = E + X_{1},\quad \overline{F}_k = F_k + Y_k,\quad   \overline{S}   =   S_{0}+S_1+ X_{1}+X_2+Y_1,\quad \overline{P}= P_0+P_1+X_2+Y_2$$
for $k=1,2$. 
If total amounts are given, the system to be solved consists of 11 equations in 11 variables which are the equations for the total amounts and, for instance, equations \eqref{eq1jj1}-\eqref{eq1jj5}. As usual we obtain 
$$ X_1 = \eta_1 E S_{0}, \quad X_2 = \eta_2 S_1 P_{0}, \quad Y_1=\delta_1 F_1 S_1,\quad Y_2=\delta_2 F_2 P_1$$
with constants $\eta_k=a^E_k/(b^E_k + c^E_k)$ and $\delta_k=a^F_k/(b^F_k + c^F_k)$.
Equations \eqref{eq1jj6},\eqref{eq1jj3} ($k=2$) and \eqref{eq1jj4},\eqref{eq1jj2},\eqref{eq1jj5}, ($k=1$)  provide the  relations
$$X_k = \mu_k Y_k,\qquad \mu_k= c^F_{k}/ c^E_{k}$$
for $k=1,2$.

The concentrations $E,F_1,F_2,X_1,X_2$ are solved in terms of $Y_1,Y_2$ from the total amounts $\oE,\oF_1,\oF_2$ and the two relations above. If $\oE,\oF_k> 0$, then $E,F_k\neq  0$ at steady state. Hence, any BMSS satisfies $0\leq Y_k<\oF_k$, $0\leq Y_1<\oE/\mu_1$. 

Let $\xi_1 = \min(\oF_1,\oE/\mu_1)$. Then we find  
$$S_0=\frac{\mu_1 Y_1}{\eta_1(\oE -\mu_1 Y_1)} \quad\textrm{and}\quad  S_1=\frac{Y_1}{\delta_1(\oF_1 - Y_1)},$$
which are non-negative, increasing, continuous functions of  $Y_1\in [0,\xi_1)$. It follows that 
$$\varphi_1(Y_1)= S_{0}+S_1+ X_{1}+Y_1$$
 is a non-negative, increasing, continuous function of $Y_1\in [0,\xi_1)$. Also, it tends to infinity as $Y_1$ tends to $\xi_1$ and thus the image of $\varphi_1$ over $[0,\xi_1)$ is $\overline{R}_+$. From the conservation law for  $\oS$ we find
$$ Y_2 = f(Y_1)= \frac{1}{\mu_2} (\oS - \varphi_1(Y_1)), $$ 
which is a decreasing function in $Y_1\in[0,\xi_1)$. The concentration $Y_2$ is non-negative provided $\varphi_1(Y_1)\leq \oS$. Let $\xi_2= \varphi_1^{-1}(\oS)$ ($\varphi_1$ is increasing, hence invertible). Then $Y_2,S_0,S_1,X_1$ are all functions of $Y_1$ and non-negative provided $Y_1\in [0,\xi_2]$. We have that $f(\xi_2)=0$ and $f(0)=\oS/\mu_2$. It follows from the inverse function theorem that $f$ can be inverted so that there exists a continuous decreasing function $Y_1=g(Y_2)$ for $Y_2\in [0,\oS/\mu_2]$ with $g(0)=\xi_2$ and $g(\oS/\mu_2)=0$.  Then $S_0,S_1,X_1,Y_1$ are all functions of $Y_2$ and non-negative provided $Y_2\in  [0,\oS/\mu_2]$.
Note that if $\oS>0$ then $S_1=0$ is not a solution of the steady state equations and hence any BMSS satisfies $S_1>0$.
Since $S_1=g(Y_2)/(\delta_1(\oF_1-g(Y_2)))$ we require  $Y_2<\oS/\mu_2$.

Let $\xi= \min(\oF_2,\oS/\mu_2)$ and $\Gamma=[0,\xi)$. 
Next we find the relations
$$P_0 = \frac{\mu_2 Y_2}{\eta_2S_1},\quad \textrm{and},\quad P_1=\frac{Y_2}{\delta_2(\oF_2 - Y_2)},$$
which are non-negative continuous functions of $Y_2\in \Gamma$. Since $S_1$ is increasing in $Y_1$ and thus decreasing in $Y_2$, we have that both $P_0,P_1$ are increasing in $Y_2$. 

We have seen that all concentrations at steady state are non-negative if and only if $Y_2\in \Gamma$. Further, when $Y_2$ tends to $\xi$, either $P_0$ or $P_1$ tend to infinity. Using the conservation law for $\oP$ we obtain:

\begin{result}[Motif (j)]\label{motif1j} 
Let a cascade of one-site modification cycles   be given with positive  total amounts $\oS,\oP,\oE,\oF_1,\oF_2$. Then the system has a  unique BMSS. 
Specifically, the BMSS satisfies $\oP=\varphi(Y_2)$ for $Y_2$ in $\Gamma= [0,\xi)$, where
$$\oP=\varphi(Y_2)=\frac{\delta_1 \mu_2Y_2(\oF_1-g(Y_2))}{g(Y_2)}+\frac{Y_2}{\delta_2(\oF_2 - Y_2)}+(1+\mu_2)Y_2$$
is an increasing  $\mmC^1$ function, which tends to infinity as $Y_2$ tends to $\xi$ and fulfills $\varphi(0)=0$.
  \end{result}

\paragraph{Motif (k).}
We consider here the combination of two one-site modification cycles  in a cascade motif where the phosphatase is not layer specific; that is  the same phosphatase acts in both layers.  The chemical reactions of the system are:
 
 \centerline{
\xymatrix{
S_{0} + E \ar@<0.5ex>[r]^(0.6){a^E_1} & X_1  \ar@<0.5ex>[l]^(0.4){b^E_1} \ar[r]^(0.4){c^E_1} & S_{1} + E & 
P_{0} + S_1 \ar@<0.5ex>[r]^(0.6){a^E_2} & X_2  \ar@<0.5ex>[l]^(0.4){b^E_2} \ar[r]^(0.4){c^E_2} & P_{1} + S_1 \\
S_{1} + F \ar@<0.5ex>[r]^(0.6){a^F_1} & Y_1 \ar@<0.5ex>[l]^(0.4){b^F_1} \ar[r]^(0.4){c^F_1} & S_{0} + F &
P_{1} + F \ar@<0.5ex>[r]^(0.6){a^F_2} & Y_2  \ar@<0.5ex>[l]^(0.4){b^F_2} \ar[r]^(0.4){c^F_2} & P_{0} + F
}}
\noindent
The differential equations describing the  system  are the following:

\vspace{-0.4cm}
\noindent
\hspace{-0.4cm}
\begin{minipage}[t]{0.63\textwidth}
{\small \begin{align}
\dot{S_{0}} &=  b^E_{1} X_{1}+  c^F_{1} Y_{1}  - a^E_{1} E S_{0}   \nonumber   \\
\dot{P_{0}} &=  b^E_{2} X_{2} +  c^F_{2} Y_{2}  - a^E_{2} S_1 P_{0}   \nonumber   \\
\dot{E} &= (b^E_1 + c^E_1)X_1- a^E_1 E S_{0} \nonumber \\
\dot{F} &= (b^F_1 + c^F_1)Y_1  + (b^F_2 + c^F_2)Y_2 - (a^F_1 S_1+a^F_2 P_1)F    \nonumber \\
\dot{S_{1}} &=  c^E_1 X_1  +b^F_1 Y_1  + (b^E_2 + c^E_2)X_2- (a^F_1 F +a^E_2 P_{0})S_1 \label{eq1j5}  
\end{align}}
\end{minipage}
\hspace{-0.2cm}
\begin{minipage}[t]{0.38\textwidth}
{\small \begin{align}
\dot{X_1} &= a^E_1 E S_{0} -(b^E_1 + c^E_1)X_1  \label{eq1j1} \\
\dot{X_2} &= a^E_2 S_1 P_{0} -(b^E_2 + c^E_2)X_2  \label{eq1j4} \\
\dot{Y_1} &=a^F_1 F S_1 -(b^F_1 + c^F_1)Y_1   \label{eq1j2}  \\
\dot{Y_2} &=a^F_2 F P_1 -(b^F_2 + c^F_2)Y_2  \label{eq1j6} \\
\dot{P_{1}} &=  b^F_{2} Y_{2} +  c^E_{2} X_{2} - a^F_{2} F P_{1}.  \label{eq1j3}  
\end{align}}\end{minipage}

\vspace{-0.3cm}
The conservation laws are given by  
$$
\overline{E} = E + X_{1},\quad \overline{F} = F + Y_1+Y_2,\quad   \overline{S}   =   S_{0}+S_1+ X_{1}+X_2+Y_1,\quad \overline{P}= P_0+P_1+X_2+Y_2.$$
If total amounts are given, the system to be solved consists of 10 equations in 10 variables which are the equations for the total amounts and, for instance, equations \eqref{eq1j5}-\eqref{eq1j3}. From the latter equations, we obtain 
$$ X_1 = \eta_1 E S_{0}, \quad X_2 = \eta_2 S_1 P_{0}, \quad Y_1=\delta_1 F S_1,\quad Y_2=\delta_2 F P_1$$
with constants $\eta_k=a^E_k/(b^E_k + c^E_k)$ and $\delta_k=a^F_k/(b^F_k + c^F_k)$.
Equations \eqref{eq1j6},\eqref{eq1j3} ($k=2$) and \eqref{eq1j5},\eqref{eq1j4} and \eqref{eq1j2} ($k=1$)  provide
$$X_k = \mu_k Y_k,\qquad \mu_k= c^F_{k}/ c^E_{k}$$
for $k=1,2$.

 Note  that if $\oE,\oF>0$, then $E,F\neq  0$ at steady state. Therefore, any BMSS satisfies $0\leq Y_k< \oF$ and $Y_1< \oE/\mu_1$. If $S_1=0$, then it follows from the steady state equations that $S_0=X_1=X_2=Y_1=0$ and thus $\oS=0$. Hence, if $\oS>0$ we have that $S_1>0$ and so $Y_2=X_2/\mu_2<\oS/\mu_2$ for any BMSS.

Let $\xi_2=\min(\oS/\mu_2,\oF)$, $\xi_1=\min(\oE/\mu_1,\oF)$ and $\Gamma_k=[0,\xi_k)$, $k=1,2$. 
For  fixed $Y_2\in \Gamma_2$, let $\xi_1(Y_2)=\min(\oE/\mu_1,\oF-Y_2)$ and note that $\xi_1(Y_2)\leq \xi_1$.  For $Y_2$ fixed,   the steady states  satisfy the steady state equations of a  one-site phosphorylation cycle  with species $S_0,S_1,X_1,Y_1,E,F$ and positive total amounts $\oE$, $\oF-Y_2$ and $\oS - \mu_2Y_2$ (these are independent of $\oP$).  Using Result \ref{motif1a} with $\varphi_1(Y_1,Y_2)$ denoting $\varphi(Y_1)$ for the fixed $Y_2$, the BMSSs satisfy the relation
\begin{equation*}
\oS  = \Phi(Y_1,Y_2)=\varphi_1(Y_1,Y_2)+\mu_2Y_2 
\end{equation*}
 for  $0\leq  Y_1 < \xi_1(Y_2)\leq \xi_1$. 
  Under these conditions, the concentrations $S_0,S_1,E,X_1,Y_1$, $X_2,Y_2$, $F$ are non-negative. Note that 
 $$S_0=\frac{\mu_1Y_1}{\eta_1(\oE-\mu_1Y_1)},\quad S_1=\frac{Y_1}{\delta_1(\oF-Y_1-Y_2)},$$
 (with non-zero denominators)  and $\varphi_1(Y_1,Y_2)=S_0+S_1+X_1+Y_1$, such that  only $S_1$ depends on $Y_2$.  Since the BMSS is unique in the one-site cycle, it follows that for every $Y_2\in \Gamma_2$ there exists $Y_1\in \Gamma_1$ ($Y_1$ is non-zero provided $\oS> 0$) satisfying $\oS=\Phi(Y_1,Y_2)$. Thus, there is a function $f$ so that  for any BMSS,
  \begin{equation}\label{fprop4}
  Y_1=f(Y_2).
  \end{equation}
 The function $f$ is $\mmC^1$ and decreasing for $Y_2\in \Gamma_2$. Indeed, this follows from the implicit function theorem and the fact that $\Phi$ is  $\mmC^1$ in $\Gamma_2$, the derivative with respect to $Y_1$ is  positive  and the derivative with respect to $Y_2$ is positive too. This can be checked by direct computation.

By construction, we have $0<f(Y_2)<\xi_1(Y_2)$  for every $Y_2$. Using Remark \ref{poly}, the value $Y_1=f(Y_2)$  is the first positive root of the polynomial in $Y_1$ (for fixed $Y_2$) obtained from $\Phi(Y_1,Y_2)$ by elimination of denominators. Also, if $Y_2=0$ then $Y_1=f(0)$ is the BMSS value of a one-site phosphorylation system with total amounts $\oS,\oE,\oF$ and hence it is  positive.
 
Consider now the conservation law for $\oP$:
 $$P_0=\frac{\mu_2Y_2}{\eta_2 S_1}=\frac{\mu_2\delta_1(\oF-Y_1-Y_2) Y_2}{\eta_2 Y_1}, \qquad P_1=\frac{Y_2}{\delta_2 F}=\frac{Y_2}{\delta_2 (\oF-Y_1-Y_2)},$$
 and hence $\oP=\varphi_2(Y_1,Y_2)$ is   a function of $Y_1$ and $Y_2$, which is $\mmC^1$ with respect to each variable. Further, each summand $P_0,P_1,X_2$ and $Y_2$ of $\oP$ is non-negative if $Y_2\in \Gamma_2$ and so all concentrations at steady state are non-negative if and only if $Y_2\in \Gamma_2$. Accordingly, we let $\Gamma=\Gamma_2$. 
 
The BMSSs are characterized by the relation
 $$\oP=\varphi(Y_2)= \varphi_2(f(Y_2),Y_2).$$
The function $\varphi$ is $\mmC^1$  for $Y_2\in \Gamma$.  Further, $\varphi(0)=0$ and as $Y_2$ tends $\xi_2$, $f(Y_2)$ tends to zero and hence $\varphi(Y_2)$ tends to $+\infty$. Consequently, there is at least one BMSS. 

Determination of the number of  BMSSs follows from the behavior of the function $\varphi$ in $\Gamma$. Since $\Gamma$ is a connected interval, multistationarity can only occur if $\varphi'(Y_2)<0$ for some value $Y_2\in \Gamma$ (Figure~\ref{summary}(c)). In this case multiple BMSSs occur for  $\oP_{\min}\leq \oP\leq \oP_{\max}$ for some $\oP_{\min}<\oP_{\max}$. Analysis  of the behavior of this function is not straightforward but still some general conclusions can be derived.

Clearly $Y_2+\mu_2Y_2$ is an increasing function of $Y_2$. The derivatives of $P_0,P_1$ with respect to $Y_2$ are 
$$\frac{\partial P_0}{\partial Y_2}= \frac{\mu_2\delta_1}{\eta_2f^2} (f(\oF - f - 2Y_2) - f'Y_2(\oF-Y_2)) \quad \textrm{and}\quad 
\frac{\partial P_1}{\partial Y_2} = \frac{\oF - f + f'Y_2}{\delta_2 (\oF-f-Y_2)^2}.$$
 Either $\frac{\partial P_0}{\partial Y_2}$ or $\frac{\partial P_1}{\partial Y_2}$ must be negative for $\varphi'(Y_2)<0$. 
 
Let $\Delta_1=\mu_1\oF - \oE$.

\begin{proposition}\label{nomulticascade}
Let $\overline{\sigma }= \min(\oS,\oP)$ and $\sigma_{\mu} = \min (1+\mu_1,\mu_2/2). $ 
If  $\Delta_1>0$ and $$\sigma_{\mu}(\oF - \oE/\mu_1) > \overline{\sigma}$$
then $\frac{\partial P_0}{\partial Y_2}(Y_2),\frac{\partial P_1}{\partial Y_2}(Y_2)>0$ for all $Y_2\in \Gamma$. Consequently,  multistationarity cannot occur.
\end{proposition}

Finally, let us analyze the situation where  $\oS$ is large.  

\begin{proposition}\label{varphinf} 
The following statements hold:
\begin{enumerate}[(i)]
\item Fix $Y_2> \oF-\oE/\mu_1$. Then,  as $\oS$ tends to $+\infty$, $\varphi'(Y_2)$ becomes positive, $\varphi'(Y_2)>0$. 
\item Assume that $\Delta_1>0$. Let $\Delta_2=27\mu_1\mu_2^2\delta_1^2\eta_2\oE(\mu_1\oF-\oE) -\delta_2(\delta_1\mu_2\mu_1\oF-(\mu_2\delta_1+\eta_2(1+\mu_2))\oE)^3$. If $\Delta_2<0$, then $\varphi$ decreases for some $Y_2\leq\oF-\oE/\mu_1$ as $\oS$ becomes large.  
\end{enumerate}
\end{proposition}

Implicitly Proposition~\ref{varphinf}  assumes that  $\oP$ is large too: In (i), $Y_2$ is fixed and $\oS$ becomes large restricting the possible values of $Y_1$ and $\oP$. In (ii), the same is in play. Thus, contradictory conclusions cannot be reached from the two propositions. For example,  if $\oF\gg \oE$, then $\Delta_2<0$ and Proposition~\ref{varphinf}(ii) guarantees that for $\oS$ (and thus $\oP$) large,  multistationarity exists. If $\oF\gg \oE$ and $\oP, \oS$ fixed, Proposition~\ref{nomulticascade} ensures monostationarity.
Also, it follows from Proposition~\ref{varphinf} (i) that if $ \oF-\oE/\mu_1<0$, multistationarity for $\oS$ large cannot occur. 

Multistationarity can also occur if $\Delta_1\leq 0$. For instance, the total amounts $\oF=4,\oE=10,\oS=50,\oP=195$ and rate constants $\mu_2=15, \mu_1=2,\eta_2=\eta_1=\delta_1=1,\delta_2=0.1$ produce three steady states  with $Y_2=2.25,2.78,3.26$, respectively. However, contrary to the  situation with the reversed inequality,  $\Delta_1>0$, multistationarity cannot occur for $\oS$ large (Proposition \ref{varphinf}(i)).

We have the following result:

\begin{result}[Motif (k)]\label{result_cascade2} 
Let a cascade be given with one phosphatase acting in both layers. Further, assume that the  total amounts $\oE,\oF,\oS,\oP$ are positive. 
\begin{enumerate}[(i)]
\item If $\oF \gg \oE$ and $\oS$ large, then there exist values $\oP_{\min}<\oP_{\max}$ such that for  all $\oP_{\min}\leq \oP\leq \oP_{\max}$ the system admits more than one BMSS.
\item If $\oF - \oE/\mu_1 > \overline{\sigma}/\sigma_{\mu}$ with $\overline{\sigma }= \min(\oS,\oP)$ and $\sigma_{\mu} = \min (1+\mu_1,\mu_2/2)$, then multistationarity cannot occur.
\item If $\oF - \oE/\mu_1 <0$,   then  multistationarity cannot occur for large $\oS$.
\end{enumerate}
\end{result}

We conclude that multistationarity occurs in this motif. Note that statement (iii) implies that if  multistationarity occurs for some values $\oE/\mu_1>\oF$ and some $\oS,\oP$, then increasing  $\oS$ eventually makes  the system monostationary.   Observe also that according to (i) and (ii)  together, $\oP_{min}$ is required to be so large that the condition $\oF - \oE/\mu_1 > \oP_{min}/\sigma_{\mu}$ is not fulfilled.

 \paragraph{Motif (l).}
This motif is a combination of Motif (c) and Motif (j).  The chemical reactions of the system are:
 
 \centerline{
\xymatrix@C=14pt{
S_{0} + E \ar@<0.5ex>[r]^(0.6){a^E_1} & X_1  \ar@<0.5ex>[l]^(0.4){b^E_1} \ar[r]^(0.4){c^E_1} & S_{1} + E & 
P_{0} + S_1 \ar@<0.5ex>[r]^(0.6){a^E_2} & X_2  \ar@<0.5ex>[l]^(0.4){b^E_2} \ar[r]^(0.4){c^E_2} & P_{1} + S_1 & P_{0} + E \ar@<0.5ex>[r]^(0.6){a^E_3} & X_3  \ar@<0.5ex>[l]^(0.4){b^E_3} \ar[r]^(0.4){c^E_3} & P_{1} + E\\
S_{1} + F_1 \ar@<0.5ex>[r]^(0.6){a^F_1} & Y_1 \ar@<0.5ex>[l]^(0.4){b^F_1} \ar[r]^(0.4){c^F_1} & S_{0} + F_1 &
P_{1} + F_2 \ar@<0.5ex>[r]^(0.6){a^F_2} & Y_2  \ar@<0.5ex>[l]^(0.4){b^F_2} \ar[r]^(0.4){c^F_2} & P_{0} + F_2
}}
\noindent
The differential equations describing the  system  are the following:

\noindent
\hspace{-0.3cm}
\begin{minipage}[h]{0.62\textwidth}
{\small\begin{align}
\dot{P_{0}} &=  b^E_{2} X_{2} +b^E_3X_3 +  c^F_{2} Y_{2} - a^E_{2} S_1 P_{0}-a^E_3EP_0    \nonumber   \\
\dot{E} &= (b^E_1 + c^E_1)X_1+(b^E_3 + c^E_3)X_3- (a^E_1 S_{0}+ a^E_3 P_{0})E \nonumber \\
\dot{F}_1 &= (b^F_1 + c^F_1)Y_1 - a^F_1 F_1 S_1 \nonumber \\
\dot{F}_2 &=  (b^F_2 + c^F_2)Y_2 - a^F_2 F_2 P_1    \nonumber \\
\dot{P_{1}} &=  b^F_{2} Y_{2}+  c^E_{2} X_{2}+c_3^EX_3  - a^F_{2} F_2 P_{1}    \label{eq1l3}   \\
\dot{S_{1}} &=  c^E_1 X_1  +b^F_1 Y_1  + (b^E_2 + c^E_2)X_2- (a^F_1 F_1 +a^E_2 P_{0})S_1 \label{eq1l5}  
\end{align}}
\end{minipage}
\hspace{-0.4cm}
\begin{minipage}[h]{0.39\textwidth}
{\small\begin{align}
\dot{S_{0}} &=  b^E_{1} X_{1} +  c^F_{1} Y_{1} - a^E_{1} E S_{0}    \nonumber   \\
\dot{X_1} &= a^E_1 E S_{0} -(b^E_1 + c^E_1)X_1  \label{eq1l1} \\
\dot{X_2} &= a^E_2 S_1 P_{0} -(b^E_2 + c^E_2)X_2  \label{eq1l4} \\
\dot{X_3} &= a^E_3 E P_{0} -(b^E_3 + c^E_3)X_3 \label{eq1l7} \\
\dot{Y_1} &=a^F_1 F_1 S_1 -(b^F_1 + c^F_1)Y_1   \label{eq1l2}  \\
\dot{Y_2} &=a^F_2 F_2 P_1 -(b^F_2 + c^F_2)Y_2  \label{eq1l6} 
\end{align}}
\end{minipage}

\vspace{-0.3cm}
The conservation laws are given by  
$$
\overline{E} = E + X_{1}+X_3,\quad \overline{F}_k = F_k + Y_k,\quad   \overline{S}   =   S_{0}+S_1+ X_{1}+X_2+Y_1,\quad \overline{P}= P_0+P_1+X_2+Y_2+X_3.$$
If total amounts are given, the system to be solved consists of 12 equations in 12 variables which are the equations for the total amounts and, for instance, equations \eqref{eq1l3}-\eqref{eq1l6}. We obtain 
$$ X_1 = \eta_1 E S_{0}, \quad X_2 = \eta_2 S_1 P_{0},\quad X_3=\eta_3 E P_0, \quad Y_1=\delta_1 F_1 S_1,\quad Y_2=\delta_2 F_2 P_1$$
with constants $\eta_k=a^E_k/(b^E_k + c^E_k)$ and $\delta_k=a^F_k/(b^F_k + c^F_k)$.
Equations \eqref{eq1l5},\eqref{eq1l4},\eqref{eq1l2}, and \eqref{eq1l3},\eqref{eq1l6}  provide
$$X_1 = \mu_1 Y_1,\quad \mu_2^{-1}X_2+\mu_3^{-1}X_3-Y_2=0$$ with $\mu_1= c^F_{1}/ c^E_{1},$ $\mu_2=c_2^F/c_2^E$ and $\mu_3=c_2^F/c_3^E$.

We write $X_1$ as a function of $Y_1$ and isolate $E,F_1,F_2$ using the conservation laws for the total amounts $\oE,\oF_1,\oF_2$.  If $\oE,\oF_k> 0$ then $E,F_k\neq 0$ at steady state and it follows that $\oE - X_1-X_3>0$ for any BMSS. Further, for any BMSS we require $0\leq Y_k<\oF_k$ and $Y_1<\oE/\mu_1$. It follows as well that $X_2<\mu_2\oF_2$.  If in addition, $\oS> 0$ and $\oP>0$ then also   $S_0,S_1,X_1,Y_1,X_2,P_0,P_1,Y_2,X_3>0$ for any BMSS.

Let $\xi_1=\min(\oF_1,\oE/\mu_1)$, $\Gamma_1=(0,\xi_1)$ and $\Gamma_2=(0,\mu_2\oF_2)$. It is convenient for this motif to exclude 0 in the intervals. A necessary condition for positivity (i.e. a BMSS) is thus $Y_1\in \Gamma_1,X_2\in \Gamma_2$. We proceed to write $S_1$ as an increasing $\mmC^1$ functions of $Y_1\in \Gamma_1$ and $P_0$ as a $\mmC^1$ function increasing in $X_2\in \Gamma_2$ and  decreasing in $Y_1\in \Gamma_1$:
$$S_1= \frac{Y_1}{\delta_1(\oF_1-Y_1)},\quad P_0 =\frac{\delta_1X_2(\oF_1-Y_1)}{\eta_2Y_1}.$$
 Using $X_3=\eta_3 E P_0$, we express $X_3$ as a positive $\mmC^1$ function, increasing in $X_2\in \Gamma_2$ and decreasing in $Y_1\in \Gamma_1$:
$$ X_3 =\Phi_X(Y_1,X_2)= \frac{\delta_1\eta_3X_2(\oF_1-Y_1)(\oE-\mu_1Y_1)}{\eta_2Y_1 + \delta_1\eta_3X_2(\oF_1-Y_1)}.$$
  We proceed to write $Y_2$ as a positive $\mmC^1$ function, increasing in $X_2\in\Gamma_2$ and decreasing in $Y_1\in\Gamma_1$:
$$ Y_2 = \Phi_Y(Y_1,X_2)=\mu_2^{-1}X_2+ \mu_3^{-1}\Phi_X(Y_1,X_2).$$ 
For $Y_2<\oF_2$, we require $(Y_1,X_2)\in \Gamma'$ with  
$$\Gamma' = \{(Y_1,X_2)|\ Y_1\in \Gamma_1,\ X_2\in \Gamma_2,\  \Phi_Y(Y_1,X_2)<\oF_2 \}. $$
From $P_1=Y_2/(\delta_2F_2)$ we have 
$$P_1= \frac{ \Phi_Y(Y_1,X_2)}{\delta_2(\oF_2- \Phi_Y(Y_1,X_2))},$$ which is  increasing in $X_2$, decreasing in $Y_1$ and positive and continuous provided $(Y_1,X_2)\in \Gamma'$.  Finally, 
$$ S_0=\frac{\mu_1Y_1}{\eta_1(\oE-\mu_1Y_1-\Phi_X(Y_1,X_2))} $$
which is positive and $\mmC^1$ in $\Gamma'$.

To sum up: All concentrations at steady state are positive   if only if $(Y_1,X_2)\in \Gamma'$.  To find a final relation between $X_2,Y_1$, we consider the total amount $\oP$. For $Y_1,X_2$ in $\Gamma'$, 
$$\oP= \varphi_P(Y_1,X_2),$$  
where $\varphi_P$ a positive $\mmC^1$ function. Note also that for every fixed value of $Y_1\in \Gamma_1$, $\varphi_P(Y_1,\cdot)$ increases in $X_2$, is well-defined at $X_2=0$ where it is zero and tends to infinity as $\Phi_Y(Y_1,X_2)$ tends to $\oF_2$. Therefore, for $\oP>0$ and any $Y_1\in \Gamma_1$, there exists $0<X_2$ satisfying $\oP=\varphi_P(Y_1,X_2)$ and  $\Phi_Y(Y_1,X_2)<\oF_2$.  Thus, there exists a function $f$ defined on $\Gamma_1$ for which 
$$ X_2=f(Y_1)$$ 
at steady state and $(Y_1,f(Y_1))\in \Gamma'$. Since  $\varphi_P$ is increasing in $X_2$ and decreasing in $Y_1$, the function $f$ is $\mmC^1$ and increasing in $\Gamma:= \Gamma_1$. 

All concentrations at steady state are positive if and only if $Y_1\in \Gamma$.
The concentrations $S_1,X_1,Y_1$ are increasing functions of $Y_1$ and independent of $X_2$. $X_2$ is increasing in $Y_1$ (after substitution of $f$). $S_0$ might be increasing or decreasing in $Y_1$ since $\Phi_X(Y_1,X_2)$ is not increasing in $Y_1$. In any case, if we insert these values into $\oS$, we obtain that the BMSSs are given by a relation
$$ \oS = \varphi(Y_1),$$
where $\varphi$ is a positive $\mmC^1$ function defined on $\Gamma$.
When $Y_1$ tends to $\xi_1$, the function tends to infinity, since $f(Y_1)$ is bounded by $\mu_2\oF_2$. When $Y_1$ tends to zero, then $X_2$ tends to zero as well as  can be seen from the equations $Y_1=\delta_1F_1S_1$ and $X_2=\eta_2S_1P_0$ and the fact that $P_0$ is bounded by $\oP$. Therefore, there is at least one BMSS. Multistationarity can only occur in the form of Figure~\ref{summary}(c), implying the existence of lower and upper bounds on $\oS$, $\oS_{min}<\oS_{max}$, for the existence of multistationarity.

If $S_0$ is increasing in $Y_1$, then so is $\varphi$ and there is exactly one BMSS. The derivative of $S_0$ with respect to $Y_1$ is
$$\frac{\partial S_0}{\partial Y_1} = \frac{\mu_1(\eta_2\oE + \delta_1\eta_3 X_2 (\mu_1\oF_1-\oE))}{\eta_1\eta_2(\oE - \mu_1Y_1)^2}+\frac{\partial S_0}{\partial X_2}f'. $$ 
Since $(\partial S_0/ \partial X_2)f'>0$, 
we have that if (a) $\mu_1\oF_1-\oE\geq 0$ or (b) $\mu_1\oF_1-\oE<0$ and $\eta_2\oE + \delta_1\eta_3 \oF_2 (\mu_1\oF_1-\oE)\geq 0$, then $\frac{\partial S_0}{ \partial Y_1}>0$ and hence  multistationarity cannot occur.

\begin{proposition}\label{motifllemma}  Fix the total amount $\oF_1$. Then there exist total amounts $\oP,\oS, \oF_2, \oE$ satisfying $\mu_1\oF_1-\oE<0$   for which $\varphi'(Y_1)<0$ for some $Y_1$.
\end{proposition} 

It follows from the discussion above the proposition that   $$\eta_2\oE + \delta_1\eta_3 \oF_2 (\mu_1\oF_1-\oE)<0$$ for multistationarity to occur. In fact, $\oF_2$ and $\oE$ are chosen so large that this condition is fulfilled. The proof of this proposition is provided in  Appendix A. 

\begin{result}[Motif (l)]\label{result_cascade3} 
Consider a cascade  with  different phosphatases acting in the two layers and  where  the kinase of the first layer also acts in the second layer. Assume that the  total amounts $\oE,\oF_1,\oF_2,\oS,\oP$ are positive. 
\begin{enumerate}[(i)]
\item If (a) $\mu_1\oF_1-\oE\geq 0$ or (b) $\mu_1\oF_1-\oE<0$ and $\eta_2\oE + \delta_1\eta_3 \oF_2 (\mu_1\oF_1-\oE)\geq 0$, then   the system has exactly one BMSS.
\item For any  $\oF_1$ and $\oE,\oF_2$ large  and satisfying $\mu_1\oF_1-\oE<0$ and $\eta_2\oE + \delta_1\eta_3 \oF_2 (\mu_1\oF_1-\oE)< 0$, there exist values $\oP,\oS$ for which the system displays multistationarity. Further, in this case, there exists $\oS_{min}<\oS_{max}$ for which multiple steady states occur for $\oS_{min} \leq \oS\leq \oS_{max}$.
\end{enumerate}
\end{result}

\section{Stability analysis}
\label{unstable}
We provide here the mathematical details of the stability analysis. We prove equation [3] in the main text and show that for the motifs exhibiting multistationarity, the equilibrium points for which $\varphi'<0$ are unstable.

\subsection{The determinant of the Jacobian}
We prove here equation [3] of the main text. 
For any $1\leq j\leq n$, we define the projection $\pi^{(j)}$ of $\R^n$ to $\R^{n-1}$ that removes  the $j$-th coordinate by
\begin{eqnarray*}
\pi^{(j)}: \R^n & \rightarrow  & \R^{n-1} \\
(x_1,\dots,x_n) & \mapsto & (x_1,\dots,\widehat{x_{j}},\dots,x_n).
\end{eqnarray*}
For simplicity, let $x^{(j)}=\pi^{(j)}(x)$ for any $x\in \R^n$ and let  $\Omega^{(j)}= \pi^{(j)}(\Omega)\subseteq \R^{n-1}$ for any  open set $\Omega\subset \R^n$. Note that the latter is also an open set (projection maps are open maps).

For  a differentiable function $f=(f_1,\dots,f_n)\colon\Omega\subseteq \R^n\rightarrow \R^n$, we denote by $J_{f}(x)$ the Jacobian of $f$ at $x\in \Omega$.  The $(i,j)$ entry of   $J_f(x)$  is $\partial f_i/\partial x_j(x)$. 

In the next proposition, $\Omega$ is an open neighborhood of $z$, suitably chosen.

\begin{proposition}\label{reorder}\label{jacobian_det}
Let $f=(f_1,\dots,f_n)\colon\Omega\subseteq \R^n\rightarrow \R^n$ be a differentiable function defined on an open set $\Omega$ such that there is $z\in \Omega$ with $f(z)=0$. Assume that $x_j$ can be eliminated from the equation $f_i=0$; that is, there exists a differentiable function $\psi\colon\Omega^{(j)}\subseteq \R^{n-1}\rightarrow \R$ such that $x_j=\psi(x^{(j)})$ on $\{x\in\Omega| f_i(x)=0\}$.
Define $\bar{f} \colon \Omega^{(j)} \rightarrow \R^{n-1}$ by
$\bar{f}_k(x) = f_k(x_1,\dots,x_{j-1},\psi(x),x_j,\dots,x_{n-1})$ for all $k\neq i$.
Then, the determinant of the Jacobian of $f$ at  $z$ satisfies
\begin{equation}\label{jacobian}
 (-1)^{i+j}\frac{\partial f_i}{\partial x_j}(z)\ \det(J_{\bar{f}}(z^{(j)})) = \det(J_f(z)). 
 \end{equation}

\end{proposition}
\begin{proof}
It is sufficient to show that the proposition holds for $i=j=1$. For other values of $i$ and $j$ the result is obtained by 
reorganizing rows and columns and keeping track of the sign of the determinant.
For $i=j=1$, the proposition states that
$$
\frac{\partial f_1}{\partial x_1}(z)\ \det(J_{\bar{f}}(z^{(1)})) = \det(J_f(z)). 
$$
where $z^{(1)}=(z_2,\dots,z_n)$ and $z=(z_1,\dots,z_n)$.

If the Jacobian is written in column vector notation, we have that
$$ J_f(z) = (D_1f(z),\dots,D_nf(z))  \quad \textrm{where } D_kf = \Big(\frac{\partial f_1}{\partial x_k},\dots,\frac{\partial f_n}{\partial x_k}\Big)^T.  $$
That is, $D_kf$ is the vector of the partial derivatives of $f$ with respect to $x_k$.

By assumption, $z=(\psi(z^{(1)}),z_2,\dots,z_n)\in \R^n$. Observe that if $(\partial f_1/\partial x_1)(z)= 0$ then from $\frac{\partial f_1}{\partial x_1}(z)\frac{\partial \psi}{\partial x_k}(z^{(1)}) +\frac{\partial f_1}{\partial x_k}(z)=0$, we have that $(\partial f_1/\partial x_k)(z)= 0$ for all $k$. Consequently, $\det(J_f(z))=0$ (the matrix has one row of zeroes) and hence the proposition holds. 

Hence, we can assume that $(\partial f_1/\partial x_1)(z)\neq 0$. By implicit differentiation,
\begin{equation}\label{psi_jac}
\frac{\partial \psi}{\partial x_k}(z^{(1)}) =-\frac{\partial f_1/\partial x_k}{\partial f_1/\partial x_1}(z).
\end{equation}
Additionally, by the chain rule we have ($k\not=1)$
$$\frac{\partial \bar{f}_l}{\partial x_k}(z^{(1)}) = \frac{\partial f_l}{\partial x_k}(z) + \frac{\partial f_l}{\partial x_1}(z)\frac{\partial \psi}{\partial x_k}(z^{(1)}).$$
Let $f^{(1)} = (f_2,\dots,f_n)$ be the projection of $f$ onto the last $n-1$ coordinates. 
From the equation above, it follows that 
\begin{align*}
J_{\bar{f}}(z^{(1)}) &= (D_2\bar{f},\dots,D_n\bar{f})(z^{(1)}) =(D_2f^{(1)} ,\dots,D_nf^{(1)} )(z) \\ & \qquad \phantom{ (D_2\bar{f},\dots,D_n\bar{f})(z^{(1)}) }+ \Big(D_1f^{(1)}(z)\frac{\partial \psi}{\partial x_2}(z^{(1)}),\dots,D_1f^{(1)}(z)\frac{\partial \psi}{\partial x_n}(z^{(1)})\Big),
\end{align*}
where matrices are written as column vectors.

Note that $\frac{\partial \psi}{\partial x_k}(z^{(1)})$ is  a scalar and  $D_1f^{(1)}(z)\frac{\partial \psi}{\partial x_k}(z^{(1)})$ is the vector with $l$-th component $\frac{\partial f_l}{\partial x_1}(z)\frac{\partial \psi}{\partial x_k}(z^{(1)})$. Further, using the multilinear expansion of a determinant, we obtain
\begin{align*}
\det( J_{\bar{f}}(z^{(1)})) &= \det((D_2f^{(1)} ,\dots,D_nf^{(1)} )(z)) \\ &+ \sum_{k=2}^n
 \det(D_2f^{(1)}(z),\dots, D_1f^{(1)}(z)\frac{\partial \psi}{\partial x_k}(z^{(1)}),\dots,D_nf^{(1)}(z)) \\
 & + \sum_{\substack{\{k_1,\dots,k_l\}\subset \{2,\dots,n\}\\ l>1}}
\det(a_2,\dots,a_n),
\end{align*}
where $a_s= D_sf^{(1)}(z)$ for $s\neq k_1,\dots,k_l$ and $a_s= D_1f^{(1)}(z)\frac{\partial \psi}{\partial x_s}(z^{(1)})$ otherwise. In each of these summands, there are at least two terms of the second form, say $D_1f^{(1)}(z)\frac{\partial \psi}{\partial x_k}(z^{(1)})$ and $D_1f^{(1)}(z)\frac{\partial \psi}{\partial x_m}(z^{(1)})$ for some $k\neq m$. If  the two columns are non-zero, they are linearly dependent and the determinant of each of the matrices $(a_2,\dots,a_n)$ is zero. Further, note that 
\begin{align*}
& \det(D_2f^{(1)}(z),\dots, D_1f^{(1)}(z)\frac{\partial \psi}{\partial x_k}(z^{(1)}),\dots,D_nf^{(1)}(z))  \\ & \phantom{holaholaholahola}=\frac{\partial \psi}{\partial x_k}(z^{(1)}) \det((D_2f^{(1)},\dots, D_1f^{(1)},\dots,D_nf^{(1)})(z))  \\
& \phantom{holaholaholahola} =(-1)^{k-2}\frac{\partial \psi}{\partial x_k}(z^{(1)}) \det((D_1f^{(1)},\dots,\widehat{D_kf^{(1)}},\dots,D_nf^{(1)})(z))
\end{align*}

Finally, using \eqref{psi_jac},  we obtain
\begin{align*}
 \frac{\partial f_1}{\partial x_1}(z)\det(J_{\bar{f}}(z^{(1)})) &= \sum_{k=1}^n (-1)^{k-1}\frac{\partial f_1}{\partial x_k}(z)  \det(D_1f^{(1)},\dots,\widehat{D_kf^{(1)}},\dots,D_nf^{(1)})(z) \\ & = \det(J_f(z)) \end{align*}
by considering the  development of the determinant of $J_f(z)$ along the first row.
\end{proof}

\paragraph{Application.}
Let a dynamical system in $\Omega\subseteq \R^n$ be given
$$\dot{x}=f(x) $$
with $x=(x_1,\dots,x_n)$ and $f=(f_1,\dots,f_n)$. Let $z=(z_1,\dots,z_n)$ be an equilibrium point, i.e., $f(z_1,\dots,z_n)=0$. Let $\det(J_f(z))$ be the determinant of the Jacobian of $f$ at $z$. We make the following observation:
\begin{itemize}
\item If $n$ is odd and $\det(J_f(z))> 0$, then $z$ is unstable.
\item If $n$ is even and $\det(J_f(z)) < 0$, then $z$ is unstable.
\end{itemize}
Indeed, if $z$ is  (asymptotically) stable and $\det(J_f(z))\not=0$, all  eigenvalues have negative real parts. Since $\det(J_f(z))$ is a real number (it is the determinant of a real matrix), the complex eigenvalues come in pairs of conjugates $a+bi, a-bi$ and the product of the eigenvalues is a positive number. If $n$ is odd and all eigenvalues have negative real parts, their product must be negative and hence the determinant of the Jacobian, which equals the product of the eigenvalues, must be negative. If $n$ is even and $z$ (asymptotically) stable, then the product of the eigenvalues must be positive.

An equilibrium point satisfies $f(z_1,\dots,z_n)=0$. If elimination of variables can be applied then we can use Proposition \ref{jacobian_det} to track the sign of the determinant  and potentially detect instability.

\paragraph{Chemical reaction networks.}
Consider any of the motifs, or in general, any chemical reaction network with conservation laws. 
The conserved total amounts imply that the dynamics of the associated dynamical system takes place in a fixed subspace of $\R^n$. In general, we have a dynamical system
$$\dot{x}=f(x)   $$
and a series of (independent) conservation laws
\begin{equation}\label{cons}
 g_i(x) - c_i=0, \qquad i=1\dots,k
\end{equation}
where $g_i$ a linear function of $x$ satisfying $g_i(\dot{x})=0$ and $c_i\in \R$ (these correspond to the total amounts). By independence we mean that the rank of this linear system is $k$. The conservation laws do not depend on the rate constants and $g_i$ does not depend on the total amounts.

The existence of conservation laws implies that the determinant of the Jacobian of $f$ at any point $x$ is zero, since the matrix has linear relations among the rows. Therefore, stability of equilibrium points cannot be analyzed directly, but need to be considered inside the stoichiometry class  they belong to. 

Since \eqref{cons} is a linear system of rank $k$, Gauss elimination allows the elimination of $k$ variables.
For simplicity, we can rename the variables such that the eliminated variables are $x_1,\dots,x_k$ and those that remain are $x_{k+1},...,x_n$. Apply the same renaming to the functions $f_i$, such that if $x_l$ is now variable $x_j$, function $f_l$ is labeled $f_j$. By elimination, there exist (polynomial) functions 
$$x_i= \psi_i(x_{k+1},\dots,x_n,c), \quad i=1,\dots,k$$
such that  $f_i (\psi_1,\dots,\psi_k,x_{k+1},\dots,x_n)=0$. Here $c=(c_1,\dots,c_k)$ is the vector of initial total amounts.

For a fixed stoichiometry class $c=(c_1,\dots,c_k)$ and $\bar{x}=(x_{k+1},\dots,x_n)$,  let 
$$\bar{f}_{j}(\bar{x},c) = f_j (\psi_1(\bar{x},c),\dots,\psi_k(\bar{x},c),x_{k+1},\dots,x_n),\quad j=k+1,\dots,n. $$
To investigate stability of an equilibrium point $(z_1,\dots,z_n)$ belonging to the stoichiometry class $c$ with  $z_i-\psi_i(z_{k+1},\dots,z_n,c)=0$, we consider the the eigenvalues of the Jacobian  of the function
$$(\bar{f}_{k+1}(\bar{z},c),\dots, \bar{f}_{n}(\bar{z},c)),$$
evaluated at $\bar{z}=(z_{k+1},\dots,z_n)$. This function corresponds to the system called \emph{reduced} in the main text.

By Proposition \ref{jacobian_det} the sign of the determinant of this Jacobian at $\bar{z}$ is exactly the sign
of the determinant of the Jacobian of the system
$$
S_f=\begin{cases}
x_i-\psi_i(\bar{x},c)=0 & i=1,\dots,k \\
f_i(x)=0 & i=k+1,\dots,n
\end{cases}
$$
evaluated at $z$. Indeed, the process leading from this system to the reduced system consists of successive eliminations with $i=j=1$ and the derivative of the eliminated function corresponding  to the eliminated variable is $1$ (and thus positive).

Note that the reduced system has $n-k$ variables. Let $J(S_f)$ denote the Jacobian of $S_f$ and let $z$ be an equilibrium point with total amounts $c$. We conclude that:

\begin{result}\label{criterion} With the notation introduced above:
\begin{itemize}
\item If $n-k$ is odd and $\det(J(S_f)(z))> 0$, then $z$ is unstable.
\item If $n-k$ is even and $\det(J(S_f)(z))< 0$, then $z$ is unstable.
\end{itemize}
\end{result}

\subsection{Instability in the multistationary motifs}

Here we show  how to apply Result \ref{criterion} to each of the multistationary motifs: (f), (g), (i), (k), (l). Using Proposition \ref{jacobian_det}, we find for  all motifs but Motif (l)  that
$$ \sign(\varphi'(z_i))= \sign (\det(J(S_f)(z))),$$ 
where $z_i$ is the variable of $\varphi$ (typically an intermediate complex $Y$) and $n-k$ is even. For Motif (l), $n-k$ is odd and there is a change of sign. Hence, using Result \ref{criterion},  the steady state $z$ is unstable whenever  $\varphi$ is decreasing.
In particular, let $\epsilon=\pm 1$ (i.e.~either $\epsilon=+1$ or $\epsilon=-1$) such that 
$$\epsilon \cdot \sign(\varphi'(z_i))=\sign (\det(J(S_f)(z))).$$ 

\begin{obs}\label{signchanges}
The sign of the determinant of a matrix remains unchanged if a linear combination of rows are added to another row (in fact the determinant does not change). If we multiply a row by a negative number, the determinant changes sign. For example, if the equation $ -(b^E + c^E)X + a^E E S_{0}=0$ is  transformed into $X-\eta E S_0=0$, then the sign of the determinant changes. If two such equations are transformed in this way, then the determinant remains with unchanged sign. In the sequel, the sign remains unchanged if the number of transformations of this type is even, or equivalently, if the number of constants $\delta_*,\eta_*$ is even. 
\end{obs}

In the sequel, elimination is tracked using a table as in the main text and $\epsilon=\prod_{l=1}^{n-1} \epsilon_l$. The column `Behavior' in the table shows $(i,j,s)$  where $i$, respectively $j$, are the indices of the equation, respectively the variable, that iteratively are being eliminated and $s$  indicates whether $\bar{f}_i$ ($f_i$ after substitution of the previous eliminated variables) is increasing ($s$ is $+$) or decreasing ($s$ is $-$) as function of $x_j$. 

Motif (f) is covered in the main text and it is thus skipped here.

\paragraph{Motif (g).}
Consider the variables in the following order
$$(x_1,x_2,x_3,x_4,x_5,x_6,x_7,x_8,x_9)=(E,F,S_0,X_1,X_2,S_1,S_2,Y_2,Y_1).$$
The variables of the reduced system are $X_1,X_2,S_1,S_2,Y_2,Y_1$ and the variables eliminated using the conservation laws are $E,F,S_0$. The sign of the determinant of the Jacobian of the system $S_f$ is the same as the sign of the determinant of the Jacobian of the system

\vspace{-0.5cm}
\begin{align*}
f_1(x)  &= E + X_1+X_2 - \oE  &  f_5(x) &=  X_2 - \eta_2 E S_{1}  \\
f_2(x) &= F + Y_1 + Y_2 -\oF & f_6(x) &= X_1-\mu_1 Y_1  \\  
f_3(x)  &= S_0 + S_1 + S_2+X_1+X_2 + Y_1+Y_2 - \oS  & f_7(x) &=X_2- \mu_2 Y_2\\ 
 f_4(x) &= X_1 - \eta_1 E S_{0} & f_8(x) &=  Y_2- \delta_2 F S_2 \\ 
 & & f_{9}(x) &= Y_1- \delta_1 F S_1.
\end{align*}
With the notation introduced above, $\bar{x}=(X_1,X_2,S_1,S_2,Y_2,Y_1)$ and
$$\psi_1(\bar{x})=- X_1- X_2 + \oE,\quad \psi_2(\bar{x})= - Y_1 - Y_2 +\oF,\quad \psi_3(\bar{x})= -S_1 - S_2-X_1-X_2 - Y_1-Y_2 + \oS. $$

Here $n=9$ and $k=3$, such that $n-k=6$ is even. The function $\varphi$ is equation $f_3$ with all variables but $x_6=S_1$ eliminated. The eliminations we performed in Section \ref{motiftwosite} (Motif (g)) are summarized in the following table:

\begin{center}
\renewcommand{\arraystretch}{1.3}
\setlength{\tabcolsep}{3pt}
\begin{tabular}{cccc||cccc}
\hline
$l$ & Elimination &  Behavior & $\epsilon_l$ & l & Elimination &  Behavior & $\epsilon_l$\\ \hline
1 & $(f_1,E)$  &  $(1,1,+)$ & $+$  &6 & $(f_4,S_0)$ & $(2,1,-)$ & +  \\
2 & $(f_2,F)$  & $(1,1,+)$ & +  & 7 & $(f_8,S_2)$ &  $(3,2,-)$ & +  \\
4 & $(f_6,X_1)$  &  $(4,2,+)$ & +  & 8 & $(f_9,Y_1)$ & $(3,3,+)$ & +  \\ 
5 & $(f_7,X_2)$  & $(4,2,+)$ & + & 9 & $(f_{5},Y_2)$ & $(2,2,+)$ & + 
 \\ \hline
\end{tabular}
\end{center}

Since $\epsilon=1$, we conclude that $ \sign(\varphi'(z_6))= \sign (\det(J(S_f)(z)))$
for any equilibrium point $z$ with $z_6=S_1$. 

\paragraph{Motif (i).}
Consider the variables in the following order
$$(x_1,x_2,x_3,x_4,x_5,x_6,x_7,x_8,x_9,x_{10})=(E,F,S_0,P_0,X_1,X_2,S_1,P_1,Y_1,Y_2).$$
The variables of the reduced system are $X_1,X_2,S_1,P_1,Y_1,Y_2$ and the variables eliminated using the conservation laws are $E,F,S_0,P_0$. The sign of the determinant of the Jacobian of the system $S_f$ is the same as the sign of the determinant of the Jacobian of the system
\begin{align*}
f_1(x)  &= E + X_1+X_2 - \oE & f_5(x) &=  X_1 - \eta_1 E S_{0} & f_8(x) &= X_2- \mu_2 Y_2 \\ 
f_2(x) &= F + Y_1 + Y_2 -\oF  &  f_6(x) &=  X_2 - \eta_2 E P_0   & f_9(x) &= Y_1- \delta_1 F S_1\\
f_3(x) &=  S_0 + S_1 + X_1+ Y_1 - \oS  &  f_7(x) &= X_1-\mu_1 Y_1 &  f_{10}(x)&= Y_2- \delta_2 F P_1 \\  f_4(x) &=P_0+P_1+X_2+Y_2-\oP 
\end{align*}
Let $\bar{x}=(X_1,X_2,S_1,P_1,Y_1,Y_2)$. Then
\begin{align*}
\psi_1(\bar{x}) & =- X_1- X_2 + \oE    & \psi_3(\bar{x}) & = -S_1 -X_1- Y_1+ \oS \\  \psi_2(\bar{x}) & = - Y_1 - Y_2 +\oF  & \psi_4(\bar{x}) & =-P_1-X_2-Y_2+\oP. 
\end{align*}
Here $n=10$ and $k=4$, such that $n-k=6$ is even. The function $\varphi$ is equation $f_4$ with all variables but $x_{10}=Y_2$ eliminated. The eliminations we performed in Section~\ref{motiftwodiff} (Motif (i)) are summarized in the following table:

\begin{center}
\renewcommand{\arraystretch}{1.3}
\setlength{\tabcolsep}{3pt}
\begin{tabular}{cccc||cccc}
\hline
$l$ & Elimination &  Behavior & $\epsilon_l$ & l & Elimination &  Behavior & $\epsilon_l$\\ \hline
1 & $(f_1,E)$  &  $(1,1,+)$ & +  & 6 & $(f_{10},P_1)$ & $(5,3,-)$ & $-$  \\
2 & $(f_2,F)$  & $(1,1,+)$ & +  & 7 & $(f_5,S_0)$ &  $(3,1,-)$ & $-$  \\
3 & $(f_7,X_1)$ &  $(5,3,+)$ & +  & 8 & $(f_6,P_0)$ & $(3,1,-)$ & $-$  \\ 
4 & $(f_8,X_2)$  &  $(5,3,+)$ & +  & 9 & $(f_{3},Y_1)$ & $(1,1, +)$ & + \\
5 & $(f_9,S_1)$  & $(5,3,-)$ & $-$ \\ \hline
\end{tabular}\end{center}
Note that the last elimination $f_3$ corresponds to $\varphi_1(Y_1,Y_2)-\oS=0$, which is increasing in $Y_1$.
Since $\epsilon=\prod_{l=1}^{9} \epsilon_l=1$, we conclude that $ \sign(\varphi'(z_{10}))= \sign (\det(J(S_f)(z)))$
for any equilibrium point $z$ and $z_{10}=Y_2$. 

\paragraph{Motif (k).}
Consider the variables in the following order
$$(x_1,x_2,x_3,x_4,x_5,x_6,x_7,x_8,x_9,x_{10})=(E,F,S_0,P_0,X_1,X_2,S_1,P_1,Y_1,Y_2).$$
The variables of the reduced system are $X_1,X_2,S_1,P_1,Y_1,Y_2$ and the variables eliminated using the conservation laws are $E,F,S_0,P_0$. The sign of the determinant of the Jacobian of the system $S_f$ is the same as the sign of the determinant of the Jacobian of the system
\begin{align*}
f_1(x)  &= E + X_1 - \oE  & f_5(x) &=  X_1 - \eta_1 E S_{0} &  f_8(x) &= X_2- \mu_2 Y_2  \\ f_2(x) &= F + Y_1 + Y_2 -\oF & f_6(x) &=  X_2 - \eta_2 S_1 P_0 & f_9(x) &= Y_1- \delta_1 F S_1     \\
f_3(x) &=  S_0 + S_1 + X_1+ Y_1 +X_2 - \oS  &  f_7(x) &=  X_1-\mu_1 Y_1  &  f_{10}(x)&= Y_2- \delta_2 F P_1.  \\  f_4(x) &=  P_0+P_1+X_2+Y_2-\oP  
\end{align*}
If $\bar{x}=(X_1,X_2,S_1,P_1,Y_1,Y_2)$ then
\begin{align*}
\psi_1(\bar{x}) & =- X_1 + \oE    & \psi_3(\bar{x}) & = -S_1 -X_1- Y_1-X_2+ \oS \\  \psi_2(\bar{x}) & = - Y_1 - Y_2 +\oF  & \psi_4(\bar{x}) & =-P_1-X_2-Y_2+\oP. 
\end{align*}
Here $n=10$ and $k=4$, such that $n-k=6$ is even. The function $\varphi$ is equation $f_4$ with all variables but $x_{10}=Y_2$ eliminated. The eliminations we performed in Section \ref{motifcascade} (Motif (k)) are summarized in the following table:

\begin{center}
\renewcommand{\arraystretch}{1.3}
\setlength{\tabcolsep}{3pt}
\begin{tabular}{cccc||cccc}
\hline
$l$ & Elimination &  Behavior & $\epsilon_l$ & l & Elimination &  Behavior & $\epsilon_l$\\ \hline
1 & $(f_1,E)$  &  $(1,1,+)$ & +  & 6 & $(f_{10},P_1)$ & $(5,3,-)$ &$-$  \\
2 & $(f_2,F)$  & $(1,1,+)$ & +  & 7 & $(f_5,S_0)$ &  $(3,1,-)$ & $-$  \\
3 & $(f_7,X_1)$ &  $(5,3,+)$ & +  & 8 & $(f_6,P_0)$ & $(3,1,-)$ & $-$  \\ 
4 & $(f_8,X_2)$  &  $(5,3,+)$ & +  & 9 & $(f_{3},Y_1)$ & $(1,1, +)$ & + \\
5 & $(f_9,S_1)$  & $(5,3,-)$ & $-$ \\ \hline
\end{tabular}\end{center}
Note that the last elimination $f_3$ corresponds to $\Phi(Y_1,Y_2)-\oS=0$.
Since $\epsilon=\prod_{l=1}^{9} \epsilon_l=1$, we conclude that $ \sign(\varphi'(z_{10}))= \sign (\det(J(S_f)(z)))$
for any equilibrium point $z$ and $z_{10}=Y_2$. 

\paragraph{Motif (l).}
Consider the following order of the variables
$$(x_1,x_2,x_3,x_4,x_5,x_6,x_7,x_8,x_9,x_{10},x_{11},x_{12})=(E,F_1,F_2,S_0,P_0,X_1,X_2,X_3,S_1,P_1,Y_1,Y_2).$$
The variables of the reduced system are $X_1,X_2,X_3,S_1,P_1,Y_1,Y_2$ and the variables eliminated using the conservation laws are $E,F_1,F_2,S_0,P_0$. The sign of the determinant of the Jacobian of the system $S_f$ is  \emph{opposite}  the sign of the determinant of the Jacobian of the system
\begin{align*}
f_1(x)  &= E + X_1 +X_3 - \oE &  f_7(x) &=  X_2 - \eta_2 S_1 P_0      \\
f_2(x) &= F_1 + Y_1  -\oF   & f_8(x)&= X_3-\eta_3 E P_0  \\ 
f_3(x)&= F_2+Y_2-\oF_2   & f_9(x) &= X_1-\mu_1 Y_1    \\
f_4(x) &=  S_0 + S_1 + X_1+ Y_1 +X_2 - \oS     & f_{10}(x) &= X_2- (\mu_2/\mu_3)X_3 - \mu_2Y_2    \\
 f_5(x) &= P_0+P_1+X_2+Y_2+X_3-\oP & f_{11}(x) &=  Y_1- \delta_1 F_1 S_1  \\ 
 f_6(x) &=  X_1 - \eta_1 E S_{0}&  f_{12}(x)&= Y_2- \delta_2 F_2 P_1. 
\end{align*}
Indeed, the reason for the change in sign comes from Remark \ref{signchanges}, since there is an odd number of equations that are rearranged by multiplication of negative numbers (corresponding to $\eta_1,\eta_2,\eta_3,\delta_1,\delta_2$).

Let $\bar{x}=(X_1,X_2,X_3,S_1,P_1,Y_1,Y_2)$. Then
\begin{align*}
\psi_1(\bar{x}) & =- X_1- X_3 + \oE    & \psi_4(\bar{x}) & = -S_1 -X_1- Y_1-X_2+ \oS 
\\  \psi_2(\bar{x}) & = - Y_1  +\oF_1  & \psi_5(\bar{x}) & =-P_1-X_2-Y_2-X_3+\oP \\ 
\psi_3(\bar{x}) & = - Y_2  +\oF_2. 
\end{align*}
Here $n=12$ and $k=5$, such that $n-k=7$ is odd. In this case, we should see that $\epsilon=1$. The function $\varphi$ is equation $f_4$ with all variables but $x_{11}=Y_1$ eliminated. The eliminations we performed in Section~\ref{motifcascade} (Motif (l)) are summarized in the following table:

\begin{center}
\renewcommand{\arraystretch}{1.3}
\setlength{\tabcolsep}{3pt}
\begin{tabular}{cccc||cccc}
\hline
$l$ & Elimination &  Behavior & $\epsilon_l$ & l & Elimination &  Behavior & $\epsilon_l$\\ \hline
1 & $(f_1,E)$  &  $(1,1,+)$ & +  & 7 & $(f_{8},X_3)$ & $(4,3,+)$ & $-$  \\
2 & $(f_2,F_1)$  & $(1,1,+)$ & +  & 8 & $(f_{10},Y_2)$ &  $(4,5,-)$ & $+$  \\
3 & $(f_3,F_2)$ &  $(1,1,+)$ & +  & 9 & $(f_{12},P_1)$ & $(4,3,-)$ & $+$  \\ 
4 & $(f_9,X_1)$  &  $(6,3,+)$ & $-$  & 10 & $(f_{6},S_0)$ & $(3,1, -)$ & $-$ \\
5 & $(f_{11},S_1)$  & $(7,5,-)$ & $-$ & 11 & $(f_{5},X_2)$ & $(2,1, +)$ & $-$ \\
6 & $(f_{7},P_0)$  & $(4,2,-)$ & $-$ \\ \hline
\end{tabular}\end{center}
The last elimination $f_5$ corresponds to $\varphi_P(Y_1,X_2)-\oS=0$, which is increasing in $X_2$.
Since $\epsilon=\prod_{l=1}^{11} \epsilon_l=1$, we conclude that $ \sign(\varphi'(z_{10}))= -\sign (\det(J(S_f)(z)))$ for any equilibrium point $z$ and $z_{11}=Y_1$. Together with the fact that $n-k$ is odd, the steady states for which $\varphi'<0$ are unstable.

\subsection{The monostationary motifs}
For Motifs (a)-(d), (e), (h) and (j), the above procedure is non-conclusive. We only show this for Motif (j), since for the other motifs we can prove  that the steady state is asymptotically stable using the Routh-Hurwitz criterion.

\paragraph{Motif (j).}
Consider the variables in the following order
$$(x_1,x_2,x_3,x_4,x_5,x_6,x_7,x_8,x_9,x_{10},x_{11})=(E,F_1,F_2,S_0,P_0,X_1,X_2,S_1,P_1,Y_1,Y_2).$$
The variables of the reduced system are $X_1,X_2,S_1,P_1,Y_1,Y_2$ and the variables eliminated using the conservation laws are $E,F_1,F_2,S_0,P_0$. The sign of the determinant of the Jacobian of the system $S_f$ is the same as the sign of the determinant of the Jacobian of the system
\begin{align*}
f_1(x)  &= E + X_1 - \oE  &  f_6(x) &=  X_1 - \eta_1 E S_{0}  &  f_9(x) &= X_2- \mu_2 Y_2 \\ 
f_2(x) &= F_1 + Y_1  -\oF_1 & f_7(x) &=  X_2 - \eta_2 S_1 P_0 & f_{10}(x) &= Y_1- \delta_1 F_1 S_1 \\
 f_3(x) &= F_2 + Y_2  -\oF_2   &  f_8(x) &= X_1-\mu_1 Y_1&  f_{11}(x)&= Y_2- \delta_2 F_2 P_1.  \\
f_4(x) &=  S_0 + S_1 + X_1+ Y_1 +X_2 - \oS  \\
f_5(x) &= P_0+P_1+X_2+Y_2-\oP 
\end{align*}
Here $n=11$ and $k=5$, such that $n-k=6$ is even. The function $\varphi$ is equation $f_5$ with all variables but $x_{11}=Y_2$ eliminated. The eliminations we performed in Section~\ref{motifcascade} (Motif (j)) are summarized in the following table:

\begin{center}
\renewcommand{\arraystretch}{1.3}
\setlength{\tabcolsep}{3pt}
\begin{tabular}{cccc||cccc}
\hline
$l$ & Elimination &  Behavior & $\epsilon_l$ & l & Elimination &  Behavior & $\epsilon_l$\\ \hline
1 & $(f_1,E)$  &  $(1,1,+)$ & +  & 6 & $(f_{6},S_0)$ & $(3,1,-)$ & $-$  \\
2 & $(f_2,F_1)$  & $(1,1,+)$ & +  & 7 & $(f_{10},S_1)$ &  $(4,2,-)$ & $-$  \\
3 & $(f_3,F_2)$ &  $(1,1,+)$ & +  & 8 & $(f_7,P_0)$ & $(3,1,-)$ & $-$  \\ 
4 & $(f_8,X_1)$  &  $(5,3,+)$ & +  & 9 & $(f_{11},P_1)$ & $(3,1, -)$ & $-$ \\
5 & $(f_9,X_2)$  & $(5,3,+)$ & +  & 10 & $(f_{4},Y_1)$ & $(1,1, +)$ & + \\ \hline
\end{tabular}\end{center}
Since $\epsilon=\prod_{l=1}^{10} \epsilon_l=1$, we conclude that $ \sign(\varphi'(z_{11}))= \sign (\det(J(S_f)(z)))$ for any equilibrium point $z$ and $z_{11}=Y_2$. Since $n-k$ is even, stable steady states have positive determinant, and since $\varphi$ is always increasing, nothing can be concluded.

\paragraph{Routh-Hurwitz.}
We have computationally checked that Motifs (a)-(d), (e) and (h) have asymptotically stable BMSSs. 
The Routh-Hurwitz criterion establishes when all  roots of a polynomial have real negative parts from a condition on the coefficients of the polynomial. Specifically, consider a real polynomial
$$p(z)= \alpha_0z^n + \alpha_{1}z^{n-1} + \dots + \alpha_{n-1} z + \alpha_{n}, $$
where we can assume that $\alpha_0>0$.
One constructs the following matrix:
$$
A=\left(\begin{array}{cccccc}
\alpha_1 & \alpha_3 & \alpha_5 & \dots & \dots & 0 \\
\alpha_0 & \alpha_2 & \alpha_4 & \alpha_6 & \dots & 0 \\
0 & \alpha_1 & \alpha_3 & \alpha_5 & \dots & 0 \\
0 & \alpha_0 & \alpha_2 & \alpha_4 & \dots &  \\
\vdots & \vdots & \vdots & \vdots &\vdots & \alpha_n
\end{array}\right).
$$
That is, the entry $(i,j)$ of the matrix is $\alpha_{i+2(j-i)}$ with the convention that if $i+2(j-i)>n$ or $i+2(j-i)<0$, then $\alpha_{i+2(j-i)}=0$. The criterion states that  if all leading principal minors of the matrix have positive sign, then all  roots of the polynomial $p(z)$ have negative real parts.

In our case, the polynomial to be analyzed is the characteristic polynomial of
the Jacobian of the reduced system at a steady state. For example, by eliminating the variables $E$ and $S_0$, the reduced system of Motif (b) is
\begin{align*}
\dot{X} &=  -(b_1 + c_1)X + a_1 (\oE-X-Y)(\oS-S_1-X-Y) \nonumber \\
\dot{Y} &= -(b_2 + c_2)Y + a_2  (\oE-X-Y) S_1 \nonumber \\
\dot{S_{1}} &=  c_1 X +b_2 Y - a_2 (\oE-X-Y) S_1.\nonumber
\end{align*}

The Jacobian can easily be computed using a program that handles symbolic computations (e.g. Mathematica$^{TM}$). For Motif (b) the Jacobian $J$ is
{\small $$ \left(\begin{array}{ccc}
-b_1 - c_1 - a_1(\oE+\oS-S_1-2X-2Y) & - a_1(\oE+\oS-S_1-2X-2Y)  & - a_1(\oE-X-Y) \\
-a_2 S_1 & -b_2-c_2- a_2S_1 & a_2(\oE-X-Y) \\
c_1+a_2S_1 & b_2+a_2S_1 &  -a_2(\oE-X-Y)
\end{array}\right)$$}
Note that for any BMSS, the terms $\oE+\oS-S_1-2X-2Y$ and $\oE-X-Y$ are positive.
The leading principal minors of the Routh-Hurwitz matrix are polynomials in the entries of $J$.
  The task of  computationally determining the sign of these minors is greatly simplified if  
  we substitute back $\oE=E+X+Y$, and $\oS=S_0+S_1+X+Y$ into $J$ and write:
$$ J= \left(\begin{array}{ccc}
-b_1 - c_1 - a_1(E+S_0) & - a_1(E+S_0)  & - a_1E \\
-a_2 S_1 & -b_2-c_2- a_2S_1 & a_2E \\
c_1+a_2S_1 & b_2+a_2S_1 &  -a_2E
\end{array}\right).$$
The  entries of the matrix $J$ are now in the set of variables $$\mathcal{V}=\{E,X,Y,S_0,S_1,a_1,a_2,b_1,b_2,c_1,c_2\}$$ and  all variables in $\mathcal{V}$ take positive values at any BMSS. Thus, each of the entries of $J$ has a fixed sign for any BMSS.
The coefficients of the characteristic polynomial are polynomials in $\mathcal{V}$ and  thus, the leading principal minors of the Routh-Hurwitz matrix are also polynomials in $\mathcal{V}$. If the coefficients of these polynomials are all positive, then we are guaranteed that they take positive values for any choice of positive rates and any set of positive concentrations. This computation can easily be implemented and checked with Mathematica$^{TM}$. 
If we used  the matrix $J$ involving the terms $\oE+\oS-S_1-2X-2Y$ and $\oE-X-Y$, then the corresponding polynomials take values in $\{\oE,\oS,X,Y,S_1,a_1,a_2,b_1,b_2,c_1,c_2\}$ but their coefficients are no longer positive. Positivity checking requires a convenient grouping of  terms. 

This procedure shows stability for all motifs but Motif (j) in which case some negative terms in the expansion of some of the minors occur. These are computationally difficult to handle and we have not been able to show that the minors all are positive. However, by random generation of values, it seems that the steady state is stable, though it remains to be proven.

\begin{small}


\section{Proofs}\label{B}

\begin{proof}[Proposition \ref{motif1f1}]
Figure~\ref{ffig} illustrates the different cases of the proposition. Recall that $\Gamma=\{Y_1\in \Gamma_1| f(Y_1)\in \Gamma_2 \}$ and $\Delta_1=\mu_1\oF_1-\oE$.
The derivative of $f$ is 
$$f'(y)=\frac{\eta_2 ( \delta_1 \oE \oF_1 - 2 \delta_1\mu_1\oF_1 y +(\delta_1-\eta_2)\mu_1 y^2) }{\mu_2( \delta_1(\oF_1-y)+\eta_2 y)^2} $$
The denominator  is always positive. The numerator is a degree two polynomial in $y$ which might have positive real roots and is positive for $y=0$. We  consider $y\leq \xi_1=\min(\oF_1,\oE/\mu)$. 

 The numerator evaluated at $y=\oF_1$ is  $\eta_2 (\delta_1\oE - (\delta_1+\eta_2)\mu_1 \oF_1)\oF_1$. It is positive if and only if
$$\Lambda =(1+\eta_2/\delta_1)\mu_1 \oF_1 - \oE \leq 0.$$
If this is the case, then $\xi_1=\oF_1$ and the numerator is bounded from below by $\eta_2\mu_1(\delta_1(\oF_1-y)^2 + \eta_2(\oF_1^2-y^2))> 0$. Hence both the numerator and $f'(y)$ are positive for all $y<\xi_1$.

Therefore, for $\Lambda\leq 0$  the function $f$ is increasing for all $y\in \Gamma_1=[0,\oF_1)$ (see Figure~\ref{ffig}(A)). Since $f(0)=0$, the condition $f(Y_1)\in \Gamma_2$ is equivalent to $Y_1<f^{-1}(\xi_2)$ and hence $\Gamma=[0,\min(\oF_1,f^{-1}(\xi_2))$. This proves (i).

To prove (ii), note that $f'(\oF_1)<0$ for $\Lambda>0$.  Further, $f'(\oE/\mu_1)<0$ if and only if $\oE(\delta_1-\eta_2)-\delta_1\mu_1\oF_1<0$. This is clearly the case when $\Delta_1\geq 0$. If $\Delta_1<0$, then using $\Lambda>0$, it follows that
$$ \oE(\delta_1-\eta_2)-\delta_1\mu_1\oF_1 < \eta_2 (-\oE+\mu_1\oF_1)<0.$$
In either case, $f'(\oE/\mu_1)<0$. 
It follows that if  $\Lambda>0$, then, whatever the sign of $\Delta_1$, we have $f'(\xi_1)<0$ and hence the numerator of $f'$ has exactly one positive real root $\alpha$ in $\Gamma_1$. 
Thus, $f$ is a function that increases up to $y=\alpha$, and  decreases  after $\alpha$ (see Figure~\ref{ffig}).
We have that $f(\oE/\mu_1)=0$ and hence,  if $\Delta_1>0$,  $\xi_1=\oE/\mu_1$,  $f(\xi_1)=0$, and the function $f$ starts and ends in zero in $\Gamma_1$. Oppositely,  if $\Delta_1\leq 0$ then $\xi_1=\oF_1$ and $f(\xi_1) = -\Delta_1/\mu_2>0$. These differences are depicted in Figure~\ref{ffig}(B-C).

The three cases (ii)(a)-(c) of the proposition follow from the determination of $\Gamma$. The idea is depicted in Figure~\ref{regions}.  For $f(Y_1)\in \Gamma_2$ we require that $f(Y_1)<\xi_2$. The maximal value of $f$ in $\Gamma_1$ is $f(\alpha)$ which is strictly smaller than $\oE/\mu_2$ because $\mu_1Y_1+\mu_2Y_2=\mu_1Y_1+\mu_2f(Y_1)=\oE$. Thus,  we have
\begin{itemize}
\item If $f(\alpha)< \oF_2$, then for all $Y_1\in \Gamma_1$, we have $f(Y_1)\leq f(\alpha)< \xi_2=\min(\oF_2,\oE_2/\mu_2)$. Thus, $\Gamma=[0,\xi_1)$, which corresponds to case (c) of the proposition.
\item If $f(\alpha)\geq \oF_2$, then the horizontal line $Y_2=\xi_2=\oF_2$ intersects the graph of $f$ over $[0,\oE/\mu_1)$ in two points corresponding to $Y_1$-values $\alpha_1\leq \alpha \leq \alpha_2<\oE/\mu_1$ ($\alpha_1=\alpha_2$ if $\xi_2=f(\alpha)$).  We also have that $\alpha_1<\xi_1$, since $\alpha_1\leq\alpha<\xi_1$.
\end{itemize}
In the latter case,  if $Y_1\in [0,\alpha_1)\cup (\alpha_2,\oE/\mu_1)$, then $f(Y_1)\in \Gamma_2$, while $f(Y_1)$ lies outside $\Gamma_2$ if $Y_1\in [\alpha_1,\alpha_2]$.   Therefore, if $f(\alpha)\geq \oF_2$, we have
  $$\Gamma=([0,\alpha_1)\cup (\alpha_2,\oE/\mu_1))\cap [0,\xi_1)=
  \begin{cases}
  [0,\alpha_1) & \textrm{if } \xi_1\leq \alpha_2 \\
  [0,\alpha_1)\cup (\alpha_2,\xi_1) & \textrm{if }  \xi_1 > \alpha_2.   
  \end{cases}$$
 The first case happens if and only if $\xi_1=\oF_1$ and $\oF_1\leq \alpha_2$. Since $\alpha\leq \oF_1$, this condition is equivalent to $\Delta_1\leq 0$ and $-\Delta_1/\mu_2=f(\oF_1)\geq f(\alpha_2)=\oF_2$, or $0\leq \oE-(\mu_1\oF_1+\mu_2\oF_2)$.  This proves $(a)$.
  
The second case is equivalent to either $\xi_1=\oE/\mu_1$ or $\xi_1=\oF_1>\alpha_2$. Proceeding as above, this is equivalent to $0<\mu_1\oF_1+\mu_2\oF_2-\oE$, which proves case (b).
\end{proof}

\begin{proof}[Proposition \ref{varphider}]  Assume that $\Lambda>0$, that is Proposition~\ref{motif1f1}(ii) applies. Let $\bar{\varphi}_1(Y_1)=\varphi_1(Y_1,f(Y_1))$ so that  
$\varphi(Y_1)=\bar{\varphi}_1(Y_1)+\varphi_2(f(Y_1))$. The derivative of $\varphi$ with respect to $Y_1$ is
$$\varphi'(Y_1)=\frac{\partial \bar{\varphi}_1}{\partial Y_1}(Y_1) + \frac{\partial  \varphi_2}{\partial Y_2}(f(Y_1))\frac{\partial f}{\partial Y_1}(Y_1)\quad \textrm{with }\quad \frac{\partial  \varphi_2}{\partial Y_2}=(1+\mu_2)+\frac{\oF_2}{\delta_2(\oF_2-Y_2)^2}.$$ 
The term  
$\frac{\partial  \varphi_2}{\partial Y_2}$ is the only one that depends on $\oF_2$.  Note in addition that this term is decreasing in $\oF_2$ for $Y_2<\oF_2$.

The variables $S_1, Y_1,X_1$ are all increasing functions of $Y_1$. As for $S_0=\frac{\mu_1 Y_1}{\eta_1(\oE -\mu_1 Y_1-\mu_2f(Y_1))}$, we have that the derivative of $\mu_1Y_1+\mu_2f(Y_1)$ is
\begin{equation}\label{df}
\frac{\delta_1 \eta_2(\oE\oF_1-\mu_1Y_1^2)+  \delta_1^2\mu_1(\oF_1-Y_1)^2 }{(\eta_2Y_1+\delta_1(\oF_1-Y_1))^2},
\end{equation}
which is positive for $Y_1<\oF_1,\oE/\mu_1$. Thus, $S_0$ is also increasing in $Y_1$ and hence $\frac{\partial  \overline{\varphi}_1}{\partial Y_1}(Y_1)>0$.

By Proposition \ref{motif1f1} (ii)(b) and the discussion following the proposition,  if $\oF_2=f(\alpha)$  there exist (many) values of $Y_1\in \Gamma''$ for which $\varphi'(Y_1)<0$.  Fix one such $Y_1$. Since $\varphi'$ is continuous in $\oF_2$, then there exists an $\epsilon>0$ for which $\varphi'(Y_1)<0$ for any $\oF_2\in (f(\alpha)-\epsilon,f(\alpha)+\epsilon)$. This ensures that for values of $\oF_2\geq f(\alpha)$ close  to $f(\alpha)$,  we  have multistationarity.

On the other hand, the term $\frac{\partial  \varphi_2}{\partial Y_2}(f(Y_1))$ is positive for any $Y_1\in \Gamma$. Therefore, if   $\varphi'$ vanishes for some $Y_1$, it must be that $Y_1> \alpha$, where $f$ is decreasing.  In the interval $I=(\alpha,\xi_1)$, $-\partial f/\partial Y_1$ is positive, bounded from above and independent of $\oF_2$. It follows from the expression for $\partial f/\partial Y_1$ given in the proof of Proposition~\ref{motif1f1}.   Further, for $Y_1\in I$, $\partial \overline{\varphi}_1/\partial Y_1$ is positive (stated after equation~\eqref{df}), bounded from below and independent of $\oF_2$. The latter two statements follow from Result~\ref{motif1a} and equation~\eqref{df} upon differentiation (where~\eqref{df} is used to differentiate $S_0$).

Finally, when $\oF_2$ tends to infinity, $\partial \varphi_2/\partial Y_2$ tends to zero uniformly for all $Y_1\in I$, since $Y_2=f(Y_1)\leq f(\alpha)$ is independent of $\oF_2$. Putting it all together, we find that for $\oF_2$ large
\begin{equation}\label{phi_dev}
\frac{\partial  \overline{\varphi}_1}{\partial Y_1}(Y_1) > -\frac{\partial  \varphi_2}{\partial Y_2}(f(Y_1))\frac{\partial f}{\partial Y_1}(Y_1)
\end{equation}
for all $Y_1\in I$. Consequently the function $\varphi$ is increasing in all $\Gamma$ and there cannot be multistationarity. 

Therefore,  for values of $\oF_2$ above $f(\alpha)$  but `close' to $f(\alpha)$ there is multistationarity, and for large  values of $\oF_2$,  there is monostationarity. If $\oF_2'$ is such that $\varphi(Y_1)$ is increasing for some value $Y_1$, then for any $\oF_2\geq \oF_2'$, $\varphi(Y_1)$ must be increasing too (as $\varphi'(Y_1)$ continues being positive when increasing $\oF_2$, cf.~\eqref{phi_dev}). 

This implies that the two regions (mono- versus multistationarity) are separated by a certain boundary value of $\oF_2$, say $M$, where $M$ is the infimum of all $\oF_2$ such that $\varphi'(Y_1)>0$ for all $Y_1\in\Gamma$. For $\oF_2=M$, $\varphi$ is also strictly increasing: The derivative $\varphi'$ might be zero (as $M$ is the infimum) but this can only happen in a finite number of points because the derivative is a non-zero rational function.
Therefore, we see that for all $\oF_2\geq M$, there is only one BMSS, while for all $M> \oF_2>f(\alpha)$ there is multistationarity. 
\end{proof}

\begin{proof}[Proposition \ref{motif1i_prop}]
Recall that $\sigma=(\mu_1-\mu_2)( \mu_1\delta_1\eta_2 - \mu_2\delta_2\eta_1)$  and $\psi=(\varphi_1,\varphi_2)$. 

Let us compute $\det(J_{\psi}(Y_1,Y_2))$. We have, writing $C_F=\oF-Y_1-Y_2$ and $C_E=\oE-\mu_1Y_1-\mu_2Y_2$,
\begin{align*}
\frac{\partial \varphi_1}{\partial Y_1}  &= 1+\mu_1 + \frac{\oF }{\delta_1C_F^2}- \frac{Y_2}{\delta_1C_F^2} +\frac{\mu_1\oE}{\eta_1C_E^2} - \frac{\mu_1\mu_2 Y_2}{\eta_1C_E^2}  & \quad \frac{\partial \varphi_1}{\partial Y_2} =  \frac{Y_1}{\delta_1C_F^2} + \frac{\mu_1\mu_2 Y_1}{\eta_1C_E^2} \\
\frac{\partial \varphi_2}{\partial Y_2} &=1+\mu_2 + \frac{\oF}{\delta_2C_F^2}-\frac{Y_1}{\delta_2C_F^2} +\frac{\mu_2\oE}{\eta_2C_E^2}-\frac{\mu_2\mu_1 Y_1}{\eta_2C_E^2} & \quad \frac{\partial \varphi_2}{\partial Y_1}  =  \frac{Y_2}{\delta_2C_F^2} + \frac{\mu_1\mu_2 Y_2}{\eta_2C_E^2}  
\end{align*}
Let $C=\frac{\partial \varphi_2}{\partial Y_2}-1-\mu_2>0$ and $D=\frac{\partial \varphi_1}{\partial Y_1}-1-\mu_1>0$.
Then, the determinant is given as
\begin{align*}
\det(J_{\psi}(Y_1,Y_2))&= (1+\mu_1)(1+\mu_2)+(1+\mu_1)C + (1+\mu_2)D +\frac{\oF}{\delta_1\delta_2 C_F^3} + \frac{\mu_1\mu_2\oE}{\eta_1\eta_2C_E^3} + A+B
\end{align*}
with
\begin{align*}
A &= \frac{\mu_2(\oE\ \oF - \mu_1Y_1\oF - Y_2\oE)}{\delta_1\eta_2 C_E^2C_F^2} =  \frac{\mu_2}{\delta_1\eta_2C_EC_F}\Big(1 + \frac{Y_1}{C_F} + \frac{\mu_2Y_2}{C_E} + \frac{(\mu_2-\mu_1)Y_1Y_2}{C_EC_F}\Big) \\
B &= \frac{\mu_1(\oE\ \oF - \mu_2Y_2\oF - Y_1\oE)}{\delta_2\eta_1 C_E^2C_F^2} = \frac{\mu_1}{\delta_2\eta_1C_EC_F}\Big(1 + \frac{Y_2}{C_F} + \frac{\mu_1Y_1}{C_E} + \frac{(\mu_1-\mu_2)Y_1Y_2}{C_EC_F}\Big)
\end{align*}
where in the last equality of both equations, we substitute $\oF$ by $C_F+Y_1+Y_2$ and $\oE$ by $C_E+\mu_1Y_1+\mu_2Y_2$. Hence, $\det(J_{\psi}(Y_1,Y_2))$ is a sum of positive terms together with
$$N=\frac{Y_1Y_2}{\delta_1\delta_2\eta_1 \eta_2C_E^2C_F^2} (\mu_2\delta_2\eta_1(\mu_2-\mu_1) +
  \mu_1\delta_1\eta_2(\mu_1-\mu_2))$$
Therefore, if
$$ \sigma=(\mu_1-\mu_2)( \mu_1\delta_1\eta_2 - \mu_2\delta_2\eta_1)\geq 0   $$
then $\det(J_{\psi}(Y_1,Y_2))>0$ and the function $\varphi$ is always increasing. This proves (i).

Let us assume now that $\sigma<0$. 
We make the following observation: the only negative term $N$ of $\det(J_{\psi}(Y_1,Y_2))$    is also the term with denominator of highest degree in $C_E,C_F$. Further, all numerators in the expression of $\det(J_{\psi}(Y_1,Y_2))$ above are bounded. Thus, by letting $C_E,C_F$ simultaneously tend to zero,  we can obtain a negative determinant. 

Since $C_F=\oF-Y_1-Y_2$ and $C_E=\oE-\mu_1Y_1-\mu_2Y_2$, it is possible to make them small simultaneously by varying $Y_1,Y_2$, if and only if the two lines $r_F\colon \oF-Y_1-Y_2=0$ and $r_E\colon \oE-\mu_1Y_1 -\mu_2Y_2=0$ intersect for valid values of $Y_1,Y_2$. The intersection of the two lines is the point $(Y_1,Y_2)$ with
$$ Y_1 =\frac{\mu_2\oF - \oE}{\mu_2-\mu_1},\quad Y_2 =\frac{\mu_1\oF - \oE}{\mu_1-\mu_2}.  $$
If $\mu_1-\mu_2>0$, then for positive intersection values $Y_1,Y_2$ we require $\mu_2\oF < \oE < \mu_1 \oF$. 
If $\mu_2-\mu_1>0$, then we require $\mu_1\oF < \oE < \mu_2 \oF$. 

Hence, if these conditions above are satisfied, we are guaranteed the existence of values of $Y_1,Y_2$ for which the determinant is negative and thus multistationarity occurs. Further, if $C_E,C_F$ are small then $\oS,\oP$ must be large. 

Assume finally that $\sigma<0$ but  also 
one of the two following conditions holds:
\begin{enumerate}[(a)]
\item $\mu_1-\mu_2>0$ and either $\mu_2\oF \geq \oE $ or $\oE \geq \mu_1 \oF$.
\item $\mu_2-\mu_1>0$ and either $\mu_1\oF \geq \oE$ or $\oE \geq \mu_2 \oF$. 
\end{enumerate}
In this case multistationarity cannot occur. Assume that (a) holds (case (b) is similar).  Since  $\mu_1-\mu_2>0$, we have $\mu_1\oF > \mu_2\oF$ and:
\begin{itemize}
\item If $\mu_1\oF > \mu_2\oF \geq \oE $, then using $-Y_k\oE\geq-\mu_kY_k\oF$, we have
$$\oE\oF - \mu_1Y_1\oF - Y_2\oE  \geq \oE\oF - \mu_1Y_1\oF - \mu_2Y_2\oF =\oF C_E > 0$$
and 
$$\oE\oF - Y_1\oE- \mu_2Y_2\oF   \geq \oE\oF - \mu_1Y_1\oF - \mu_2Y_2\oF =\oF C_E > 0. $$
\item If $\oE\geq \mu_1\oF > \mu_2\oF $, then using $-\mu_kY_k\oF \geq -Y_k\oE$, we have
$$\oE\oF - \mu_1Y_1\oF - Y_2\oE  \geq \oE\oF - Y_1\oE - Y_2\oE  =\oE C_F>0$$
and
$$\oE\oF  - Y_1\oE - \mu_2Y_2\oF \geq \oE\oF - Y_1\oE - Y_2\oE  =\oE C_F>0.$$
\end{itemize}
Therefore, both $A$ and $B$ are positive and the proposition is proved.
\end{proof}

\begin{proof}[Proposition \ref{nomulticascade}]
Recall that $f$ is a function of $Y_2$ defined in equation \eqref{fprop4} by elimination of $Y_1$ from $\oS=\Phi(Y_1,Y_2)$. Thus, we have  that $f'=-\Phi_2/\Phi_1$, where
\begin{align*}
\Phi_2: = \partial \Phi/\partial Y_2 &=  \mu_2 + \frac{Y_1}{\delta_1(\oF-Y_1-Y_2)^2}>0,  \\
\Phi_1:= \partial \Phi/\partial Y_1 &= 1+\mu_1 + \frac{\mu_1\oE}{\eta_1(\oE-\mu_1Y_1)^2}+ \frac{\oF-Y_2}{\delta_1(\oF-Y_1-Y_2)^2}>0.
\end{align*}
Recall that
$$\frac{\partial P_0}{\partial Y_2}= \frac{\mu_2\delta_1}{\eta_2f^2} (f(\oF - f - 2Y_2) - f'Y_2(\oF-Y_2)), \quad \textrm{and}\quad 
\frac{\partial P_1}{\partial Y_2}= \frac{\oF - f + f'Y_2}{\delta_2 (\oF-f-Y_2)^2}.$$
$\frac{\partial P_1}{\partial Y_2}<0$ only if $\oF - f + f'Y_2<0$, that is
\begin{equation}\label{help}
0  > \Phi_1(\oF - Y_1 ) - \Phi_2 Y_2 = (1+\mu_1)(\oF-Y_1)-\mu_2Y_2 + \frac{\mu_1\oE(\oF-Y_1)}{\eta_1(\oE-\mu_1Y_1)^2} + \frac{\oF}{\delta_1(\oF-Y_1-Y_2)}.
\end{equation}
Since the last two summands are positive for valid values of $Y_1,Y_2$, a necessary condition for $\frac{\partial P_1}{\partial Y_2}<0$ is $C1: (1+\mu_1)(\oF-Y_1)-\mu_2Y_2 <0$. Note that if $1+\mu_1>\mu_2$, then the condition cannot be fulfilled and $P_1$ is always an increasing function of $Y_2$ for any $\oE,\oF,\oS$. 

Since $f'<0$,  we require $C2: \oF - Y_1 - 2Y_2<0$ in order to have  $\frac{\partial P_0}{\partial Y_2}<0$. 

It follows that $C1,C2$ are necessary conditions for multistationarity. If $\oP,\oS$ are given, the following are restrictions on  $Y_1,Y_2$:
$$Y_1+Y_2<\oF,\quad \mu_1Y_1<\oE, \quad (1+\mu_1)Y_1 + \mu_2Y_2<\oS,\quad \mu_2 Y_2<\oP. $$ 
Let $\overline{\sigma }= \min(\oS,\oP)$. We find that $\mu_2Y_2<\overline{\sigma}$ and hence 
$$ \oF - Y_1 - 2Y_2 > \oF - \oE/\mu_1 - 2 \overline{\sigma}/\mu_2\quad\textrm{and}\quad (1+\mu_1)(\oF-Y_1)-\mu_2Y_2 >  (1+\mu_1)(\oF -\oE/\mu_1) -\overline{\sigma}.  $$ 
Let $\sigma_{\mu} = \min (1+\mu_1,\mu_2/2)$. If  $\oF > \oE/\mu_1$ and $\sigma_{\mu}(\oF - \oE/\mu_1) > \overline{\sigma}$, then we are guaranteed that  both $C1,C2>0$ and multistationarity cannot occur.

\end{proof}

\begin{proof}[Proposition \ref{varphinf}] 
Note that for $\Phi$ (i.e. $\oS$) to tend to $+\infty$ for any fixed $Y_2\in \Gamma_2=[0,\min(\oS/\mu_2,\oF))\subseteq [0,\oF)$, either $Y_1$ must tend to $\oE/\mu_1$  or to $\oF-Y_2$.   The first case arises if and only if $Y_2\leq \oF-\oE/\mu_1$. 

(i) Assume that  $Y_2>\oF-\oE/\mu_1$. Consider the functions $\Phi_2,\Phi_1$ from the proof of Proposition \ref{nomulticascade} and the expression for $\frac{\partial P_0}{\partial Y_2}$ and $\frac{\partial P_1}{\partial Y_2}$: The derivative $f'=-\Phi_2/\Phi_1$ tends to $-1$ as $\oS$ tends to infinity (and $Y_1$ tends to $\oF-Y_2$).  Likewise, it follows that  $f(\oF - f - 2Y_2) - f'Y_2(\oF-Y_2)$ tends to zero, and hence $\frac{\partial P_0}{\partial Y_2}$ also tends to zero. From equation~\eqref{help} in the previous proof  it follows that $\frac{\partial P_1}{\partial Y_2}>0$ as $Y_1$ approaches $\oF-Y_2$, that is as $\oS$ becomes large (the last term is positive and dominates the other terms). Hence the derivative of $\varphi(Y_2)$ becomes positive: $\varphi'(Y_2)=1+\mu_2+\frac{\partial P_0}{\partial Y_2}+\frac{\partial P_1}{\partial Y_2}>0$. This proves (i).

(ii) Assume that  $\Delta_1=\mu_1\oF-\oE>0$ and consider $Y_2\leq \oF-\oE/\mu_1$. As $\oS$  tends to infinity, $Y_1$ tends to $\oE/\mu_1$. The limit curve $\varphi_{\infty}$ is
$$\varphi_{\infty}(Y_2) = (1+\mu_2)Y_2 + \frac{\mu_2\delta_1(\overline{\Delta}_1-Y_2) Y_2}{\eta_2 \oE/\mu_1}+\frac{Y_2}{\delta_2 (\overline{\Delta}_1-Y_2)}$$
for $\overline{\Delta}_1=\Delta_1/\mu_1$.
Existence of   values of $Y_2$ for which $\varphi_{\infty}'(Y_2)<0$ is equivalent to
$$p(Y_2):= \delta_2 (\overline{\Delta}_1-Y_2)^2(1+\mu_2 +\beta (\overline{\Delta}_1-2Y_2))+\overline{\Delta}_1<0  $$  
with $\beta=\mu_1\mu_2\delta_1/\eta_2\oE$. The function $p(Y_2)$ is a polynomial of degree $3$. For
$Y_2=0$, it is  positive and when $Y_2$ tends to infinity, the polynomial tends to $-\infty$. Therefore, it takes negative values in $Y_2\in [0,\overline{\Delta}_1)$ if and only if there is a  root in this interval. At $Y_2=\overline{\Delta}_1$ it is positive, and hence, there is at least one root after $\overline{\Delta}_1$. The derivative with respect to $Y_2$ is
$$p'(Y_2)=-2 \delta_2(\overline{\Delta}_1-Y_2)(1+\mu_2 +\beta (2\overline{\Delta}_1-3Y_2)).$$ 
One zero of the derivative is $Y_2=\overline{\Delta}_1$, while the other is $\gamma=(1+\mu_2+2\beta\overline{\Delta}_1)/3\beta$. Note that at $\overline{\Delta}_1$, $p(\overline{\Delta}_1)$ is positive. Therefore, for negative values of $p$ to occur between $0$ and $\overline{\Delta}_1$, we need $\gamma<\overline{\Delta}_1$ and $p(\gamma)<0$. 

We have $\gamma<\overline{\Delta}_1$ if and only if $\overline{\Delta}_1>(1+\mu_2)/\beta$, in which case $\overline{\Delta}_1$ is a maximum and $\gamma$ is a minimum. Evaluating $p$ in $\gamma$ gives the following condition:
$$\Delta_2=27\mu_1\mu_2^2\delta_1^2\eta_2\oE(\mu_1\oF-\oE) -\delta_2(\delta_1\mu_2\mu_1\oF-(\mu_2\delta_1+\eta_2(1+\mu_2))\oE)^3<0,$$ 
which concludes the proof of (ii).

\end{proof}

\begin{proof}[Proposition \ref{motifllemma}]
It is convenient to write $\varphi$ as a function of $X_2,Y_1$, i.e., without substituting $X_2$ by $f(Y_1)$.
The derivative of $\varphi$ with respect to $Y_1$ takes the form 
$$\varphi' =   \frac{\mu_1(\eta_2\oE + \delta_1\eta_3 X_2 (\mu_1\oF_1-\oE))}{\eta_1\eta_2(\oE - \mu_1Y_1)^2} + 1+\mu_1 + \frac{\oF_1}{\delta_1(\oF_1-Y_1)^2}+ \left(\frac{\mu_1\delta_1\eta_3(\oF_1-Y_1)}{\eta_1\eta_2(\oE-\mu_1Y_1)}+1\right)f' $$ 
with only the first term susceptible of being negative since $f$ is an increasing function. If we can control $f'$ by varying some total amounts we might be able to make the derivative $\varphi'$ negative.

We will see that there exist values of $X_2, \oF_2$, and $\oE$ for which $\varphi'$ as a function of the two variables $Y_1$ and $X_2$ is negative.  Then, by defining $\oP=\varphi_P(Y_1,X_2)$, we obtain regions of multistationarity. Note that this is equivalent to the approach taken in Motif (i), when we studied the Jacobian of $(\varphi_S,\varphi_P)$ rather than the derivative of $\varphi$.

Using the implicit function theorem, the derivative of $f$ with respect to $Y_1$ is  $f'=-\frac{\partial \varphi_P/\partial Y_1}{\partial \varphi_P/\partial X_2}$. We compute the partial derivatives of $P_0,P_1,X_2,Y_2,X_3$ for each term in the numerator and the denominator. First of all we repeat some definitions,
$$X_3=\Phi_X(Y_1,X_2)=\frac{\delta_1\eta_3 X_2(\oF_1-Y_1)(\oE-\mu_1Y_1)}{\eta_2Y_1+\delta_1\eta_3X_2(\oF_1-Y_1)},$$
and
$$Y_2=\Phi_Y(Y_1,X_2)=\mu_2^{-1}X_2+\mu_3^{-1}\Phi_X.$$
It follows that
$$\frac{\partial \varphi_P}{\partial X_2}  =\frac{\partial (P_0+P_1)}{\partial X_2}   + 1+ \mu_2^{-1} + (1+\mu_3^{-1}) \frac{\partial \Phi_X}{\partial X_2}, \qquad
\frac{\partial \varphi_P}{\partial Y_1} =\frac{\partial (P_0+P_1)}{\partial Y_1}  +   (1+\mu_3^{-1}) \frac{\partial \Phi_X}{\partial Y_1}.$$
For the denominator we have:
\begin{align*}
\frac{\partial P_0}{\partial X_2} &= \frac{\delta_1(\oF_1-Y_1)}{\eta_2 Y_1},  \qquad &
 \frac{\partial \Phi_X}{\partial X_2} = \frac{\delta_1\eta_3\eta_2Y_1(\oF_1-Y_1)(\oE-\mu_1Y_1)}{(\eta_2Y_1 + \delta_1\eta_3X_2(\oF_1-Y_1))^2}, \\
\frac{\partial P_1}{\partial X_2} &=  \frac{\oF_2}{\delta_2(\oF_2- \Phi_Y)^2}\cdot (\mu_2^{-1}+\mu_3^{-1} \frac{\partial \Phi_X}{\partial X_2}).
\end{align*}

For the numerator, we have
\begin{align*}
-\frac{\partial P_0}{\partial Y_1} &= \frac{\delta_1\oF_1 X_2 }{\eta_2 Y_1^2} & \hspace{-0.4cm}
- \frac{\partial \Phi_X}{\partial Y_1} &=  \frac{\delta_1\eta_3X_2\mu_1(\oF_1-Y_1)}{\eta_2Y_1 + \delta_1\eta_3X_2(\oF_1-Y_1)}+  \frac{\eta_2\oF_1(\oE-\mu_1Y_1)} {(\eta_2Y_1 + \delta_1\eta_3X_2(\oF_1-Y_1))^2}.  \\  -\frac{\partial P_1}{\partial Y_1} &=  \frac{-\oF_2}{\delta_2\mu_3(\oF_2- \Phi_Y)^2} \frac{\partial \Phi_X}{\partial Y_1} 
\end{align*}

Now fix $Y_1, \oF_1$ and let $X_2,\oF_2,\oE$ tend to $+\infty$. Assume that $\oE$ goes slower than $X_2$, for example put $X_2=\oE^{1+p}$ for some $p>0$. Let $\oF_2-\Phi_Y=K$ be constant such that $\oF_2= K+\mu_2^{-1}X_2+\mu_3^{-1}\Phi_X$ and $X_2$ go at the same rate towards infinity. Note that we can vary $X_2,\oF_2,\oE$ without changing $Y_1, \oF_1$ or restricting their range of definition.  

With these assumptions  the numerator of $f'$, $-\frac{\partial \phi_P}{\partial Y_1} $, goes to infinity proportionally to $X_2$. Also the denominator of $f'$, $\frac{\partial \phi_P}{\partial X_2} $, goes to infinity  proportionally to $X_2$. Consequently, $f'$ converges towards a constant that depends on $\oF_1, Y_1$ and $K$ (in addition to the rate constants)
$$f'\rightarrow K^2\frac{\delta_1\delta_2\mu_2^2\oF_1}{\eta_2Y_1^2}+\frac{\mu_1\mu_2}{\mu_3}.$$
Going back to the expression for $\varphi'$, we find that under the same assumptions,
$$\varphi'\approx \frac{-\mu_1\delta_1\eta_3 X_2}{\eta_1\eta_2\oE} + 1+\mu_1 + \frac{\oF_1}{\delta_1(\oF_1-Y_1)^2}+ f',$$
where terms that eventually vanish are not shown. It follows that for large $\oE$,  $X_2=\oE^{1+p}$ and $\oF_2=K+\mu_2^{-1}X_2+\mu_3^{-1}\Phi_X$, the derivative $\varphi'$ eventually becomes negative for any choice of $Y_1$ and $\oF_1$, as desired.

\end{proof}
\end{small}

\end{document}